\documentclass[11pt]{article}
\usepackage{amsfonts, graphicx,color,multirow, amsthm,amsmath,natbib,soul,xcolor}
\setcitestyle{round}
\usepackage{setspace,url}
\usepackage{multirow}
\usepackage{bm}
\usepackage{bbm}
\usepackage{rotating}
\usepackage{pdflscape}

\usepackage{titlesec}

\def\boxit#1{\vbox{\hrule\hbox{\vrule\kern6pt\vbox{\kern6pt#1\kern6pt}\kern6pt\vrule}\hrule}}

\usepackage{mathtools}

\DeclarePairedDelimiterX{\inp}[2]{\langle}{\rangle}{#1, #2}

\usepackage{enumitem}
\usepackage[margin=1in]{geometry}

\begin{document}

\newcommand{\argmax}[0]{\mbox{argmax}}
\newcommand{\argmin}[0]{\mbox{argmin}}
\newcommand{\ol}[1]{\overline{#1}}

\theoremstyle{plain}
\newtheorem{thm}{Theorem}
\newtheorem{lem}{Lemma}
\newtheorem{prop}[thm]{Proposition}
\newtheorem{cor}[thm]{Corollary}

\theoremstyle{definition}
\newtheorem{defn}{Definition}[section]
\newtheorem{conj}{Conjecture}[section]
\newtheorem{exmp}{Example}[section]
\newtheorem{condition}{Condition}[section]

\theoremstyle{remark}
\newtheorem{rmk}{Remark}
\newtheorem{note}{Note}
\newtheorem{case}{Case}

\newcommand{\var}{\mbox{var}}
\newcommand{\Var}{\mbox{Var}}
\newcommand{\bb}{\mbox{\bf b}}
\newcommand{\bff}{\mbox{\bf f}}
\newcommand{\bx}{\mbox{\bf x}}
\newcommand{\by}{\mbox{\bf y}}
\newcommand{\bg}{\mbox{\bf g}}
\newcommand{\bA}{\mbox{\bf A}}
\newcommand{\ba}{\mbox{\bf a}}
\newcommand{\bB}{\mbox{\bf B}}
\newcommand{\bC}{\mbox{\bf C}}
\newcommand{\bD}{\mbox{\bf D}}
\newcommand{\bE}{\mbox{\bf E}}
\newcommand{\bF}{\mbox{\bf F}}
\newcommand{\bG}{\mbox{\bf G}}
\newcommand{\bH}{\mbox{\bf H}}
\newcommand{\bI}{\mbox{\bf I}}
\newcommand{\bU}{\mbox{\bf U}}
\newcommand{\bV}{\mbox{\bf V}}
\newcommand{\bQ}{\mbox{\bf Q}}
\newcommand{\bR}{\mbox{\bf R}}
\newcommand{\bS}{\mbox{\bf S}}
\newcommand{\bX}{\mbox{\bf X}}
\newcommand{\bY}{\mbox{\bf Y}}
\newcommand{\bZ}{\mbox{\bf Z}}
\newcommand{\br}{\mbox{\bf r}}
\newcommand{\bv}{\mbox{\bf v}}
\newcommand{\bL}{\mbox{\bf L}}
\newcommand{\bbP}{\mathbb{P} }
\newcommand{\bone}{\mbox{\bf 1}}
\newcommand{\bsone}{\mbox{\scriptsize \bf 1}}
\newcommand{\bzero}{\mbox{\bf 0}}
\newcommand{\bveps}{\mbox{\boldmath $\varepsilon$}}
\newcommand{\bet}{\mbox{\boldmath $\eta$}}
\newcommand{\bxi}{\mbox{\boldmath $\xi$}}
\newcommand{\bzeta}{\mbox{\boldmath $\zeta$}}
\newcommand{\beps}{\mbox{\boldmath $\varepsilon$}}
\newcommand{\bbeta}{\mbox{\boldmath $\beta$}}
\newcommand{\balpha}{\mbox{\boldmath $\alpha$}}
\newcommand{\bPsi}{\mbox{\boldmath $\Psi$}}
\newcommand{\bomega}{\mbox{\boldmath $\omega$}}
\newcommand{\hbbeta}{\hat{\bbeta}}
\newcommand{\hbeta}{\hat{\beta}}
\newcommand{\Bbeta}{ \bar{\beta}}
\newcommand{\bmu}{\mbox{\boldmath $\mu$}}
\newcommand{\bgamma}{\mbox{\boldmath $\gamma$}}
\newcommand{\mv}{\mbox{V}}
\newcommand{\bSigma}{\mbox{\boldmath $\Sigma$}}
\newcommand{\cov}{\mbox{cov}}
\newcommand{\beq}{\begin{eqnarray}}
\newcommand{\eeq}{\end{eqnarray}}
\newcommand{\beqn}{\begin{eqnarray*}}
\newcommand{\eeqn}{\end{eqnarray*}}
\newcommand{\Mstar}{{\mathcal{M}_{\star}}}
\newcommand{\lammax}{{\lambda_{\max}}}
\newcommand{\A}{{\bbP_n \bPsi_jY}}
\newcommand{\B}{{(\bbP_n \bPsi_j\bPsi_j^T)^{-1}}}
\newcommand{\ea}{{E \bPsi_jY}}
\newcommand{\eb}{{(E \bPsi_j\bPsi_j^T)^{-1}}}


\title{\large{\bf Boosting High Dimensional Predictive Regressions with Time Varying Parameters}}

\author{\textsc{Kashif Yousuf}
\thanks{ Department of Statistics, Columbia University, New York, NY 10027 (Email: ky2304@columbia.edu)}
\and \textsc{Serena Ng}
\thanks{(Corresponding author) Department of Economics and NBER, Columbia University, 420 W. 118 St. MC 3308, New York, NY 10027.(Email: sn2294@columbia.edu) 
\newline
Financial support from the National Science Foundation (SES-1558623) is gratefully acknowledged.
}
}
\clearpage
\date{\today}
\maketitle
\thispagestyle{empty}

\begin{abstract}
\begin{onehalfspace}
High dimensional predictive regressions are useful in wide range of applications.  However, the theory is mainly developed assuming that the model is stationary with time invariant parameters. This is at odds with the prevalent evidence for parameter instability in economic time series, but theories for parameter instability are mainly developed for models with a small number of covariates. In this paper, we present two $L_2$ boosting algorithms for estimating high dimensional models in which the coefficients are modeled as functions evolving smoothly over time and the predictors are locally stationary.  The first method uses componentwise local constant estimators as base learner, while the second relies on componentwise local linear estimators. We establish consistency of both methods, and address the practical issues of choosing the bandwidth for the base learners and the number of boosting iterations. In an extensive application to macroeconomic forecasting with many potential predictors, we find that the benefits to modeling time variation are substantial and they increase with the forecast horizon. Furthermore, the timing of the benefits suggests that the Great Moderation is associated with substantial instability in the conditional mean of various economic series.

\end{onehalfspace}
\end{abstract}


\clearpage
\setcounter{page}{1}

\section{Introduction} \label{sec:introduction}

Due to the rapid improvements in the information technology, high dimensional time series datasets are frequently encountered in a variety of fields in economics and finance (see \cite{fan2011sparse,shapiro_2017} for examples). In these settings, the number of candidate predictors ($p_T$) is much larger than the number of samples ($T$), and accurate estimation and prediction is made possible by relying on some form of dimension reduction. \citet{ng-handbook} puts the methods used in high dimension predictive regressions into two classes: a dense class which assumes that the covariates have a low rank representation that can be exploited for subsequent modeling, and a sparse class which assumes that  the number of relevant predictors is far smaller than the number of predictors available. Research within the first class usually assumes a linear latent factor model which is estimated by principal components or partial least squares.\footnote{\cite{SW2003,bai2002determining} and \cite{Kelly2015}.} The second class treats the problem as one of variable selection in high dimension. Prominent methods in this class include screening, penalized likelihood, lasso, and boosting methods. 

This paper contributes to the literature in the second class. A key assumption made in the vast majority of works on sparse modeling is of a stationary underlying model with time invariant parameters.\footnote{Examples include \cite{MM2016}, \cite{Kock2015}, \cite{han2017high}, and \cite{basu2015} which focus on the Lasso or the adaptive Lasso, and \cite{lutz2006} which focuses on $L_2$ boosting for stationary VAR models.} The assumption is very restrictive in practice, as empirical evidence of parameter instability and time varying effects have been well documented in macroeconomics.\footnote{See \citep{Stock96,rossi2013advances,hamilton1989new}, asset pricing \citep{Welch2003,paye2006instability,rapach2010out, dangl2012predictive}, and exchange rate prediction \citep{schinasi1989}.} Parameter instability can be driven by structural changes in technological advancements, government or monetary policy changes, and preference shifts at the individual level \citep{Bin12}. Ignoring these instabilities can lead to large forecasting errors, with \cite{Clements96} and others even arguing that these instabilities are the main source of error for forecasting models. 

Consider a high dimensional linear time varying parameter (TVP) model:
\begin{equation} Y_t= \bm{\beta}_{t}\bm{x}_{t-h} +\epsilon_t  \textrm{ for } t=1,\ldots,T, \label{htvp}\end{equation} 
where $Y_t$ is the response,
$\bm{x}_{t-h}=(X_{1,t-h},\ldots,X_{p_T,t-h})$ is a $p_T$-dimensional vector of predictors (with $p_T >>T$), $\bm{\beta}=(\beta_{1,t},\ldots,\beta_{p_T,t})$ is a vector of time varying parameters, and $\epsilon_t$ are errors; the precise assumptions on the model will be stated in section \ref{methodology}. Given the evidence for parameter instability, the question remains on how to best represent and model this change, especially when dealing with high dimensional predictors. Parameter instability is most commonly represented in the econometrics literature by random walks or by one or more discrete structural breaks.\footnote{The first approach has a long history in macroeconomics, some examples include \cite{cogley2001evolving,primiceri2005time,koop2013large}. For the literature on structural breaks, see \cite{perron2006dealing,casini2018structural} for surveys.} Modeling variations by random walks can be quite restrictive as it imposes a specific structure on the evolution of the parameters. Discrete breaks  require knowledge of the break dates, and not all time variations are well characterized by discrete shifts. Technology and taste shifts are arguably evolving slowly over time. Smooth transition models as in \cite{terasvirta:94} are still tightly parameterized. Furthemore, these methods are mainly designed for a fixed $p_T$.  A third approach is to use rolling-window estimation to capture the  smooth change in the parameters. As will soon be clear, rolling-window estimation is a special case of our proposed approach with a particular choice of kernel and bandwidth. 

In this paper, we model these high-dimensional parameters as smooth functions of time whose functional forms are unknown and are estimated non-parametrically. We present two $L_2$ boosting algorithms which differ in their choice of base learners; the first uses componentwise local constant estimators as base learners, while the second relies on componentwise local linear estimators as base learners. We consider the use of local linear estimators since they have been shown to be a superior estimator theoretically, with smaller asymptotic bias at the boundaries of the sample \citep{cai2007}. We establish consistency of both our methods when dealing with high dimensional locally stationary predictors and errors with only polynomially decaying tails. Although we focus on linear time varying parameter models, $L_2$ boosting methods can easily be adapted to fit more general non-linear models by considering alternative base learners such as regression trees with varying degrees of depth. This makes the $L_2$ boosting framework more flexible than the often used $\ell_1$ penalized likelihood approaches.

The smooth TVP model considered in this paper has been studied in the econometrics literature for the case when the number of  predictors is fixed and assumed known. Under this assumption, \cite{robinson1989,robinson1991} studied the asymptotic properties of the local constant estimator of the coefficient functions. The theory was further developed in several directions.\footnote{ Some examples include: \cite{orbe2005,orbe2006} considered shape restricted estimation. \cite{cai2007} analyzed the asymptotic properties of the local linear estimator. \cite{inoue2017} considered the question of optimal bandwidth selection for the local constant estimator when using the uniform kernel. \cite{zhang2015time}, \cite{hu2018estimation}, and \cite{vogt2012nonparametric} allow for non-stationary predictors and non-linear time varying functions of these predictors. \cite{zhou2009local,zhou2010nonparametric} considered local linear quantile estimation, \cite{phillips2017estimating} obtained results for cointegration models, and \cite{chen2015modeling} dealt with models with endogenous predictors.}  To our knowledge, there were only two attempts at modeling sparse high dimensional smooth TVP models, both dealing with locally stationary sub-Gaussian predictors, and rely on $l_1$ regularization methods along with kernel smoothing to estimate the coefficient functions. In particular, \cite{ding2017} deals with locally stationary sparse VAR processes, and proposes a hybrid estimator which combines $l_1$ regularization with local constant estimation. \citet{lee2016} deals with models where the set of non-zero coefficient functions does not change with over time, and proposes a computationally intensive penalized local linear estimation method. Our work adds to this line of research by proposing $L_2$ boosting algorithms for high dimensional smooth TVP models characterized by (\ref{htvp}).  

Our methods compare favorably to more commonly used alternatives for modeling time varying parameters such as assuming the coefficients are stochastic and generated by a random walk, or using a rolling window estimator with a fixed window length. These models are typically estimated via MCMC, or other computationally intensive methods, which excludes the use of high dimensional datasets. Rolling window forecasts, although they are usually not presented this way, are actually equivalent to using a local constant estimator using a uniform kernel and a fixed bandwidth. This choice of fixed bandwidth is arbitrary and can lead to larger forecast errors vs using the optimal bandwidth \citep{inoue2017}. Additionally, local constant estimators have higher asymptotic bias at the boundary of the sample vs local linear estimators. In contrast, our $L_2$ boosting algorithms are capable of variable selection and estimation simultaneously at a very low computational cost even for very high dimensional data. Also, using non-parametric methods to estimate the time varying coefficient functions allows our method to perform well even under model misspecifications such as discrete breaks, stochastic coefficients generated by a random walk, and time invariant coefficients; see \cite{giraitis2013adaptive,inoue2017} and our simulations section for more details.

On the empirical side we include an extensive application to macroeconomic forecasting. Although parameter instability has long been established in the econometrics literature \citep{stock2003forecasting,stock2009forecasting,breitung2011testing}, the question of whether one can exploit this instability to improve macroeconomic forecasts is far less clear (see section \ref{Empirical} or \cite{rossi2013advances} for more details). Some issues which have hindered the utility of modeling time variation are: 1) the bias-variance tradeoff encountered when using a reduced sample for modeling, 2) misspecification and/or estimation error incurred when trying to estimate the nature of time variation, and 3) computational constraints restricting the use of high dimensional predictors when estimating traditional TVP models with stochastic coefficients. 

To analyze the effectiveness of modeling time variation with our methods, we use a panel of 123 monthly series from the FRED-MD database and focus on forecasting 8 major macroeconomic series over a range of forecast horizons. Using an out of sample period of over 47 years, we find that: \textbf{1)} the benefits of modeling time variation with our methods are substantial, especially when considering longer forecast horizons, \textbf{2)} the timing of these benefits suggests that the Great Moderation is associated with large amounts of parameter instability in the conditional mean of various economic series, and \textbf{3)} the benefits of modeling time variation appear to be confined to the high dimensional setting, as we confirm the results in \cite{Stock96} that modeling time variation in AR models offers little to no benefits for the majority of series.


The rest of the paper is organized as follows. Section \ref{sec:section 2} reviews the locally stationary framework, along with the functional dependence measure which will be used to quantify dependence. We also discuss the assumptions placed on the structure of the covariate and response processes; these assumptions are very mild, allowing us to represent a wide variety of stochastic processes which arise in practice. Section \ref{methodology} introduces our boosting algorithms for both local constant or local linear least squares base learners, and studies the asymptotic properties of these procedures. The asymptotic properties, and the number of predictors allowed depend on the strength of dependence, and the moment conditions of the underlying processes. Section \ref{simulations} presents results from Monte Carlo simulations, and sections \ref{Empirical} and \ref{results} contain our application to macroeconomic forecasting. Lastly, concluding remarks are in section \ref{disucussion}.




\section{The Econometric Framework}\label{sec:section 2}

We first start with a review of locally stationary processes which were first introduced by \cite{dahlhaus1996kullback,dahlhaus1997fitting} using a time varying spectral representation. This was expanded in \cite{dahlhaus2017towards} to a more general definition which facilitated theoretical results for a large class of non-linear processes; see \cite{dahlhaus2012locally} for a partial survey of the results pertaining to locally stationary processes. Heuristically speaking, a locally stationary process is a non-stationary process which can be well approximated by a stationary process locally in time. This is a convenient framework to model non-stationarity induced by smooth time varying parameters. Consider the model (\ref{htvp}), with $\bm{\beta}_t$ being a vector of unknown deterministic smooth functions of time, as a consequence $Y_t$ in (\ref{htvp}) is clearly non-stationary. Due to this non-stationarity, letting $T \rightarrow \infty$ will not lead to consistent estimates of $\bm{\beta}_t$, since future observations may not contain any information about the probabilistic structure of the process at the present time $t$. Therefore, it is common to work in the infill asymptotics framework with rescaled time $t/T \in [0,1]$, with $\bm{\beta}_t=\bm{\beta}(t/T)$ \citep{dahlhaus1997fitting,robinson1989,cai2007}. Letting $T \rightarrow \infty$ now implies that we observe $\bm{\beta}(t/T)$ on a finer grid within the same interval, thereby increasing the amount of local information available. Although this setting is not commonly seen in forecasting time series, a prediction theory is still possible. For example, we can view our data as having been observed for $t=1,\ldots,T/2$ (i.e. on the interval $[0,1/2]$), and we are forecasting the next few observations (see \cite{dahlhaus1997fitting,dahlhaus1996kullback}).  

For a formal description of locally stationary processes we use the definition and assumptions stated in \cite{dahlhaus2017towards} and \cite{richter2017cross}:

\begin{defn} Let $q>0$, and $||W||_q=(E|W|^{q})^{1/q}$. Let $Y_{t,T}, t=1,\ldots,T$ be a triangular array of stochastic processes. For each $u \in [0,1]$, let $\tilde{Y}_t(u)$ be a stationary and ergodic process satisfying:
\begin{enumerate}
    \item $D_q=\max\{\sup_{u\in [0,1]}||\tilde{Y}_t(u)||_q,\sup_{T\in N}\sup_{t=1,\ldots,T}||Y_{t,T}||_q\} < \infty$

    \item There exists $C_B>0$ such that uniformly in $t=1,\ldots,T$ and $u,v \in [0,1]$:
    \begin{align}
        ||\tilde{Y}_t(u)-\tilde{Y}_t(v)||_q\leq C_B|u-v|, \textrm{ } ||Y_{t,T}-\tilde{Y}_{t}(t/T)||_q \leq C_BT^{-1}
    \end{align}
\end{enumerate}
\label{defL}
\end{defn}
From the second assumption we obtain: $||Y_{t,T}-\tilde{Y}_t(u)||_q\leq O(|t/T-u|+T^{-1})$, thus for rescaled time points $t/T$ near $u$, the process $Y_{t,T}$ can be approximated by a stationary process $\tilde{Y}_t(u)$ with asymptotically negligible error. Consider the model used in \cite{robinson1989,cai2007}: $Y_{t,T}=\bm{\beta}(t/T)\bm{X}_{t}+\epsilon_{t}$, where $\bm{X}_t, \epsilon_t$ are stationary processes, and $\bm{\beta}(\cdot)$ is a lipschitz continuous function. Under these conditions $Y_{t,T}$ is a locally stationary process, with stationary approximation: $\tilde{Y}_{t}(u)=\bm{\beta}(u)\bm{X}_{t}+\epsilon_{t}$. A slightly more complicated example is a tvAR(1) process: $Y_{t,T}=\alpha(t/T)Y_{t-1,T}+\epsilon_t=\sum_{j=0}^{\infty}[\prod_{k=1}^{j-1}\alpha(\frac{t-k}{T})]\epsilon_{t-j}$. Intuitively one can see that if we assume $\alpha(\cdot)$ is lipschitz continuous then the process is locally stationary with stationary approximation: $\tilde{Y}_t(u)=\alpha(u)\tilde{Y}_{t-1}(u)+\epsilon_t$, and $||Y_{t,T}-\tilde{Y}_t(u)||_q \leq O(|t/T-u|+T^{-1})$.\footnote{Under appropriate conditions, more general non-linear time varying processes which satisfy the recursion:  $Y_{t,T}=G_{\epsilon_t}(Y_{t-1,T},\ldots,Y_{t-p,T},\max(t/T,0)), \textrm{ for }t\leq T$, can be shown to be locally stationary \citep{dahlhaus2017towards}. Examples of such processes include time varying ARMA, time varying GARCH, time varying VAR, and time varying random coefficient processes.} The stationary approximation is the key to estimation and formulating an asymptotic theory when dealing with locally stationary processes. Estimation of parameters such as $\alpha(u)$ and local covariances is carried out by assuming, for each rescaled time point $u$, that the process is essentially stationary on a small window around $u$. We then carry out estimation via stationary methods using observations within this window.\footnote{We note that assuming approximate stationarity on a small window is essentially the justification of the commonly used rolling window estimators.}
 
In order to establish asymptotic properties of our $L_2$ boosting procedures, we rely on the functional dependence measure introduced in \cite{Wu2005} and used in the context of locally stationary processes in \cite{dahlhaus2017towards, richter2017cross}. We first introduce the following notation: Let $\{e_t\}_{t \in \mathcal{Z}}$ be a sequence of iid random variables, and let $\mathcal{F}_t=(e_t,e_{t-1},\ldots)$, $\mathcal{F}_t^{*}=(e_t,e_{t-1},\ldots,e_{0}^{*},e_{-1},\ldots)$ with $e_{0}^{*}, e_t, t \in \mathbb{Z}$ being iid. Additionally, let $\mathcal{H}_t= (\bm{\eta}_t,\bm{\eta}_{t-1},\ldots)$, $\mathcal{H}_t^{*}=(\bm{\eta}_t,\bm{\eta}_{t-1},\ldots,\bm{\eta}_{0}^{*},\bm{\eta}_{-1},\ldots)$ with $\bm{\eta}_{0}^{*},\bm{\eta}_{t}, t\in \mathbb{Z}$ being iid random vectors. Throughout this paper, we assume the following structure for the stationary approximation for univariate processes (such as the response and error processes), and multivariate processes (such as the covariate process) respectively: \begin{align} \tilde{Y}_{t}(u)=g(u,\mathcal{F}_t) \textrm{ and } \bm{\tilde{x}}_{t}(u)=\bm{h}(u,\mathcal{H}_t)=(h_{1}(u,\mathcal{H}_t),\ldots,h_{p_T}(u,\mathcal{H}_t)), \label{form}
\end{align}
where $g(\cdot,\cdot)$, and $\bm{h}(\cdot,\cdot)$ are real valued measurable functions. These representations allow us to define the functional dependence measure as: $\delta_{q}^{\tilde{Y}(u)}(t)=||\tilde{Y}_{t}(u)-g(u,\mathcal{F}_t^{*})||_q$, and $\delta_{q}^{\tilde{X}_{j}(u)}(t)=||\tilde{X}_{j,t}(u)-h_{j}(u,\mathcal{H}_t^{*})||_q$. Additionally, we assume short range dependence of the form:
\begin{align}
\Delta_{0,q}^{\tilde{Y}}=\sum_{k=0}^{\infty} \sup_{u \in [0,1]}\delta_{q}^{\tilde{Y}(u)}(k) \leq \infty, 
\textrm{ and } \Phi_{0,q}^{\bm{\tilde{x}}}=\max_{j \leq p_T}\sum_{k=0}^{\infty} \sup_{u \in [0,1]}\delta_{q}^{\tilde{X}_{j}(u)}(k) \leq \infty,
\label{weakdependence0}
\end{align}
for some $q>2$ to be specified in the next section. 

We place assumptions on the stationary approximation rather than directly on the process itself. This leads to results using weaker assumptions, and  to more interpretable dependence measures. For an intuitive explanation of this measure, we consider the stationary approximation at time $u_0$ ($\tilde{Y}_t(u_0)$) and we obtain $\delta_{q}^{\tilde{Y}(u_0)}(k)=||\tilde{Y}_k(u_0)-g(u_0,\mathcal{F}_k^{*})||_q$. We can view $\delta_{q}^{\tilde{Y}(u_0)}(k)$ as measuring the dependence of $\tilde{Y}_k(u_0)$ on the innovation $\epsilon_0$, which for weakly dependent processes decreases suitably quickly as $k \rightarrow \infty$. For a concrete example, consider a stationary AR(1) process  $\tilde{Y}_t(u_0)=\sum_{j=0}^{\infty} a(u_0)^j e_{t-j}$ with $e_i$ iid, then $\delta_{q}^{\tilde{Y}(u_0)}(k)=|a(u_0)^k|||e_0-e_0^{*}||_q$, and $\Delta_{0,q}^{\tilde{Y}(u_0)}=||e_0-e_0^{*}||_q\sum_{k=0}^{\infty}|a(u_0)^k|$. Now in the locally stationary setting, we take the supremum over the rescaled time interval to account for the non-stationarity of the processes, thereby obtaining $\Delta_{0,q}^{\tilde{Y}}=||e_0-e_0^{*}||_q\sup_{u \in [0,1]}\sum_{k=0}^{\infty}|a(u)^k|$. A very wide variety of locally stationary processes encountered in practice including time varying linear processes, tv-ARMA, tv-GARCH, tv-TAR, and tv-VAR, and time varying random coefficient processes have stationary approximations which satisfy (\ref{weakdependence0}), and have geometrically decaying functional dependence measures (see \cite{dahlhaus2017towards}).

Compared to mixing coefficients, functional dependence measures are easier to interpret and compute since they are directly related to the data generating mechanism of the underlying process. For example, consider a stationary ARMA($p,q$) process, $\xi_t=e_t+\sum_{i=1}^{p}\alpha_i\xi_{t-i}+\sum_{j=1}^{q}\beta_je_{t-j}$ with $e_t$ iid, under appropriate conditions we can show this process is $\beta$-mixing with $\beta(t)=O(\rho^{t})$ and $\rho \in (0,1)$ \citep{Doukhan94}. However, there exists no known mapping between the parameters and $\rho$ \citep{mcdonald2015}. In contrast, by using the moving average representation for $\xi_t=\sum_{j=0}^{\infty} f_j e_{i-j}$, we see the functional dependence measure is directly related to the data generating process: $\delta_{q}(\xi_t)=O(f_t)$. Additionally, in many cases using functional dependence measures also requires less stringent assumptions \citep{Yousuf2017,Wu2005}. \footnote{The discussions focus on stationary processes and we note that verifying mixing conditions for locally stationary processes usually requires additional assumptions, as the techniques for showing strong mixing for stationary processes no longer apply \citep{fryzlewicz2011mixing}. In contrast, since we place our dependence assumptions on the stationary approximation we can directly use results for functional dependence measures obtained in the stationary setting.} We do note that functional dependence measures are restricted to a more limited class of processes, specifically those possessing the representation (\ref{form}). Fortunately, this is a very weak restriction as virtually all time series models used in practice have this representation (see \cite{hormann2010weakly,Wu2011} and references therein).

\section{Boosting High Dimensional TVP Models}\label{methodology}

Ever since the introduction of AdaBoost in the 1990's \citep{freund1997decision}, boosting algorithms have been one of the most successful and widely utilized machine learning methods \citep{ESL}. AdaBoost, which was developed for classification, consisted of iteratively fitting a series of weak classifiers or learners onto reweighted data and taking a weighted average of the predictions from each of these simple models. The success of AdaBoost was originally thought to originate from averaging many weak classifiers and from a reweighting scheme which placed large weights on heavily misclassified observations. Later work by \cite{friedman2001greedy}, and \cite{friedman2000additive} established AdaBoost as a gradient descent algorithm in function space using an exponential loss function. This functional gradient descent view connected boosting to the common optimization view of statistical inference, and led to extensions of boosting beyond the realm of classification. \cite{friedman2001greedy} proposed several new boosting algorithms using alternative base learners and loss functions including squared error loss, leading to $L_2$ boosting. Additionally, \cite{efron2004least} and \cite{ESL}, made connections for linear models between $L_2$ boosting and common statistical procedures such as the Lasso and forward stagewise regression.\footnote{For theoretical connections one can consult Chapter 16.2 of \cite{ESL}, and additional works such as \cite{hastie2007forward,rosset2004boosting}.} \footnote{Empirical comparisons between boosting with linear least squares learners and the lasso have shown close performance with boosting performing slightly better in the case of high correlated predictors \citep{hastie2007forward,hepp2016approaches}.} These insights shed light on $L_2$ boosting as a method which performs variable selection and shrinkage leading to sparse models. 
For an excellent survey of the statistical view of boosting and results pertaining to several common boosting algorithms, one can consult \cite{Buhlmann2007}. \footnote{Additionally, one can consult \cite{Buhlmann2006} for extensions of boosting to stationary VAR processes, and \cite{bai2009boosting,NgBoosting2014} for applications to macroeconomic forecasting and recession classification respectively.}


We are interested in estimating the following model:
\begin{align}
    Y_{t,T}=\bm{\beta}'(t/T)\bm{x}_{t-h,T}+\epsilon_{t,T} \textrm{ for } t=1,\ldots,T,\label{htvp2}
\end{align}
where $Y_{t,T}$ is the response, $\bm{x}_{t-h,T}=(X_{1,t-h,T},\ldots,X_{p_T,t-h,T})'$ is a $p_T$-dimensional vector of locally stationary predictors (with $p_T >>T$), $\bm{\beta}'(t/T)=(\beta_{1}(t/T),\ldots,\beta_{p_T}(t/T))$ is a vector of unknown functions of time defined on the grid $[0,1]$, which becomes finer as $T \rightarrow \infty$, and $\epsilon_{t,T}$ denotes the locally stationary error process with $E(\epsilon_{t,T}\bm{x}_{t-h,T})=0 \textrm{ } \forall \textrm{    } t,T$. We denote the stationary approximation of the response as $\tilde{Y}_t(u)=\bm{\beta}'(u)\tilde{\bm{x}}_{t-h}(u)+\tilde{\epsilon}_{t}(u)$. To simplify notation, we discuss estimation at the boundary point $u=T/T=1$. Before we introduce our boosting algorithms, it helps to first introduce the population version of componentwise $L_2$ boosting with linear base learners as applied to the stationary approximations $(\tilde{Y}_{T}(u),\tilde{\bm{x}}_{T-h}(u))$:\\
\\
\textbf{\textit{Algorithm: Population level $L_2$ Boosting}}

\begin{enumerate}
    \item Set $F^{(0)}(u,\tilde{\bm{x}}_{T-h}(u))=E(\tilde{Y}_{T}(u))$
    \item For $m=1,\ldots,M_{T}$, where $M_{T}$ is some stopping iteration, do:
    \begin{enumerate}
        \item Compute $\tilde{U}^{(m)}_{T}(u)=\tilde{Y}_T(u)-F^{(m-1)}(u,\tilde{\bm{x}}_{T-h}(u))$.
        \item Let $\mathcal{S}_m=\argmin_{j \leq p_T}E(\tilde{U}_{T}^{(m)}(u)-\alpha_j^{(m)}(u)\tilde{X}_{j,T-h}(u))^2$,\\
                where $\alpha^{(m)}_j(u)=E(\tilde{X}_{j,T-h}(u)\tilde{U}_{T}^{(m)}(u))/E(\tilde{X}^2_{j,T-h}(u))$.
        \item Update $F^{(m)}(u,\tilde{\bm{x}}_{T-h}(u))=F^{(m-1)}(u,\tilde{\bm{x}}_{T-h}(u))+\upsilon \cdot \alpha^{(m)}_{\mathcal{S}_m}(u)\tilde{X}_{\mathcal{S}_m,T-h}(u)$, where $\upsilon \in (0,1]$ is a step length factor.
    \end{enumerate}
    \item Output ${F}^{(M_T)}(u,\tilde{\bm{x}}_{T-h}(u))=F^{(0)}(u,\tilde{\bm{x}}_{T-h}(u))+\upsilon \sum_{m=1}^{M_T}\alpha^{(m)}_{\mathcal{S}_m}(u)\tilde{X}_{\mathcal{S}_m,T-h}(u)$

\end{enumerate}

Although we use linear base learners, we note that our methods can be extended to a broader class of models by using a more general base learner, such as $g_j(u,\tilde{X}_{j,T-h}(u))=E(\tilde{Y}_{T}(u)|\tilde{X}_{j,T-h}(u))$, and estimating using kernel regressions or smoothing splines. For the corresponding sample version of $L_2$ boosting with linear base learners, it is informative to consider the case of stationary response and predictor processes. In the stationary setting, we can remove the dependence on $T$ and the sample version of our algorithm simplifies to $\mathcal{\hat{S}}_m=\argmin_j\sum_{t=1}^{T}(U^{(m)}_{t}-\hat{\alpha}^{(m)}_j X_{t,j})^2$, where $\hat{\alpha}^{(m)}_j=T^{-1}\sum_{t=1}^{T}X_{j,t-h}U^{(m)}_t$, assuming $E(X_t),E(Y_t)=0$, and $E(X_t^2)=1$. For the case of locally stationary response and predictor processes the situation is more complicated as the above estimator is inconsistent for $\alpha_j^{(m)}(u)$. Intuitively, this inconsistency arises since observations ``far" from rescaled time $u$ contain little information about the probabilistic structure of the processes at time $u$.

To proceed with estimation in the locally stationary setting, $\forall m$ and $j \leq p_T$, we have $U^{(m)}_{t,T}=\alpha^{(m)}_j(t/T)X_{j,t-h,T}+\epsilon_{j,t,T}$, where $\alpha^{(m)}_j(t/T)=E(\tilde{X}_{j,t-h}(t/T)\tilde{U}^{(m)}_{T}(t/T))/E(\tilde{X}^2_{j,t-h}(t/T))$.\footnote{Recall that $E(X_{j,t-h,T}U^{(m)}_{t,T})/E(X_{j,t-h,T}^2)=\alpha^{(m)}_j(t/T)+O(T^{-1})$ by local stationarity.} By local stationarity and assuming appropriate smoothness conditions, we have the following expansion: 
\begin{align}
  \alpha^{(m)}_{j}(t/T)=\alpha^{(m)}_j(u)+\dot{\alpha}_j^{(m)}(u)(t/T-u)+\ddot{\alpha}_j^{(m)}(c)(t/T-u)^2,  
\label{Taylor}
\end{align}
where $\dot{\alpha}(\cdot),\ddot{\alpha}(\cdot)$ denote the first and second derivative respectively of the function, with $c$ between $u$ and $t/T$. To compute the local constant estimate for $\alpha^{(m)}_{j}(u)$, we ignore the linear term in the Taylor expansion to obtain the following approximation: $U^{(m)}_{t,T} \approx \alpha^{(m)}_{j}(u)X_{j,t-h,T}+\epsilon_{j,t,T}$ for $t/T$ near $u$. The local constant estimator for $\alpha^{(m)}_{j}(u)$ is then 
\begin{align}
    \hat{\alpha}^{(m)}_{lc,j}(u)=\frac{\sum_{t=1}^{T}K_b(t/T-u)X_{j,t-h,T}U^{(m)}_{t,T}}{\sum_{t=1}^{T} K_b(t/T-u)X^2_{j,t-h,T}},
\end{align}
where $K_b(x)=b^{-1}K(x/b)$, is a kernel function and $b$ is the bandwidth. Therefore, $\hat{\alpha}^{(m)}_{lc,j}(u)$ is a weighted least squares estimate, with the weights given by the kernel values. For now, one can think of this estimator as aiming to use information from observations ``near" time $T$, while discounting information from distant points. A simple example of the local constant estimate is the rolling window estimate: using the uniform kernel $K(x)=\mathbbm{1}_{|x| \leq 1}$, with a fixed bandwidth $b=b_0$, we obtain a rolling window estimate which uses the last $b_0T$ observations in our sample.

The local constant estimate is widely used for estimating time varying effects, however the Taylor expansion of $\alpha^{(m)}_{j}(t/T)$ suggests we can obtain a better approximation by using the linear term in the expansion (\ref{Taylor}). This was analyzed rigorously in \cite{cai2007}, which showed that for boundary points the local linear estimator is theoretically superior to the local constant estimator. Using the expansion (\ref{Taylor}), we obtain: $U^{(m)}_{t,T} \approx \alpha^{(m)}_{j}(u)X_{j,t-h,T}+\dot{\alpha}^{(m)}_{j}(u)X_{j,t-h,T}(t/T-u)+\epsilon_{j,t,T}$, for $t/T$ near $u$. Let $\bm{Z}_{j,t-h,T}\bm{\theta}^{(m)}_{j}(u)$ where $\bm{Z}_{j,t-h,T}=(X_{j,t-h,T},X_{j,t-h,T}(t/T-u))$, $\bm{\theta}^{(m)}_{j}(u)=(\alpha^{(m)}_{j}(u),\dot{\alpha}^{(m)}_{j}(u))'$. The local linear estimate is obtained by minimizing a weighted least squares criterion:
\begin{align}
   \hat{\bm{\theta}}_j^{(m)}(u)=(\hat{\alpha}^{(m)}_{ll,j}(u),\hat{\dot{\alpha}}^{(m)}_{ll,j}(u))= \argmin_{\bm{\theta}_j^{(m)}(u)}\sum_{t=1}^{T}K_b(t/T-u)(U^{(m)}_{t,T}-\bm{Z}_{t-h,T}\bm{\theta}_j^{(m)}(u))^2 
\end{align}

Using these estimators we can formulate our $L_2$ boosting algorithm for (\ref{htvp2}) using local constant, and local linear  estimators as base learners. We first start with our first algorithm which uses local constant estimators:\\
\\
\textbf{\underline{\textit{Algorithm 1: Local Constant $L_2$ Boosting (LC-Boost)}}}

\begin{enumerate}
    \item Set $\hat{F}_{lc}^{(0)}(u,\bm{x}_{t,T})=T^{-1}\sum_{i=h+1}^{T}K_b(i/T-u)Y_{i,T}$, for $t=1,\ldots,T-h$
    \item For $m=1,\ldots,M_{T}$, where $M_{T}$ is some stopping iteration, do:
    \begin{enumerate}
        \item Compute the residuals $\hat{U}^{(m)}_{i,T}=Y_{i,T}-\hat{F}_{lc}^{(m-1)}(u,\bm{x}_{i-h,T})$ for $i=h+1,\ldots,T$.
        \item Let $\hat{\mathcal{S}}_m=\argmin_{j \leq p_T}\sum_{i=h+1}^{T}K_b(i/T-u)(\hat{U}^{(m)}_{i,T}-\hat{\alpha}_{lc,j}^{(m)}(u)X_{j,i-h,T})^2$
        \item Update $\hat{F}_{lc}^{(m)}(u,\bm{x}_{i-h,T})=\hat{F}_{lc}^{(m-1)}(u,\bm{x}_{i-h,T})+\upsilon \hat{\alpha}^{(m)}_{lc,\mathcal{S}_m}(u)X_{\mathcal{S}_m,i-h,T}$, where $\upsilon \in (0,1]$ is a step length factor.
    \end{enumerate}
    \item Output $\hat{F}_{lc}^{(M_T)}(u,\bm{x}_{T-h,T})=\hat{F}_{lc}^{(0)}(u,\bm{x}_{t,T})+\upsilon\sum_{m=1}^{M_T} \hat{\alpha}^{(m)}_{lc,\mathcal{S}_m}(u)X_{\mathcal{S}_m,T-h,T}$
\end{enumerate}

Let $\bm{z}_{t,T}=(\bm{x}_{t,T},\bm{x}_{t,T}(t/T-u))$, our boosting algorithm using local linear estimates as base learners is:\\
\\
\textbf{\underline{\textit{Algorithm 2: Local Linear $L_2$ Boosting (LL-Boost)}}}

\begin{enumerate}
    \item Set $\hat{F}_{ll}^{(0)}(u,\bm{x}_{i-h,T})=T^{-1}\sum_{i=h+1}^{T}K_b(i/T-u)Y_{i,T}$, for $i=h+1,\ldots,T$
    \item For $m=1,\ldots,M_{T}$, where $M_{T}$ is some stopping iteration, do:
    \begin{enumerate}
        \item Compute the residuals $\hat{U}^{(m)}_{i,T}=Y_{i,T}-\hat{F}_{ll}^{(m-1)}(\bm{x}_{i-h,T})$ for $i=h+1,\ldots,T$.
        \item Let $\hat{\mathcal{S}}_m=\argmin_{j \leq p_T}\sum_{i=1}^{T}K_b(i/T-u)(\hat{U}^{(m)}_{i,T}-\bm{Z}_{j,i-h,T}\hat{\bm{\theta}}_{j}^{(m)}(u))^2$.\\
        \item Update $\hat{F}_{ll}^{(m)}(u,\bm{x}_{i-h,T})=\hat{F}^{(m-1)}_{ll}(u,\bm{z}_{i-h,T})+\upsilon \cdot \bm{Z}_{\mathcal{S}_m,i-h,T}\hat{\bm{\theta}}^{(m)}_{\mathcal{S}_m}(u)$, where $\upsilon \in (0,1]$ is a step length factor.
    \end{enumerate}
    \item Output $\hat{F}_{ll}^{(M_T)}(u,\bm{x}_{T-h,T})=\hat{F}_{ll}^{(0)}(u,\bm{x}_{i-h,T})+\upsilon \sum_{m=1}^{M_T}\bm{Z}_{\mathcal{S}_m,T-h,T}\hat{\bm{\theta}}^{(m)}_{\mathcal{S}_m}(u)$

\end{enumerate}

We see that boosting is a stagewise estimation procedure, where at each stage only one learner is updated and the previously selected terms are unchanged. This stagewise fitting procedure induces regularization through limiting the number of steps ($M_T$), and the step length factor ($\upsilon$). We usually fix the the step-length factor ($\upsilon$) to a low number such as $\upsilon=.1$, making the stopping iteration ($M_T$) akin to the regularization parameter of the Lasso.\footnote{Given that each predictor can be selected multiple times, especially for low values of $\upsilon$, the number of predictors in the estimated model is $\leq M_T$, and all predictors which have not been selected by step $M_T$ have an effect of zero.} In light of this, boosting can be thought of as a close relative of the lasso, with the advantage of being able to approximate the $\ell_1$ penalized solution in situations where it is impossible or computationally burdensome to compute the Lasso solution \citep{friedman2004discussion}. 

By viewing boosting as a general regularized function estimation procedure, we can formulate a generic local constant boosting procedure which can be easily be computed for a wide variety of base learners and (almost everywhere) differentiable loss functions ($L(\cdot,\cdot)$).  \\
\\
\textbf{\underline{\textit{Algorithm 3: Generic Local Constant Boosting}}}

\begin{enumerate}
    \item Set $\hat{F}_{G}^{(0)}(u,\bm{x}_{t,T})=\argmin_{c} T^{-1}\sum_{i=h+1}^{T}K_b(i/T-u)L(Y_{i,T},c)$, for $t=1,\ldots,T-h$
    \item For $m=1,\ldots,M_{T}$, where $M_{T}$ is some stopping iteration, do:
    \begin{enumerate}
        \item Compute the pointwise negative gradient: $U_{i,T}^{(m)}=\frac{d}{dF} L(Y_{i,T},F)\Bigr|_{\substack{F=\hat{F}_{G}^{(m-1)}(u,\bm{x}_{i-h,T})}}$ evaluated at  $i=h+1,\ldots,T$.
        \item Let $\hat{\mathcal{S}}_m=\argmin_{j \leq p_T}\sum_{i=h+1}^{T}K_b(i/T-u)(\hat{U}^{(m)}_{i,T}-\widehat{g}_{j}^{(m)}(u,X_{j,i-h,T}))^2$
        \item Update $\hat{F}_{G}^{(m)}(u,\bm{x}_{i-h,T})=\hat{F}_{G}^{(m-1)}(u,\bm{x}_{i-h,T})+\upsilon \widehat{g}_{\mathcal{S}_m}^{(m)}(u,X_{\mathcal{S}_m,i-h,T})$, where $\upsilon \in (0,1]$ is a step length factor.
    \end{enumerate}
    \item Output $\hat{F}_{G}^{(M_T)}(u,\bm{x}_{T-h,T})=\hat{F}_{G}^{(0)}(u,\bm{x}_{t,T})+\upsilon\sum_{m=1}^{M_T} \widehat{g}_{\mathcal{S}_m}^{(m)}(u,X_{\mathcal{S}_m,T-h,T})$
\end{enumerate}
The algorithm can be modified to allow $g_{j}(u,\cdot)$ to be a function of several variables e.g. a predictor along with a number of its lags.


\section{Implementation}
\label{hyperparameter}

Implementation of these algorithms is very simple and can be carried out using existing software packages. We first discuss the choice of the kernel function $K(\cdot)$, bandwidth ($b$), stopping iteration ($M_{T}$), and step length factor ($\upsilon$). We set $\upsilon=.1$, which is the default choice in statistical software packages and applied work \citep{Buhlmann2007,friedman2001greedy,hofner2014model}. In non-parametric statistics and machine learning the most commonly used kernels are the Gaussian Kernel and the Epanechnikov kernel $K(u)=.75(1-u^2)\mathbbm{1}_{|u| \leq 1}$, while in forecasting the uniform kernel $\mathbbm{1}_{|u| \leq 1}$ is more widely used. Both the uniform kernel and the Epanechnikov Kernel use a subset of the sample, with the Epanechnikov kernel also downweighting more distant observations within this subset. The Gaussian kernel does not truncate the sample, instead it smoothly downweights more distant observations. It has a much smoother downweighting scheme than the Epanechnikov kernel, which can be beneficial in many applications.\footnote{We decide to use the uniform kernel in our applications due to its close connections with the rolling window estimator. Using the Gaussian Kernel gave us similar results.} In general however, the choice of a kernel does not have much impact on the performance, as opposed to the selection of the bandwidth parameter which is crucial.

We first discuss bandwidth selection for an out of sample forecasting exercise. To help with exposition, we use a concrete example: assume we have monthly data ranging from 1960:1 to 2018:8, giving us about $\sim 700$ observations. We begin our forecasts on 1970:1 and move forward utilizing an expanding window framework. We use one-sided kernels to avoid looking into the future. We choose our bandwidth parameter using a cross validation approach. We first form a grid of values $B=(b_1,\ldots,b_n)$ from which to select the bandwidth parameter.  For each forecast, our cross validation procedure uses the last $\omega$ (where $\omega$ is chosen by the researcher) observations of our sample for an out of sample forecasting exercise. We then choose the bandwidth which minimizes the MSFE over this sub-sample. Therefore, the selected bandwidth is:
\begin{align*}
    b^{*}_{T_0}=\argmin_{b_i \in B}\omega^{-1}\sum_{\tau=T_0-\omega}^{T_0-h}(Y_{\tau,T}-\hat{F}_{\tau,b_i}^{(M_T)}(\tau/T,\bm{x}_{\tau-h,T}))^2,
\end{align*}
where $\hat{F}_{\tau,b_i}^{(M_T)}(\tau/T,\bm{x}_{\tau-h,T})$ refers to the LC-Boost or LL-Boost estimate of $\bm{x}_{\tau-h,T}\beta(\tau/T)$ using only observations until time $\tau$, and the bandwidth $b_i$. For our first out of sample forecast we set $T_0=120$, which is the length of the sample available at the time, and for each additional forecast we increment $T_0$ by 1 until we reach the end of the sample.\footnote{In order to simplify implementation, for each out of sample forecast we let $T=T_0$ which is the sample size available for estimation at the time of the forecast. We could alternatively define $T$ as the sample size available at the last out of sample forecast date ($T \approx 700$), and we still achieve the same estimates, but we would have to vary our grid $B$ over time.}  In the special case of using LC-Boost with a one sided uniform kernel, we are selecting the optimal window size at each time point, via cross validation, for a rolling window forecast. With the bandwidths representing the fraction of the sample available at the time of the forecast that we are using for estimation.

For in-sample estimation problems, two sided kernels are used in our algorithms with a weighted leave one out cross validation procedure to select the bandwidth. The procedure is as follows:
\begin{align*}
    b^{*}_{T_0}=\argmin_{b_i \in B}T^{-1}\sum_{\tau=h}^{T}(Y_{\tau,T}-\hat{F}_{lc,-\tau,b_i}^{(M_T)}(\tau/T,\bm{x}_{\tau-h,T}))^2K_{b_i}(\tau/T-T_0/T),
\end{align*}
where $\hat{F}_{lc,-\tau,b_i}^{(M_T)}(\tau/T,\bm{x}_{\tau-h,T})$ refers to the estimate of $\bm{x}_{\tau-h}\beta(\tau/T)$, which uses all observations except $(Y_{\tau,T},\bm{x}_{\tau-h,T})$. The kernel in the above equation discounts errors far away from the time point $t_0$ when selecting the optimal bandwidth. This procedure gives us a bandwidth for each time point in the sample, and if one wants a single bandwidth for all time points, the kernel can be removed. 

To select the stopping iteration $M_{T}$, we specify an upper bound for the number of iterations $M_{upp}$ (we set $M_{upp}=100$), where $M_{T} \leq M_{upp}$. The stopping iteration is then selected using the corrected AIC ($AIC_{c}$) statistic given in \cite{Buhlmann2006}:
$$M_{T}=\argmin_{m \leq M_{upp}} AIC_c(m),$$ where $AIC_c(m)$ is the AIC of the model using $m$ iterations.\footnote{Alternatively, we can jointly select $M_{T}$ and the bandwidth $b^{*}_{T_0}$ by forming a two dimensional grid and selecting the optimal combination using the cross validation procedure described earlier. We decide to use the $AIC_c$ statistic in this work. We note that when dealing with very large sample sizes and/or more complicated base learners which are a function of more than one variable, using cross validation to select $M_T$, using a moderately sized grid, can often be quicker since calculation of the corrected AIC requires computing the trace of the Hat matrix.} 

Our methods can be computed extremely quickly using the existing R package \textbf{mboost}. Our base learners are univariate or bivariate weighted least squares estimates which can be implemented through existing functions in the package once we specify the kernel values as weights. We can also  implement the generic local constant boosting algorithm  for wide a variety of base learners and loss functions such as absolute loss, Huber loss and quantile loss.\footnote{We refer the reader to \cite{hofner2014model} which provides an excellent introduction and tutorial to the \textbf{mboost} package. It also lists the wide variety of base learners and loss functions supported by the package.} As an example, to obtain quantiles for our forecasts, we specify the quantile loss for a given quantile\footnote{See \cite{fenske2011identifying} for more details on the quantile boosting algorithm.}, and compute the optimal bandwidth for our base learners by using the cross validation procedure mentioned above. A density forecast can be obtained from these estimated quantiles by using the procedure outlined in \cite{adrian2019vulnerable}.

\section{Asymptotic Theory}
In order to prove our asymptotic results, we need the following assumptions:

\begin{condition}\label{ConditionA}
Assume $\sup_{u \in [0,1]}|\bm{\beta}(u)|_1 < \infty$ 
\end{condition}
\begin{condition}\label{ConditionB}
Assume the error and the covariate processes are locally stationary and have representations given in (\ref{form}). Additionally, we assume the following decay rates $\Phi_{m,r}^{\bm{x}} = O(m^{-\alpha_x}), \Delta_{m,q}^{\epsilon} =O(m^{-\alpha_{\epsilon}})$, for some $\alpha_x, \alpha_{\epsilon} > 0$, $q > 2, r > 4 $ and $\tau=\frac{qr}{q+r} > 2$. 
\end{condition}

\begin{condition}\label{ConditionD}
Let $\Sigma_{\tilde{x}}(u)=E(\tilde{\bm{x}}'_{t}(u)\tilde{\bm{x}}_{t}(u))$ be the covariance matrix function. For $u \in [0,1]$, assume that $\bm{\beta}(u),\Sigma_{\tilde{x}}(u) \in \mathcal{C}^2[0,1]$, where $\mathcal{C}^2[0,1]$ denotes the class of functions defined on $[0,1]$ that are twice differentiable with bounded derivatives.  
\end{condition}
\begin{condition}\label{ConditionE}
The kernel function $K(u)$ is bounded and symmetric, and of bounded variation with compact support. Additionally, the bandwidth ($b$) satisfies $bT=S_T=O(T^{\psi})$, where $\psi \in (0,1)$.
\end{condition}

Condition \ref{ConditionA} requires $\ell_1$ sparsity of the time varying coefficients, and allows the active set of predictors to change over time. Our asymptotic results do not require sparsity in the number of non-zero coefficients ($\ell_0$ sparsity). Condition \ref{ConditionB} assumes the covariate and error processes are locally stationary, and presents the dependence and moment conditions on these processes, where higher values of $\alpha_x,\alpha_{\epsilon}$ indicate weaker temporal dependence. We assume our predictor and error processes have at least $r>4$ and $q>2$ finite moments respectively. Examples of processes satisfying condition \ref{ConditionB} were given in section \ref{sec:section 2}. 

Given that $\bm{x}_{t-h}$ can contain lags of $Y_{t,T}$, an example of a model which satisfies the above conditions is as follows: Let $\bm{W}_{t,T}=(Y_{t,T},\bm{z}_{t,T})$, where $\bm{z}_{t,T}$ represents our exogenous series, and $\bm{W}_{t,T}=\sum_{i=1}^{\ell}\bm{A}_i(t/T)\bm{W}_{t-i,T}+\bm{\eta}_{t}$. Then the stationary approximation is $\widetilde{\bm{W}}_t(t/T)=\sum_{i=1}^{\ell}\bm{A}_i(t/T)\widetilde{\bm{W}}_{t-i}(t/T)+\bm{\eta}_{t}$, with cumulative functional dependence measure $\Phi_{0,r}^{\bm{\tilde{W}}}=\sup_{u \in [0,1]}\sum_{k=0}^{\infty}O(\lambda_{\max}(\bm{A^}{*}(u))^k)$\citep{chen2013}, where $\bm{A}^{*}(u)$ is the companion matrix. We can then define $\bm{x}_{t-1,T}=(\bm{W}_{t-1,T},\ldots,\bm{W}_{t-l,T})$, and $\bm{\beta}(t/T)$ as the first row of the companion matrix $\bm{A}^{*}(u)$. We weaken the assumptions placed in the works \cite{cai2007, robinson1989,Bin12} which restricted the predictors and errors to be stationary, thus ruling out models with lagged dependent variables. Compared to previous works on high dimensional TVP models, such as \cite{ding2017,lee2016}, we use a different dependence framework, and allow the predictors and errors to have polynomially decaying tails.

Condition \ref{ConditionD} is a sufficient condition to guarantee that the expansion (\ref{Taylor}) exists, i.e: $\alpha_{j}^{(m)}(u) \in \mathcal{C}_{2}[0,1], \forall m$ and $j \leq p_T$. Sufficient conditions needed for smoothness of the covariance matrix function were given in \cite{ding2017} for the case of locally stationary VAR processes, and one can consult \cite{dahlhaus2017towards} for sufficient conditions for more general processes. Condition \ref{ConditionE} is a standard condition and it includes the commonly used Epanechnikov ($K(u)=.75(1-u^2)\mathbbm{1}_{|u| \leq 1}$) and uniform ($K(u)=\mathbbm{1}_{|u| \leq 1}$) kernels. It also places the standard conditions on the effective sample size $S_T$. Let
\begin{align*}
a_T&=\Bigg[TS_T^{-\tau+\tau\kappa}+p_T TS_T^{-r/2+r\kappa} +p_T\exp\left(-S_T^{1-2\kappa}\right)+ \exp\left(-S_T^{1-2\kappa}\right)\Bigg]
\end{align*}
The following theorem presents the uniform consistency, over all rescaled time points, of LC-Boost and LL-Boost. 
\begin{thm}{}
Let $\bm{x}^{*}_{t-h,T}$ denote a new predictor variable, independent of and with the same distribution as $\bm{x}_{t-h,T}$. Let $\kappa \in (0,1/2)$ be such that  $\kappa < \psi^{-1}-1$,
 Suppose that conditions \ref{ConditionA}, \ref{ConditionB}, \ref{ConditionD}, and \ref{ConditionE} hold. Then 
 \begin{itemize}
    \item[a.] on a set with probability at least $1-O(p_Ta_T)$, our LC-Boost estimate $\hat{F}^{(M_T)}_{lc}(\cdot,\cdot)$ satisfies: 
$\sup_{u \in [0,1]} E(|\hat{F}^{(M_T)}_{lc}(u,\bm{x}^{*}_{uT-h,T}) -\bm{\beta}'(u)\bm{x}^{*}_{uT-h,T}|^2)=o_p(1) \quad (T \rightarrow \infty)$ for some sequence $M_T \rightarrow \infty$ sufficiently slowly,
\label{theorem1} 
\item[b.] on a set with probability at least $1-O(p_Ta_T)$,  our LL-Boost estimate $\hat{F}^{(M_T)}_{ll}(\cdot,\cdot)$ satisfies  
$\sup_{u \in [0,1]} E(|\hat{F}^{(M_T)}_{ll}(u,\bm{x}^{*}_{uT-h,T}) -\bm{\beta}'(u)\bm{x}^{*}_{uT-h,T}|^2)=o_p(1) \quad (T \rightarrow \infty)$
for some sequence $M_T \rightarrow \infty$ sufficiently slowly,
\label{theorem2} 
\end{itemize}
\end{thm}

This is an extension of theorem 1 in \cite{Buhlmann2006} to the locally stationary time series setting with local constant or local linear least squares base learners. From the above theorems, we see the range of $p_T$ depends primarily on the moment conditions, the effective sample size $S_T$, and $\kappa$. For example, if we assume only finite polynomial moments with $r=q$ then, $p_T=o(S_T^{r/4-r\kappa/4}/\sqrt(T))$ for our estimates to be uniformly consistent over all rescaled time points. The $\sqrt{T}$ in the denominator is needed for uniform consistency, and can be replaced by $\sqrt{S_T}$ for pointwise consistency.  If we assume, sub-Gaussian or subexponential predictors for example we have $p_T=o(S_T^{\phi})$ for arbitrary $\phi >0$. This is the same range \cite{Buhlmann2006} obtained for iid sub-Gaussian predictors and errors. Given the $O(T^{-1})$ encountered when approximating a locally stationary process by a stationary distribution, we are unable to extend the theory to the ultra-high dimensional setting i.e $p_T=o(\exp(n^{c}))$ for $c<1$. 




We also provide results for the stationary time series with time invariant parameters. In this setting, we use the linear least squares base learner and use the entire sample for estimation. For the case of only a finite number of moments, the results in theorem \ref{theorem1} easily carry over to the stationary time invariant setting (i.e $\bm{\beta}(t/T)=\bm{\beta} \text{ } \forall t,T$), by letting $S_T=T$, and computing the relevant functional dependence measures. However, we can obtain a larger range for $p_T$, if we assume a stronger moment condition such as:

\begin{condition}\label{ConditionC} 
Assume the response and the covariate processes are stationary and have representations given in (\ref{form}). Additionally, assume $\upsilon_x=\sup_{q \geq 2} q^{-\tilde{\alpha}_x}\Phi_{0,q}^{\bm{x}} < \infty$ and $\upsilon_{\epsilon}=\sup_{q \geq 2} q^{-\tilde{\alpha}_{\epsilon}} \Delta_{0,q}^{\epsilon}< \infty$, for some $\tilde{\alpha}_{x},\tilde{\alpha}_{\epsilon} \geq 0.$
\end{condition}

Condition \ref{ConditionC} strengthens the moment condition \ref{ConditionB}, and requires that all moments of the covariate and error processes are finite. To illustrate the role of the constants $\tilde{\alpha}_{x}$ and $\tilde{\alpha}_{\epsilon}$, consider the example where $\epsilon_{t}=\sum_{j=0}^{\infty} a_j e_{t-j}$ with $e_i$ iid, and $\sum_{j=0}^{\infty} |a_{j}| < \infty$. Then $\Delta_{0,q}^{\epsilon}= ||e_{0}-e_{0}^{*}||_{q}\sum_{j=0}^{\infty}|a_{j}|$. Now if we assume $e_0$ is sub-Gaussian, then $\tilde{\alpha}_{\epsilon}=1/2$, since $||e_{0}||_{q} = O(\sqrt{q})$, and if $e_i$ is sub-exponential, we have $\tilde{\alpha}_{\epsilon}=1$.

The following corollary states the corresponding results for the stationary time series setting. We define $\tilde{\psi}=\frac{2}{1+2\tilde{\alpha}_x+2\tilde{\alpha}_{\epsilon}}, \tilde{\varphi} = \frac{2}{1+4\tilde{\alpha}_x}$, and let
\begin{align*}
b_T&= \left[ \exp\left(-\frac{T^{1/2-\kappa}}{\upsilon_x\upsilon_{\epsilon}}\right)^{\tilde{\psi}}+ p_T\exp\left(-\frac{T^{1/2-\kappa}}{\upsilon_x^2}\right)^{\tilde{\varphi}}\right],
\end{align*} and let $\hat{F}^{(M_T)}(\bm{x}_t)$ denote our $L_2$ boosting estimate for $Y_t$, we then have: 
\begin{cor}{}
Let $\kappa \in (0,1/2)$, and $\bm{x}^{*}_{T-h}$ denote a new predictor variable, independent of and with the same distribution as $\bm{x}_{T-h}$.  Suppose conditions \ref{ConditionA},  \ref{ConditionE}, and \ref{ConditionC} hold. Then on a set with probability at least $1-O(p_Tb_T)$, we have that our $L_2$ Boosting estimate $\hat{F}^{(M_T)}(\cdot)$ satisfies:
$$E(|\hat{F}^{(M_T)}(\bm{x}^{*}_{T-h}) -\bm{\beta}\bm{x}^{*}_{T-h}|^2)=o_p(1) \quad (T \rightarrow \infty).$$

\label{cor1} 
\end{cor}

In the stationary setting our theorems improve upon previous results found in \cite{lutz2006} by providing a more detailed and larger range for $p_T$. This is largely due to using different concentration inequalities; we rely on sharp Nagaev type and exponential inequalities for the case of polynomially and exponentially decaying tails respectively. Whereas \cite{lutz2006} relies on a Markov type inequality after bounding the $r^{th}$ absolute moment of the sum. For example, assuming sub-Gaussian predictors and errors we obtain $p_T=o(\exp(T^{\frac{1-2\kappa}{3}}))$, and $p_T=o(\exp(T^{\frac{1-2\kappa}{5}}))$ for subexponential predictors and errors. As a comparison \cite{lutz2006} obtained $p_T=o(T^{\phi})$, for arbitrary $\phi>0$, when applying $L_2$ boosting for stationary sub-Gaussian time series.

\section{Simulations}
\label{simulations}

In this section, we evaluate the forecasting performance of our algorithms in a finite sample setting. 
Let $Y_{t,T}$ denote our response, and let  $\bm{x}_{t-1,T}=(Y_{t-1,T},\ldots,Y_{t-3,T},\bm{z}_{t-1,T},\ldots,\bm{z}_{t-3,T})$ represent our potential set of predictors, where $\bm{z}_{t-1,T}\in \mathcal{R}^{d_T}$ represents our $d_T$ exogenous series at time $t$. We fix $T=200$, and $d_T=100$, giving us $p_T=303$ potential predictors. We consider 14 DGPs and our general model is, for $t=1,\ldots T$,
\begin{eqnarray*}
Y_{t,T}&=&\rho Y_{t-1,T}+ \sum_{j=1}^4 (b+\beta_{j}(t/T))z_{j,t-1,T} + \epsilon_{t}\\
\bm{z}_{t,T} &=& A(t/T)\bm{z}_{t-1,T}+\bm{\eta}_t
\end{eqnarray*}
and it is assumed that $\rho=.6$, $b=0.5$. For DGPs 1-12 we let $A(t/T)=\{.4^{|i-j|+1}\}_{i,j \leq d_T}$, and for DGPs 13 and 14 we let $A(t/T)=(1-t/T) A_1 + (t/T) A_2$, where the matrices $A_1=\{.2^{|i-j|+1}\}_{i,j \leq d_T}, A_2=\{.4^{|i-j|+1}\}_{i,j \leq d_T}$. Define \textsc{lgt}$(\gamma,c,t/T)= (1+\exp(-\gamma(t/T-c))^{-1}$, time variation in the coefficients is modeled as follows:

\begin{center}
\small
\begin{tabular}{l|l|llll|llllll}
    DGP & Description   & $\beta_{1}(t/T)$ & $\beta_{2}(t/T)$ &  $\beta_{3}(t/T)$ & $\beta_{4}(t/T)$ & $\bm{z}_{t,T}$\\ \hline
   1 &   constant & 0 & 0 &0  &0 &stationary   \\
   2 &   break in error variance & 0 & 0 &0  &0 &stationary   \\
   3 &   early break, $T_b=50$   & \multicolumn{4}{c|}{$-1(t>T_b)$} &stationary    \\
   4 &   mid  break, $T_b=100$    & \multicolumn{4}{c|}{$-1(t>T_b)$} & stationary  \\
   5 &   late break, $T_b=150$  & \multicolumn{4}{c|}{$-1(t>T_b)$} & stationary \\
   6 & small random walk  & \multicolumn{4}{c|}{$\Delta \beta_{j}(t/T)\sim N(0,\frac{.5}{\sqrt{T}})$} &  stationary\\
   7 & big random walk  & \multicolumn{4}{c|}{$\Delta \beta_{j}(t/T)\sim N(0,\frac{1}{\sqrt{T}})$} &  stationary\\

   8 & smooth, $c=.25$  & \textsc{lgt}(10,c) & \textsc{lgt}(5,c) & \textsc{lgt}(20,c) & \textsc{lgt}(10,c) & stationary\\
   9 & smooth, $c=.75$  & \textsc{lgt}(10,c) & \textsc{lgt}(5,c) & \textsc{lgt}(20,c) & \textsc{lgt}(10,c) & stationary\\
   10 & smooth, $c=.90$  & \textsc{lgt}(10,c) & \textsc{lgt}(5,c) & \textsc{lgt}(20,c) & \textsc{lgt}(10,c) & stationary\\
   11 & steep & $-.3(\frac{t}{T})^2$ & $(\frac{t}{T})^2$ & $ -.4(\frac{t}{T})$ & $\frac{t}{T}$ & stationary\\
   12 & exotic & 0 & 0 & 3$\cos(\frac{2\pi t}T)$ & $2\frac{t}T\sin(2\pi \frac{t}{T})$ & stationary\\
   13& smooth, $c=.75$  & \textsc{lgt}(10,c) & \textsc{lgt}(5,c) & \textsc{lgt}(20,c) & \textsc{lgt}(10,c) & locally stationary\\
   14 &   late break, $T_b=150$  & \multicolumn{4}{c|}{$-1(t>T_b)$} & locally stationary \\
   \hline
\end{tabular}
\end{center}

For all DGPs, we report results when generating the innovations as $\bm{\eta}_t \stackrel{iid}{\sim} N(0,I_{d_T})$ or from a $t_{5}(0,3/5*I_{d_T})$. For DGP 2, we have a break in the error variance: $\epsilon_t=D(0,1)(t<150)+D(0,2.5)(t \geq 150)$, where the distribution $D$ is either a normal or $t_5$ distribution. For the remaining DGPs $\epsilon_t \stackrel{iid}{\sim} N(0,1)$ or $\stackrel{iid}{\sim} t_5$.

\subsection{Methods and Forecast Design}

We consider the forecasting performance of the following methods:
\begin{itemize}
\item LC-Boost, LL-Boost.
\item $L_2$ Boosting, with time invariant coefficients estimated on the full sample.
\item Lasso, AR(3) both estimated on the full sample.
\item Rolling window $L_2$ Boosting, Rolling window AR(3) model both with window length $T/5$.
\end{itemize}  
AR and rolling window AR models are commonly used benchmarks in macroeconomic forecasting. Models estimated on the full sample assume time invariant parameters, or more generally assume the time variation is small. Estimation using LC-Boost, LL-Boost or a rolling window approach involves using a subsample of the data leading to a bias-variance tradeoff. Due to this tradeoff, methods accounting for time variation are not guaranteed to outperform their time invariant counterparts in a finite sample setting.

All boosting models are computed using the R package \textbf{mboost}, and the lasso model is computed using the R package \textbf{glmnet}. For LC-Boost and LL-Boost, we use the uniform kernel and we estimate the bandwidth via the cross validation procedure described in section \ref{hyperparameter}, with $\omega=20$, and $B=[.3,.4,\ldots,1]$. The number of steps in all boosting models is determined using AIC with the maximum number of steps set to $M_{upp}=100$. Lastly, the penalty parameter in the Lasso model is estimated using the BIC statistic.

For each simulation, and for all methods, we forecast $Y_{T,T}$ and compute the out of sample forecast error, which is then averaged over 1000 simulations. Specifically, for a given simulation, let $\hat{Y}_{T,T}^{(k)}$ represent the out of sample forecast of $Y_{T,T}$. We then compute MSFE=$\frac{1}{1000}\sum_{k=1}^{1000}(\hat{Y}_{T,T}^{(k)}-Y_{T,T}^{(k)})^2$, for each method. We report this MSFE relative to the MSFE obtained from $L_2$ boosting model with time invariant coefficients.

\subsection{Results}

The results for Gaussian innovations are in table \ref{table1}. The results for $t_5$ innovations are contained in the appendix. We first dicuss results for the Gaussian case. DGP 1 and 2 contain time invariant coefficients, with DGP 2 having a structural break in the variance of the noise. In both these DGPs using the full sample yields the best estimator. LC-Boost only has a minor error inflation compared to using the whole sample, whereas LL-Boost does worse than LC-Boost in this setting. The under performance of LL-Boost vs LC-Boost in these settings is likely due to the bias-variance tradeoff when using local linear vs local constant methods. If the time variation is non-existent or mild, as is the case here, the additional variance incurred by estimating more parameters can cancel out any benefit obtained from bias reduction. DGP $3,4,5$ all contain a discrete structural break, and we see that both LC-Boost and LL-Boost outperform other methods. When the structural break occurs near the end of the sample, LL-Boost has large gains over LC-Boost. 

DGP 6 has a slowly varying random walk, and we observe that LC-Boost and LL-Boost perform slightly better than using the full sample.  DGP 7 has larger time variation in the coefficients, and we see that LL-Boost and LC-Boost easily outperforms the other methods. DGP 8, 9 and 10 have smooth transition logistic functions, where $c$ is the analogous to the breakpoint in a discrete break model, and  $\gamma$  represents smoothness of the transition.\footnote{We note that setting $\gamma$ to infinity results in a discrete break model.} Out of the three DGPs, time varying methods perform best when $c=.75$, with the performance deteriorating in the other two cases as the time variation occurs either too close to the forecast date or too far away. DGP 11 and 12 contain coefficient functions which are highly non-linear, and LL-boost shows very large improvements vs LC-Boost. DGP 13 and 14 show that adding locally stationary predictors leads to only slight change in the results vs DGP 9 and 5 respectively.

When we have $t_5$ innovations, the results generally follow the conclusions stated earlier, except the improvements are noticeably smaller in many cases. The presence of additional noise in the data likely impacts our method in two ways: due to the additional noise in the data, the bias-variance tradeoff is less favorable to using a subset of the full sample. Additionally, the noise in the data makes the cross validation error estimate less reliable, leading to errors in estimating the optimal bandwidth parameter.\footnote{We also repeated each of the simulations using the Gaussian kernel instead of the uniform kernel. In general we found very similar performance between the two kernels. For the case of $t_5$ innovations and little to no time variation in the coefficients, we found the Gaussian kernel was more effective for LL-Boost. Given the close similarities between the kernels, we omit the results.}
\\
The results suggest the following conclusions:
\begin{enumerate}
    \item When the time variation in the coefficients is non-existent or minor, using the full sample often gives the best performance. The performance of LC-Boost is only marginally weaker than using the full sample, while the performance of LL-Boost takes a more significant hit. 
    \item LL-Boost and LC-Boost forecasts both seem to underperform forecasts using the full sample when there is a break  in the conditional variance rather than the conditional mean.
    \item Using LL-Boost leads to large improvements in forecasting performance vs LC-Boost when we have significant time variation in the coefficients. This is especially true when the time variation occurs closer to the forecast date and/or the coefficient functions are highly non-linear.  
    
    \item Time varying methods are likely to be less useful when we have a low sample size coupled with high noise. Some of the difficulties in this setting may be overcome by selecting the bandwidth parameter using a larger validation set along with a finer grid of bandwidth values.
\end{enumerate}

\begin{table}[]
\centering
\small
\caption{Relative MSFE, Gaussian Innovations}
\label{table1}
\begin{tabular}{|c|c|c|c|c|c|c|}
\hline
DGP & AR (3) & Rolling AR (3) & Rolling Boost & LC-Boost & LL-Boost & Lasso \\ \hline
1   & 2.22   & 2.41           & 1.92          & 1.05     & 1.19     & 1.06  \\ \hline
2   & 1.79   & 1.79           & 1.42          & 1.08     & 1.16     & 1.23  \\ \hline
3   & 1.16   & 1.24           & 1.01          & .61      & .67      & 1.14  \\ \hline
4   & .91    & .98            & .80           & .55      & .58      & 1.02  \\ \hline
5   & .72    & .78            & .63           & .76      & .53      & .92   \\ \hline
6   & 5.25   & 5.93           & 1.62          & .91      & .90      & 1.06  \\ \hline
7   & 3.47   & 3.75           & .96           & .68      & .60      & 1.06  \\ \hline
8   & 4.92   & 5.30           & 1.53          & .59      & .56      & 1.13  \\ \hline
9   & 1.94   & 2.08           & .73           & .52      & .35      & 1.20  \\ \hline
10  & 1.77   & 1.88           & .95           & .79      & .53      & 1.22  \\ \hline
11  & 2.81   & 3.10           & .99           & .75      & .61      & 1.15  \\ \hline
12  & 1.04   & 1.09           & .32           & .63      & .16      & 1.15  \\ \hline
13  & 2.11   & 2.20           & .83           & .63      & .39      & 1.20  \\ \hline
14  & .73    & .78            & .67           & .80      & .52      & .97   \\ \hline
\end{tabular}
\end{table}

\section{Application to Macroeconomic Forecasting}
\label{Empirical}

As discussed in the introduction, the parameter instability of various macroeconomic series has long been established in the econometrics literature. Some examples include \cite{Stock96,stock2009forecasting,breitung2011testing}, all of which find instability in either the univariate relationship of a large number of series or in the factor loadings of a dynamic factor model of a large panel of macroeconomic series. Similarly, \cite{stock2003forecasting} and \cite{rossi2010have} have found evidence of instability in the predictive ability of various series in forecasting output and inflation. However, the question of whether point forecasts can be improved by modeling parameter instability, especially when using high dimensional predictors, is far less clear. 

Proponents of modeling parameter instability include works such as \cite{Clements96,clements2000forecasting} which argue that ignoring these instabilities are the main sources of forecast breakdowns. On the other hand, empirical evidence in favor of ignoring instabilities include \cite{Stock96} which had shown there is little benefit to modeling time variation in a wide range of autoregressive and bivariate forecasts, and \cite{kim2014forecasting,koop2013forecasting} which showed forecasts estimated by recursive estimation (using the full sample) performed as well as or better than rolling window forecasts for a range of models estimated from a large panel of macroeconomic series. 
Additionally, a number of works such as \cite{pettenuzzo2017forecasting,koop2013large,eickmeier2015classical}, have estimated TVP models using Bayesian methods and their results suggest that TVP models offer only minor improvements in the accuracy of point forecasts when compared to low dimensional constant parameter models.\footnote{These works did find TVP models produced larger improvements to density forecasts. We note that works such as \cite{koop2012forecasting,groen2013real,chan2012time} have also estimated TVP models using Bayesian methods and found significant improvements to point forecasts when compared to a low dimensional constant parameter benchmark. However, these works are restricted to forecasting inflation with low dimensional predictors. Additionally, \cite{carriero2019large,petrova2019quasi} analyzed Bayesian TVP models with high dimensional predictors and found large improvements to density forecasts and small improvements to point forecasts when comparing against a high dimensional time invariant parameter model.} Lastly, on the theoretical side, \cite{bates2013consistent} has shown the standard principal components estimator remains consistent even in the presence of ``small" breaks and/or mild time variation in the factor loadings of a dynamic factor model.

To illustrate the difficulty of exploiting parameter instability, consider a simple example where there is a single discrete structural break in the forecasting model. Even if the researcher knew the precise date of the break and decided to use only post break observations for estimation there is a bias-variance trade off in using less data for estimation \citep{pesaran2007selection}. Therefore in the presence of small instabilities, such as small breaks or very slowly varying coefficients, using the entire sample through recursive estimation can be more beneficial than using only a subset of the data. Due to this bias-variance tradeoff and the uncertainty around the precise nature of time variation, the majority of works on macroeconomic forecasting tend to use the full sample available when forecasting. Furthermore, these issues are more severe when using high dimensional predictors.

Given the above discussion, we use the methods developed in this paper to answer a number of questions such as:
\begin{itemize}
    \item Does modeling parameter instability improve macroeconomic forecasts?
    \item Which models are best able to deal with underlying parameter instability?
    \item Which variables and forecast horizons benefit most from the use of time varying parameter models?
    \item During which time periods do time varying methods perform best?
\end{itemize}


To answer these questions, we use the August 2018 (2018:8) vintage of the FRED-MD database which contains 128 monthly macroeconomic series collected from a broad range of categories. See \cite{mccracken2016fred} for a more detailed description of each series, as well as transformations needed to achieve approximate stationarity.\footnote{We depart from the recommended transformations for the housing series (Group 4) which we treat as I(1) in logs.} We remove 5 series which contain large amounts of missing values, leaving us with 123 monthly macroeconomic series which run from January 1960 to August 2018. We focus our analysis on 8 major macroeconomic series: Industrial Production (IP), Total Nonfarm Payroll (PAYEMS), Unemployment Rate (UNRATE), Civilian Labor Force (CLF), Real Personal Income Excluding Transfer Receipts (RPI), Consumer Price Index (CPI), Effective Fed Funds Rate (FF), and Three Month Treasury Bill (TB3MS). For each series, we compare the out of sample forecasting performance of several models at the $h=1,3,6,12$ month forecasting horizons. 

\subsection{Methods and Forecast Design}

For all the methods we consider, let $Y_{t,T}^{h}$ denote our $h$-step ahead target variable to be forecast. As an example, for CPI our target variable is  $Y_{t,T}^{h}=\frac{1200}{h}log(\frac{CPI_{t}}{CPI_{t-h}})$, and we define the target similarly for the rest of the series except FEDFUNDS and TB3MS which are modeled as $I(1)$ in levels (i.e. $Y_{t,T}^{h}=\frac{12}{h}(\text{FEDFUNDS}_{t}-\text{FEDFUNDS}_{t-h}$)). Next let $\bm{z}_{t-h,T}$ denote the rest of our 122 predictor series at time $t-h$, and let $\bm{x_{t-h}}=(Y_{t-h,T},\ldots,Y_{t-h-3},\bm{z_{t-h,T}},\ldots,\bm{z_{t-h-3,T}})$ where $Y_{t-h,T}=Y_{t-h,T}^{1}$.

For all time varying methods we estimate the bandwidth using the cross validation procedure detailed in section \ref{hyperparameter}. For selecting the bandwidth we use a grid of values from .3 to 1 with increments of .025 i.e. $B=[.3,.325,.\ldots,1]$, and we use the last $\omega=60$ observations as our validation set.\footnote{For all local constant methods we report results using the uniform kernel, the results are very similar if we use the Gaussian kernel. For local linear methods we use the Gaussian kernel.} Additionally, we estimate all models under consideration using time invariant methods in order to assess the benefits of directly modeling time variation. We evaluate the forecasting performance of the following methods: 
\begin{table}[!htbp]
\label{methods}
\small
\centering
\begin{tabular}{l|l|lll}
   Method  &  Parameter &   Predictors considered\\ \hline
   AR  & time invariant & $(Y_{t-h,T},\ldots,Y_{t-h-3})$\\
   TV-AR & local constant & $(Y_{t-h,T},\ldots,Y_{t-h-3})$\\
   Boost & time invariant & $\bm{x_{t-h}}$\\
  Lasso & time invariant & $\bm{x_{t-h}}$ \\
  LC-Boost & local constant & $\bm{x_{t-h}}$\\
  LL-Boost & local linear & $\bm{x_{t-h}}$ \\
  LC-Boost-Factor & local constant & $(Y_{t-h,T},\ldots,Y_{t-h-3},\bm{F_{t-h,T}},\ldots,\bm{F_{t-h-3,T}})$\\
  LL-Boost-Factor & local linear & $(Y_{t-h,T},\ldots,Y_{t-h-3},\bm{F_{t-h,T}},\ldots,\bm{F_{t-h-3,T}})$\\
  DI & time invariant &$(Y_{t-h,T},\ldots,Y_{t-h-3},\bm{F_{t-h,T}})$\\
  Boost Factor & time invariant & $(Y_{t-h,T},\ldots,Y_{t-h-3},\bm{F_{t-h,T}},\ldots,\bm{F_{t-h-3,T}})$\\
  \hline
\end{tabular}
\end{table}




The last four methods are of the following form: 
\begin{align}
    Y_{t,T}^{h}=\alpha(t/T)+\sum_{j=0}^{3} \alpha_{j}(t/T)Y_{T-h-j,T}+\sum_{j=0}^{l}\bm{\beta}_{j}'(t/T)\bm{F}_{t-h-j,T}+\epsilon_{t}, \label{htvp3}
\end{align}
where $\bm{F}_{t-h,T}=(F_{1,t-h,T},\ldots, F_{k,t-h,T})$ is a $k$-dimensional vector of factors which are estimated using the principal components of our 122 predictor series $\bm{z}_{t-h,T}$. We ignore possible time variation in our predictors when estimating our factors, and rely on results showing the consistency of the principal components estimator under mild time variation and structural breaks in the factor loadings \citep{bates2013consistent}. We instead focus on modeling the time variation in the coefficients of the forecasting equation (\ref{htvp3}). As an example, for LC-Boost Factor we set $k=8$, $l=3$ and estimate the model using our LC-Boost algorithm. And for DI we set $k=4$, $l=0$ and estimate the model assuming time invariant coefficients and utilizing the full sample.\footnote{Additionally, \cite{stock2009forecasting} conjectured, for macroeconomic data, that the time variation in the coefficients $\bm{\beta}(t/T)$ is far more important than possible time variation in the factors. Their empirical results showed in-sample estimates of the factors as well in-sample forecasting results were little changed by allowing for a one time break in the factors.} 


\begin{rmk}
We note that constant parameter versions of high dimensional methods (e.g. Boost, Boost Factor, Lasso) have greater adaptability to time variation than low dimensional regressions which assume the set of relevant predictors/factors is fixed over time. The idea is that by combining information from a large set of predictors our forecasts are more robust to instabilities which occur in a specific predictor's forecasting ability.\footnote{Empirical evidence of this was provided in \cite{carrasco2016sample}.} When combined with a recursive window forecasting scheme, these methods indirectly capture at least some of the time variation present in the data. 
\end{rmk}


We use an expanding (recursive) window scheme designed to simulate real time forecasting. Our out of sample forecasting period starts in 1971:9 and ends in 2018:8 for a total of 564 months ($47$ years). To construct the first forecast of time $t$=1971:9 we estimate the factors, the coefficients, and select the hyperparameters using data available only until time 1971:9-$h$. We then expand our window by one observation and estimate the forecast of time $t+1$=1971:10 using information available until time $t-h+1$, and so on until we reach the end of our sample.

\section{Results}
\label{results}

Our benchmark model for all series and forecasting horizons is an AR(4) model with time invariant parameters. Due to space considerations we report some of our results in the appendix. We start by giving an overview of the results for the full out sample period, which are reported in table \ref{h=12} for $h=12$ and in the appendix for $h=6,3,1$, before analyzing how performance varies over time. For the time varying methods we observe the following: the TV-AR model fails to improve upon the benchmark AR model for the vast majority of series and forecast horizons, confirming the results of \cite{Stock96} on an expanded sample. Out of our four time varying Boosting methods, LC-Boost Factor appears to perform best. LC-Boost Factor outperforms the benchmark for all series and forecast horizons, it also performs best, out of all models, the majority of times. In contrast, our LL-Boosting methods appear to perform poorly relative to LC-Boost.\footnote{We omit the performance of LL-Boost as it was outperformed by both LL-Boost Factor and both LC-Boost methods.} Given the results in section \ref{simulations}, this suggests that the parameters as a function of time may not be sufficiently curvy enough for local linear methods to benefit. For time invariant methods we observe the following:  Boost-Factor and Boost performs similarly and generally outperform DI and Lasso models.

Comparing across forecast horizons: we observe that, for all high dimensional methods, improvements to the benchmark are greater as we increase our forecast horizon. For $h=1$, many of the methods appear to perform similarly, with Boost Factor and LC-Boost Factor appearing to perform best. For longer forecast horizons, LC-Boost Factor is the best performing model the majority of the time, with the gap between LC-Boost Factor and its competitors widening as we increase the forecast horizon. Additionally, the benefits to modeling time varying parameters are more apparent at longer forecast horizons.

\subsection{Analyzing Performance Over Time} 

Relying only on the aggregate performance of a model over the entire out of sample period can hide many important details and lead to misleading conclusions. We rely on two methods to analyze how performance varies over time; the first is to plot the MSFE as a function of the start date for the out of sample forecasting period. More specifically, let $T_1$ denote the start forecast date, then for a given method $i$ and horizon $h$, we calculate
\begin{align}
    MSFE^{h}_{(i)}(T_1,T_2)=\frac{\sum_{t=T_1}^{T_2}\hat{\epsilon}^{2}_{t,(i)}}{\sum_{t=T_1}^{T_2}\hat{\epsilon}^{2}_{t,(AR)}}, \quad \quad \text{with} \quad T_2=2018:8.
    \label{MSFE_startdate}
\end{align}
The second method is to analyze the forecasting performance over three important subperiods. The first subperiod, which we refer to as ``Pre-Great Moderation", consists of 136 observations, 1971:9-1982:12, and corresponds roughly to the period before the start of the ``Great Moderation". The second subperiod is from 1983:1-2006:12, and corresponds roughly to the ``Great Moderation", a period where the volatility of a large number of macroeconomic series was significantly reduced \citep{SW2003}. The third subperiod is from 2007:1-2018:8, which we refer to as ``Post Great Moderation", covers the period right before the great recession and takes us to the end of our sample.  

For the first method, we let $T_1$ vary from 1971:9 until 2006:12. We plot the MSFE by $T_1$ for the top 5 performing methods: LC-Boost Factor, LC-Boost, Boost, Boost Factor, and DI. The figures \ref{MSFE12}-\ref{MSFE1}; contain the results for horizons $h=12,6,1$ respectively.\footnote{The corresponding results for $h=3$ are reported in the appendix.} Looking at figures \ref{MSFE12}-\ref{MSFE1}, we see that LC-Boost Factor is easily the best performing method for horizon $h=12,6$, and to a lesser extent $h=1$. Comparing across all horizons, we notice:
\begin{itemize}
    \item The performance improvements for LC-Boost factor, relative to its time invariant counterparts, are more apparent as we increase the forecast horizon.
    \item As we increase $T_1$, the gap between LC-Boost Factor and the time invariant methods widens. In particular we notice a large separation in performance starting during great moderation period.
\end{itemize}
Additionally, we also observe that the commonly used DI model loses much of its predictive ability during the Great moderation and performs worse than the benchmark for about half of the series. This result suggests that DI gained most of its predictability vs the benchmark during the ``Pre-Great Moderation" period.  




Table \ref{h=12} contain the results for each of the subperiods for horizon $h=12$; the corresponding results for $h=6,3,1$ are found in the appendix. For each subperiod we report the MSFE, relative to the MSFE of the benchmark AR(4) model. We start with the ``Pre-Great Moderation" period, and note that with the exception of TV-AR models, all other models strongly outperform the benchmark model the majority of the time during this period. Time invariant methods such as Boost Factor and DI models perform best, and their performance is strongest when forecasting at longer horizons. LC-Boost factor appears to be slightly lag behind these two methods during this time period. As we enter the ``Great Moderation" period, the performance of all models generally declines relative to the AR benchmark. In particular, time invariant methods such as DI and Boost take a large hit and underperform the benchmark in many cases, especially for $h=12$. LC-Boost Factor undergoes a much smaller decline compared to the rest of the models, and emerges as the best performing model during this time period. Importantly, we also observe that LC-Boost performs at the same level or worse than its time invariant counterpart Boost in the majority of cases. This suggests that although there seems to be a large amount of time variation in this period, the bias variance tradeoff in modeling it is not favorable to a model with a large amount of potential predictors ($\sim 500$ predictors). During the ``Post Great Moderation" period the performance of LC-Boost methods show large improvements relative to the benchmark AR model, and relative to their time invariant counterparts, for all forecast horizons, with the improvement being greatest for longer horizons.

\subsection{Assessing Benefits of Modeling Time Varying Parameters}
\label{Benefits}

In order to assess the benefits of \textit{directly} modeling parameter instability, we compare the performance of LC-Boost Factor vs Boost-Factor. These also happen to be the best time varying and time invariant methods respectively. We start our analysis by first plotting the MSFE of LC-Boost Factor, \textit{relative to} the MSFE of Boost Factor, as a function of the start date for the out of sample forecast period, i.e. $MSFE_{(LCBoost Factor)}(T_1,T_2)/MSFE_{(Boost Factor)}(T_1,T_2)$. The results are in figure \ref{Benefits_TVP}. We observe that for all series and forecast horizons LC-Boost Factor almost never performs worse than Boost Factor, and outperforms it the vast majority of the time. Furthermore, the gap between the two methods widens as we increase $T_1$, the start date of the out of sample period, and as we increase $h$, the forecast horizon. For example, if we consider horizon $h=12$, and we start the out of sample period in the early 1990's, LC Boost offers, on average, over a 20 percent improvement over Boost-Factor. We observe similar patterns, although the improvements are not as large ($\sim$ 10-15 percent on average), for horizons $h=3,6$. An exception seems to be for $h=1$, which shows little improvements for the majority of series, with the exceptions coming from the two interest rate series and EMS. 

Next, we attempt to get a finer look at how the benefits of modeling parameter instability vary over time. We first define the \textit{local} MSFE (L-MSFE), of method $i$ at time $t_0$ as :
\begin{align*}
    \text{L-MSFE}_{i}(t_0)=\frac{\sum_{t=t_0-\Delta}^{t_0+\Delta}\hat{\epsilon}_{t,(i)}^{2}}{\sum_{t=t_0-70}^{t_0+\Delta}\hat{\epsilon}_{t,(AR)}^{2}}, \quad\quad \text{RL-MSFE}_{i}(t_0)=\frac{\text{L-MSFE}_{i}(t_0)}{\text{L-MSFE}_{\text{BoostFactor}}(t_0)}
\end{align*}
with the convention that $\hat{\epsilon}_{t,(i)}=0$ for  $t\leq 0, t\geq T$. This amounts to using a uniform kernel to weight the forecast errors with a bandwidth chosen such that the window size has $\Delta=70$ observations. We then plot $\text{RL-MSFE}_{i}(t_0)$ for $i=$\textsc{LCBoostFactor}, for $t_0=\text{1977:3},\ldots,\text{2012:10}$ for all series and forecast horizons. The endpoints are chosen so that the first and last values in the plot correspond to the $\text{RL-MSFE}$ during the ``Pre-Great Moderation" and ``Post-Great Moderation" periods respectively. 

The results are in figure \ref{localMSFE1}, and we observe that the first value is usually near or above one for all variables except for CLF (Civilian Labor Force). This suggests that during the ``Pre-Great Moderation" there seems to be little or no benefit to modeling time variation. This can reflect either a lack of underlying parameter instability during this time period, or the relatively low sample size available combined with high volatility made it difficult to exploit the time variation present. During the Great Moderation period, almost all series experience large declines in RL-MSFE, with the exact timing of the decline differing by series. For IP and RPI the benefits to modeling parameter instability appear to decrease from their mid 1990's levels, while for the rest of the series we see further improvements until the end of the sample. These results suggest that there is a large amount of parameter instability which started during the Great moderation period and continued though the sample.

Lastly, we attempt to examine the degree and timing of time variation by examining the bandwidth values selected. We define the local bandwidth of LC-Boost Factor as: $\text{L-BW}(t_0)=\sum_{t=t_0-\Delta}^{t_0+\Delta}\hat{b}_{t_0}/(2\Delta)$, with the convention that $\hat{b}_{t_0}=0$ for  $t\leq 0, t\geq T$. Recall that $\hat{b}_{t_0}$ is the bandwidth chosen at time $t_0$. Since we are using the uniform kernel, $\hat{b}_{t_0}$ represents the fraction of the sample \textit{available at time $t_0$} that we are using for estimation. As as example, at time $t_0$=1977:3 we have a total of 196 observations available for estimation, therefore a value of $\hat{b}_{t_0}=.61$ for $t_0$=1977:3 implies we are using $\approx$ 120 observations. We set $\Delta=70$ observations, and then plot $\text{L-BW}(t_0)$ for $t_0=\text{1977:3},\ldots,\text{2012:10}$ for all series and forecast horizons. As an additional comparison we also plot the local bandwidth implied by a rolling window estimator which uses a fixed window length of 120 observations. We also note that $b_{t_0}=1$ for Boost Factor at all time points, since it uses the entire sample available at time $t_0$.

The results are seen in figure \ref{localBW}, and we notice that for the pre Great Moderation period the local bandwidths are usually between .7-.8 for most series. As we enter the Great Moderation we notice that the local Bandwidths generally tend to increase initially before declining. However, we notice the timing and degree of declines differs by series. For some series such as IP and RPI, the local BW tends to increase after reaching their lows in the mid 1990s, whereas for other series such as UNRATE, FEDFUNDS, and TB3MS the local BW start their decline in the 1990s. In contrast, we see that using a fixed rolling window of 120 observations implies a monotonically decreasing bandwidth and assumes the same bandwidth regardless of series or horizon.\footnote{Given our definition of the bandwidth, the bandwidth for the rolling window estimator is $\hat{b}_{t_0,rolling}=120/(\text{sample size available at time } t_0)$; as $t_0$ increases this bandwidth declines monotonically.} To determine the importance of estimating the optimal bandwidth via cross valiation, we compare the local MSFE of LC-Boost Factor to Boost Factor estimated using a 120 observation rolling window in the appendix. The results show that for the vast majority of series and horizons the rolling window estimator is strongly outperformed by LC-Boost Factor with the largest gains occurring during the Great Moderation period.


Overall our results suggest the following conclusions:
\begin{enumerate}[label=\textbf{\arabic*})]
    \item Parameter instability starts to appear around the beginning of the Great Moderation period. This instability seriously deteriorates the relative forecasting performance of time invariant methods; with the effect being more severe for longer horizon forecasts.
    \item Due to the large improvements in point forecasts from our methods, it is likely that this instability has a substantial impact on the conditional mean as well as the conditional variance of various economic series.
    \item The commonly used rolling window estimation method can understate the benefits of modeling parameter instability by failing to account for differences in the degree of parameter instability by series, forecast horizon, and time period. 
    \item Lastly, there are large benefits to modeling parameter instability if done properly. Given the high bias variance tradeoff encountered in using a reduced sample size, these benefits can easily be missed. For example, models such as LC-Boost have more difficulty in learning the time variation in the data due to the large amount of potential predictors.
\end{enumerate}

To elaborate more on point \textbf{4)} above, we also compare the L-MSFE of the following models in the appendix: LL-Boost Factor vs LC-Boost Factor, LC-Boost vs LC-Boost Factor, and LC-Boost vs Boost. We see from the results that LL-Boost Factor was strongly outperformed by LC-Boost Factor in the earlier parts of the sample, suggesting that there was little time variation during the pre-great Moderation period. As our sample size available for estimation increases and as the time variation starts to become more apparent, we see the performance of LL-Boost factor improve to the point where it does as well as or outperforms LC-Boost factor in about half of the series, especially for longer horizons. Compared to LC-Boost Factor, we observe that the benefits of modeling time variation via LC-Boost are smaller and are realized far later in the sample. As an example, for $h=12$ we notice that LC-Boost performs worse than Boost during most of the Great moderation period. Additionally, for many of the series, the improvements of LC-Boost over Boost start to occur near the end of the great moderation period. In contrast, LC-Boost Factor is able to adapt to the time variation far earlier as a result of having a more favorable bias variance tradeoff. 


\section{Conclusion}
\label{disucussion}

In this work, we have presented two $L_2$ Boosting algorithms for estimating high dimensional predictive regressions with time varying coefficient. We proved the consistency of both of these methods, and showed their effectiveness in modeling the parameter instability present in macroeconomic series. Compared to other TVP methods, our methods are very efficient computationally even for high dimensional data; a single LC-Boost forecast, including implementing the cross validation procedure, can be estimated within a matter of seconds. Additionally, they can be implemented by researchers and practitioners using the easy to use R package \textbf{mboost}. Furthermore, the boosting framework can be easily adapted to fitting more complex non-linear models.

There are many topics available for further study, one such topic is in selecting the important bandwidth parameter for our models. Although our cross validation procedure seems to perform adequately, we welcome further improvements to this methodology. Lastly, although our empirical example focused on forecasting, our models are applicable in a far broader range of settings.

\clearpage

\begin{table}[]
\small
\centering
\caption{Relative MSFE $h=12$}
\label{h=12}
\begin{tabular}{|c|c|c|c|c|c|c|c|c|}
\hline
\multicolumn{9}{|c|}{Full Out of Sample Period 1971:9-2018:8}                                              \\ \hline
                & IP   & PAYEMS & UNRATE & CLF & RPI & CPI & FF   & TB3MS \\ \hline
TV-AR           & 1.04 & .99    & 1.1    & .64 & .99 & .86 & 1    & 1.07  \\ \hline
DI            & .79  & .79    & .69    & .99 & .84 & .94 & .87  & .91   \\ \hline
Lasso           & .77  & .81    & .75    & .77 & .93 & .96 & .76  & .89   \\ \hline
Boost           & .78  & .73    & .73    & .74 & .85 & .88 & .79  & .88   \\ \hline
Boost Factor    & .75  & .81    & .62    & .96 & .80 & .89 & .78  & .85   \\ \hline
LC-Boost        & .74  & .73    & .62    & .66 & .88 & .80 & .85  & .92   \\ \hline
LC-Boost Factor & .62  & .64    & .58    & .63 & .75 & .77 & .84  & .90   \\ \hline
LL-Boost Factor & .74  & .85    & .70    & .76 & .76 & .82 & 1.20 & 1.37  \\ \hline
\end{tabular}
\begin{tabular}{|c|c|c|c|c|c|c|c|c|}
\hline
\multicolumn{9}{|c|}{``Pre-Great Moderation" 1971:9-1982:12}                                                \\ \hline
                & IP   & PAYEMS & UNRATE & CLF  & RPI & CPI  & FF   & TB3MS \\ \hline
TV-AR           & 1.07 & 1      & 1.15   & .71  & .96 & 1.11 & 1.01 & 1.14  \\ \hline
DI            & .38  & .49    & .48    & 1.51 & .61 & 82   & .88  & .89   \\ \hline
Lasso           & .30  & .51    & .41    & 1.27 & .71 & 1.13 & .76  & .77   \\ \hline
Boost           & .27  & .44    & .43    & 1.24 & .67 & .92  & .72  & .75   \\ \hline
Boost Factor    & .32  & .50    & .43    & 1.34 & .65 & .84  & .74  & .80   \\ \hline
LC-Boost        & .29  & .44    & .38    & .86  & .77 & 1.03 & .70  & .79   \\ \hline
LC-Boost Factor & .31  & .54    & .43    & .60  & .66 & .94  & .77  & .81   \\ \hline
LL-Boost Factor & .41  & .83    & .66    & .86  & .80 & 1.04 & 1.28 & 1.57  \\ \hline
\end{tabular}
\begin{tabular}{|c|c|c|c|c|c|c|c|c|}
\hline
\multicolumn{9}{|c|}{``Great Moderation" 1983:1-2006:12}                                                            \\ \hline
                & IP   & PAYEMS & UNRATE & CLF & RPI  & CPI  & FF   & TB3MS \\ \hline
TV-AR           & 1.12 & 1      & 1.11   & .65 & 1.06 & 1.07 & 1.04 & 1.03  \\ \hline
DI            & 1.18 & 1.08   & .75    & .88 & 1    & 1    & .85  & .92   \\ \hline
Lasso           & 1.23 & 1.32   & .97    & .85 & 1.10 & .90  & .80  & 1.03  \\ \hline
Boost           & 1.36 & 1.20   & .95    & .81 & 1.08 & .91  & .85  & 1     \\ \hline
Boost Factor    & 1.25 & 1.16   & .67    & .89 & 1    & .90  & .80  & .90   \\ \hline
LC-Boost        & 1.43 & 1.26   & .97    & .82 & 1.16 & .90  & 1.07 & 1.05  \\ \hline
LC-Boost Factor & .99  & .89    & .71    & .71 & .89  & 1    & .97  & 1.01  \\ \hline
LL-Boost Factor & 1.03 & 1      & .79    & .90 & .86  & 1.14 & 1.16 & 1.22  \\ \hline
\end{tabular}
\begin{tabular}{|c|c|c|c|c|c|c|c|c|}
\hline
\multicolumn{9}{|c|}{``Post Great Moderation" 2007:1-2018:8}                                                             \\ \hline
                & IP   & PAYEMS & UNRATE & CLF & RPI & CPI  & FF   & TB3MS \\ \hline
TV-AR           & .96  & .98    & 1      & .58 & .96 & .41  & .79  & .83   \\ \hline
DI            & 1.10 & 1.02   & .91    & .76 & .88 & 1.02 & .94  & 1     \\ \hline
Lasso           & 1.15 & .76    & .99    & .36 & .95 & .79  & .89  & 1.07  \\ \hline
Boost           & 1.13 & .72    & .91    & .33 & .81 & .81  & 1.13 & 1.25  \\ \hline
Boost Factor    & 1.04 & .98    & .82    & .79 & .77 & .94  & 1    & 1.01  \\ \hline
LC-Boost        & .94  & 68     & .60    & .38 & .73 & .48  & 1.29 & 1.21  \\ \hline
LC-Boost Factor & .82  & .54    & .66    & .57 & .72 & .42  & .92  & .98   \\ \hline
LL-Boost Factor & 1.01 & .69    & .66    & .55 & .66 & .33  & .56  & .67   \\ \hline
\end{tabular}
\end{table}

\begin{figure}
    \centering
    \includegraphics[scale=.75]{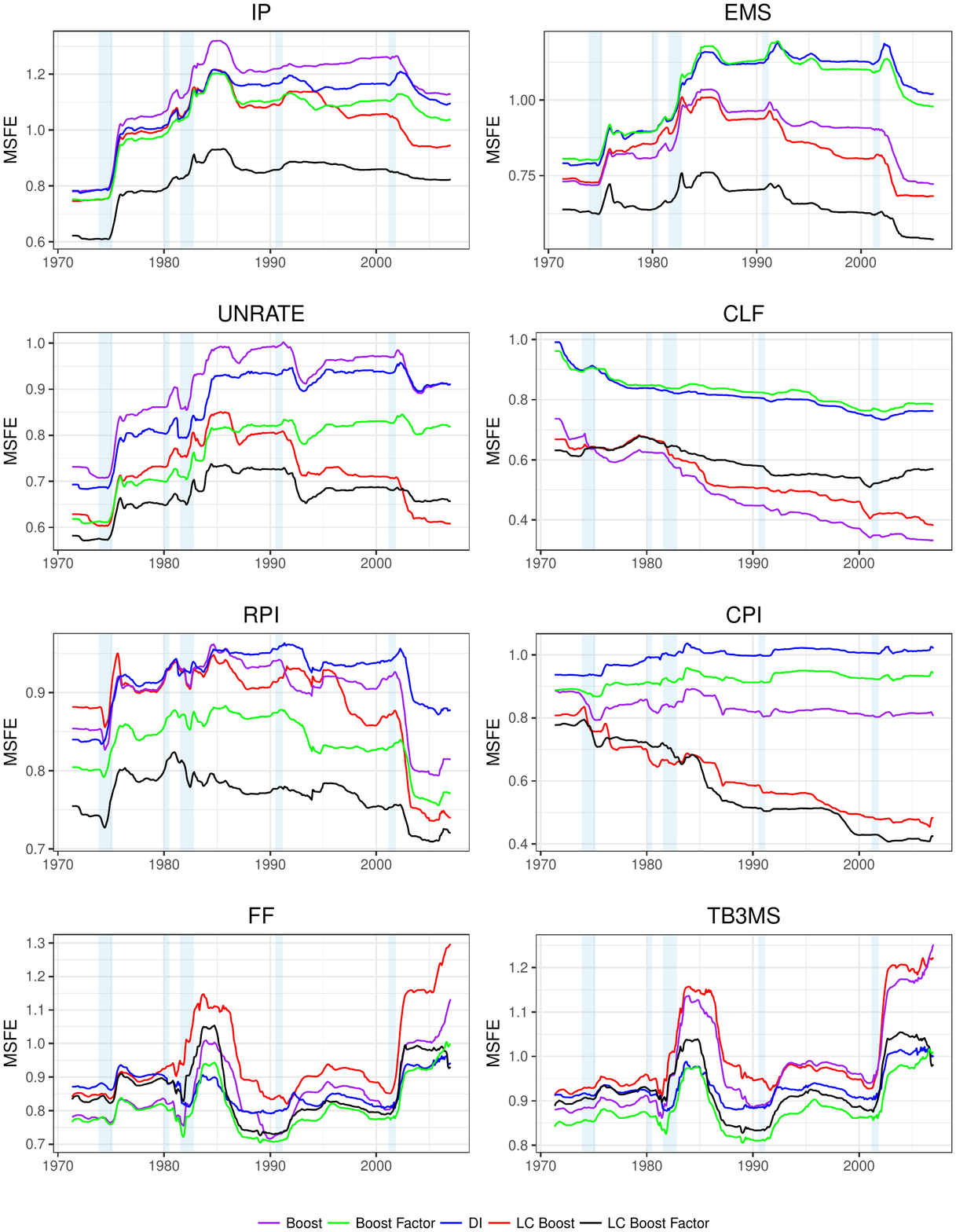}
    \caption{\footnotesize{\textbf{MSFE by start date of out of sample period. Horizon $h=12$}. More specifically we plot: $
    MSFE^{12}_{(i)}(T_1,T_2)=\sum_{t=T_1}^{T_2}\hat{\epsilon}^{2}_{t,(i)}/\sum_{t=T_1}^{T_2}\hat{\epsilon}^{2}_{t,(AR)}$, where we $T_1$ vary from 1971:9 until 2006:12,  with $T_2$=2018:8. Shaded regions represent NBER recession dates.}}
    \label{MSFE12}
\end{figure}

\begin{figure}
    \centering
    \includegraphics[scale=.75]{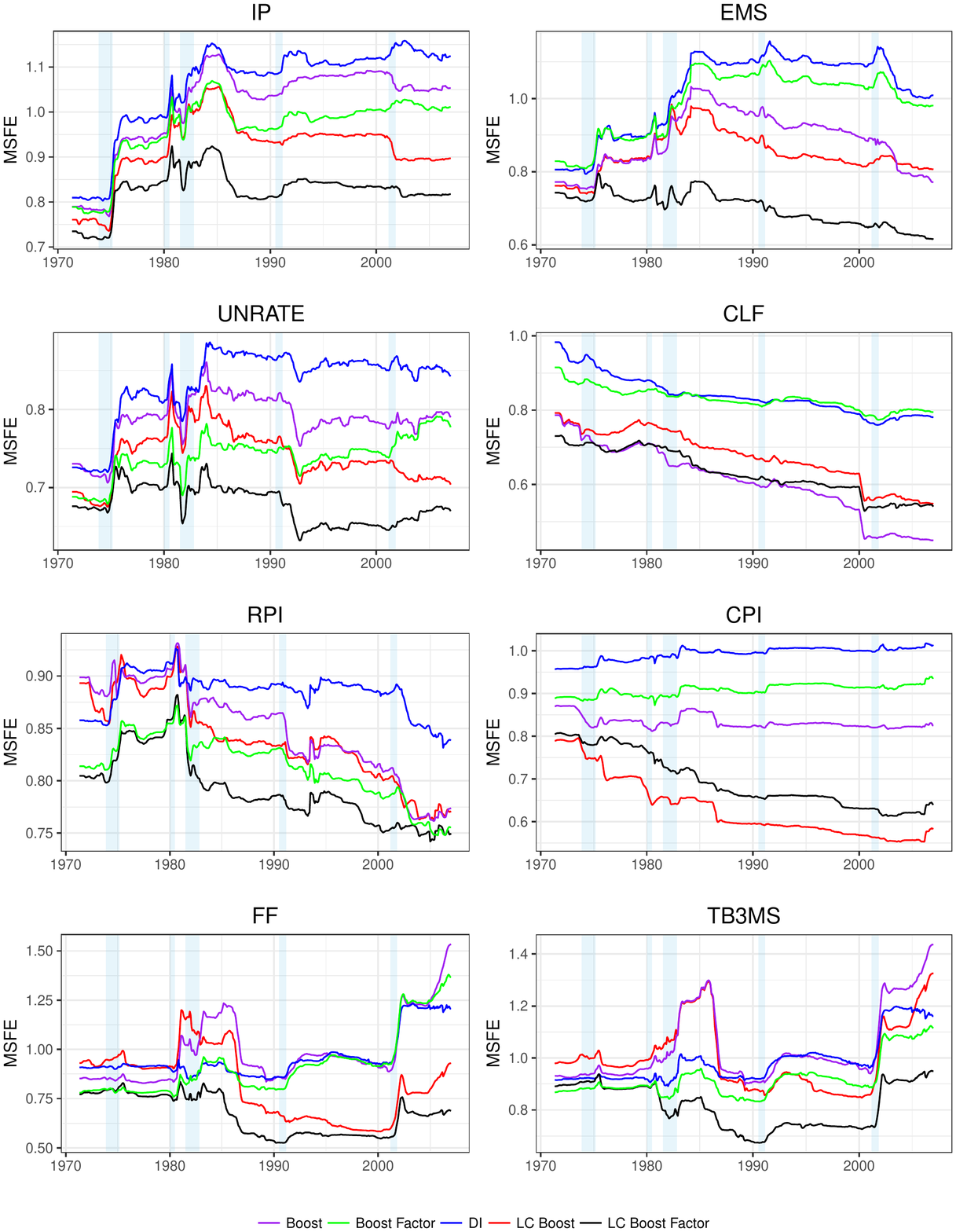}
    \caption{\footnotesize{\textbf{MSFE by start date of Out of sample period. Horizon $h=6$}. See notes to figure \ref{MSFE12}.}}
    \label{MSFE6}
\end{figure}


\begin{figure}
    \centering
    \includegraphics[scale=.75]{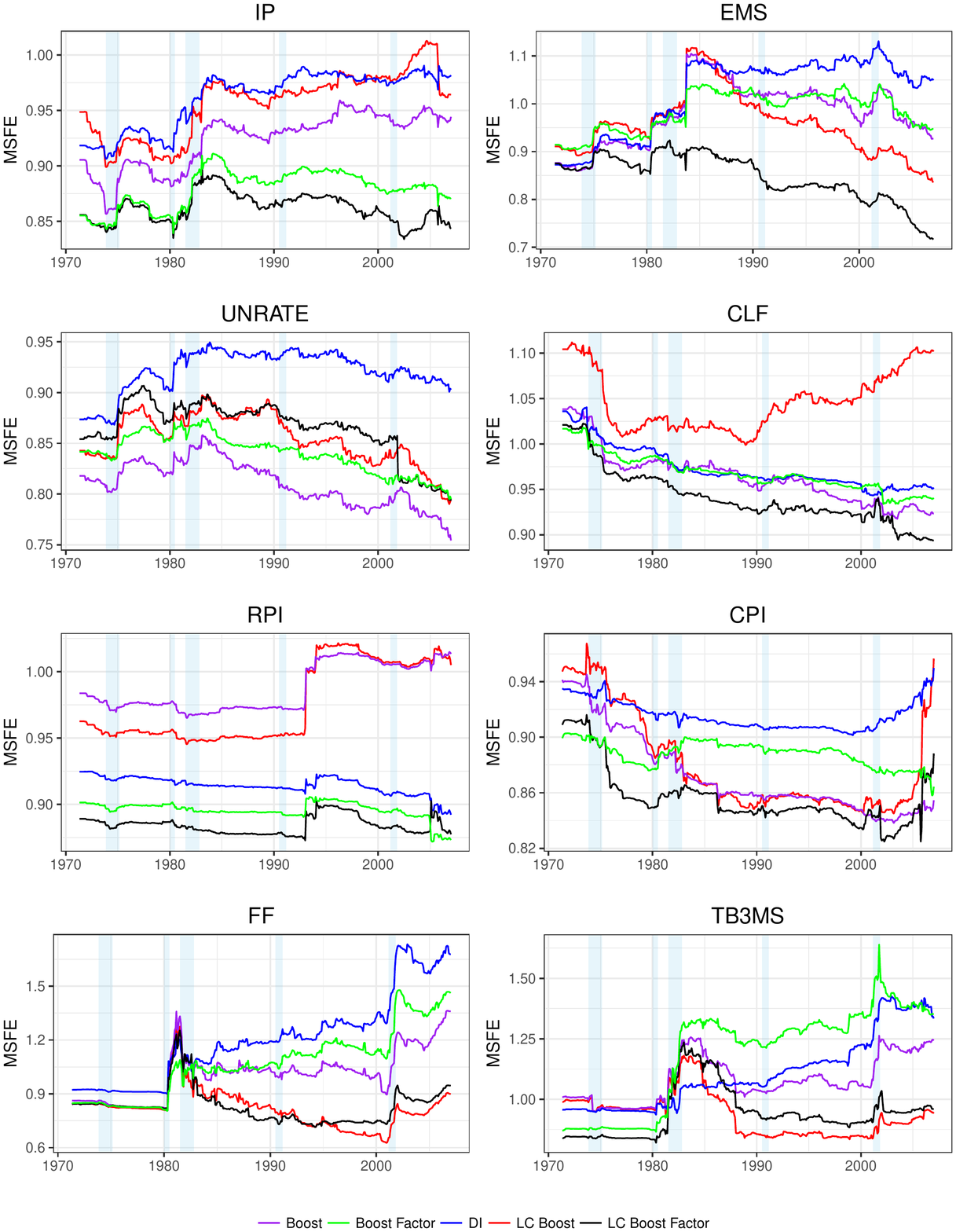}
    \caption{\footnotesize{\textbf{MSFE by start date of out of sample period. Horizon $h=1$.} See notes to figure \ref{MSFE12}.}}
    \label{MSFE1}
\end{figure}

\begin{figure}
    \centering
    \includegraphics[scale=.75]{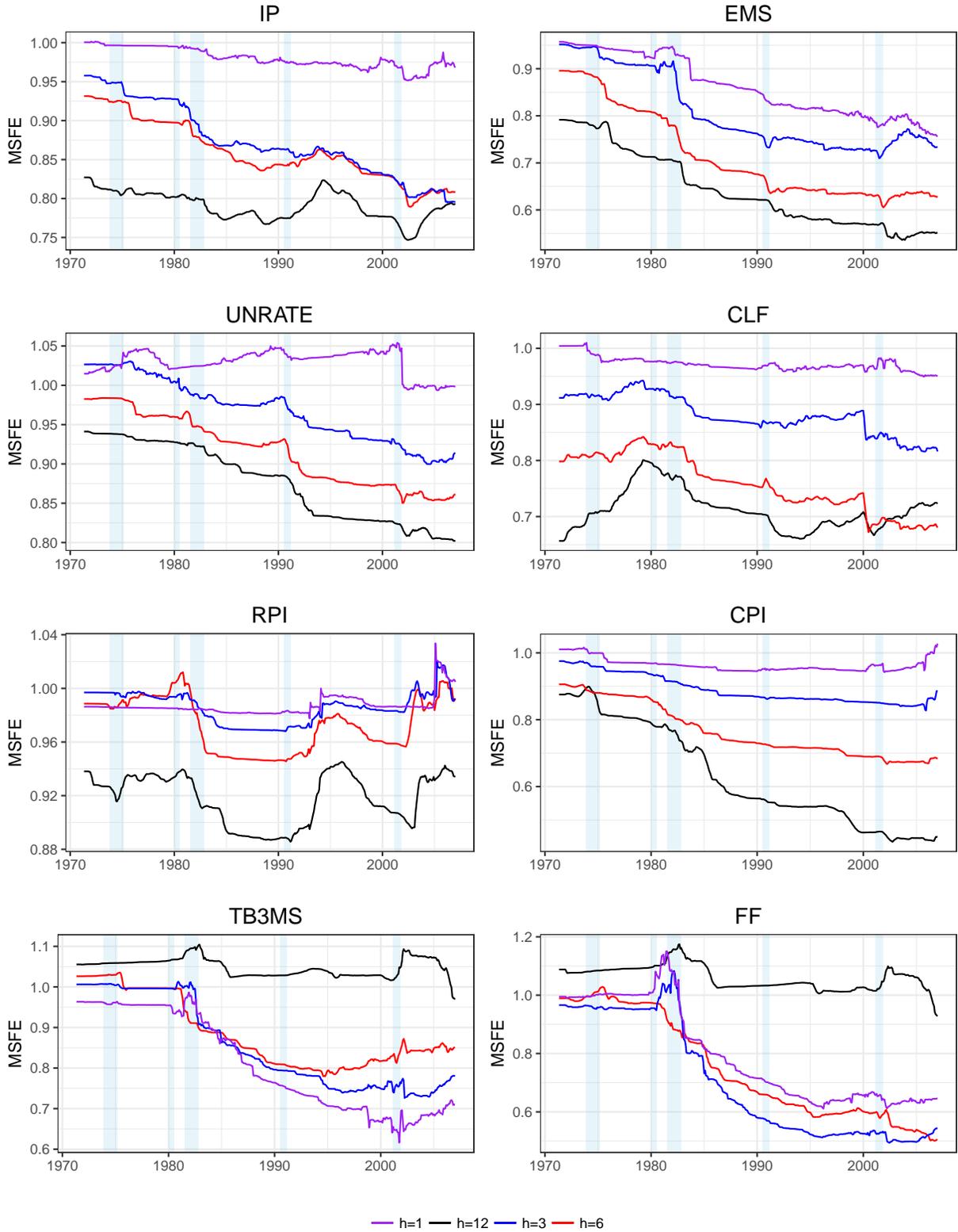}
    \caption{\footnotesize{\textbf{MSFE of LC-Boost Factor (LC-BF) relative to MSFE of Boost Factor (BF) by start date of out of sample period}: Figure plots $MSFE_{(LC-BF)}(T_0,T_2)/MSFE_{(BF)}(T_0,T_2)$ for $T_0=\text{1971:9},\ldots,\text{2006:12}$ and $T_2=\text{2018:8}$. See notes to figure \ref{MSFE12} or equation (\ref{MSFE_startdate}) for details. Colored lines represent the different horizons. Shaded regions represent NBER recession dates.}}
    \label{Benefits_TVP}
\end{figure}


\begin{figure}
    \centering
    \includegraphics[scale=.75]{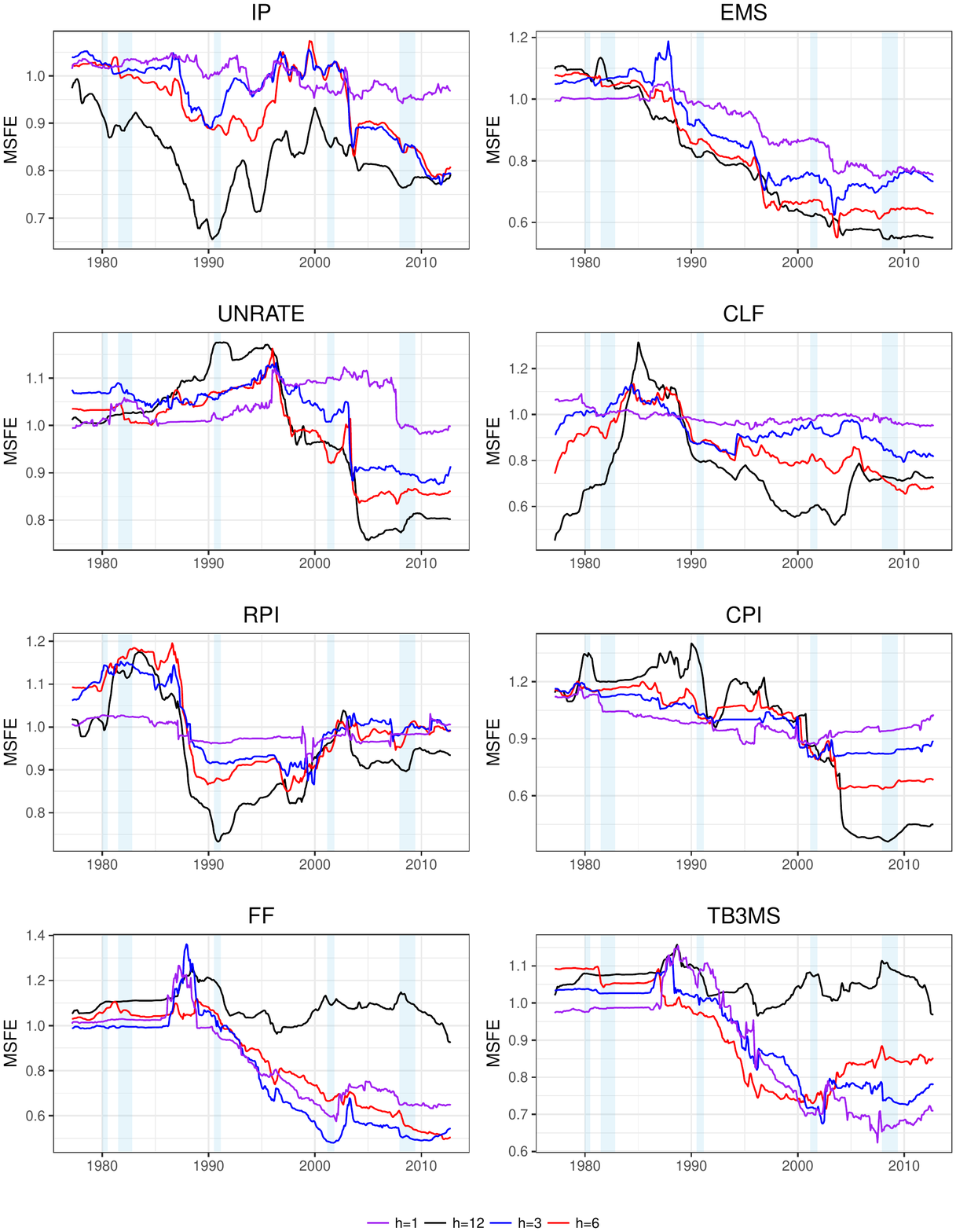}
    \caption{\footnotesize{\textbf{Local MSFE of LC-Boost Factor relative to Local MSFE of Boost Factor}: To obtain the local MSFE of a model we use a rolling mean with a window size of 70 observations, which gives us: $
    \text{L-MSFE}_{i}(t_0)=\sum_{t=t_0-70}^{t_0+70}\hat{\epsilon}_{t,(i)}^{2}/\sum_{t=t_0-70}^{t_0+70}\hat{\epsilon}_{t,(AR)}^{2}$
with the convention that $\hat{\epsilon}_{t,(i)}=0$ for $t\leq 0, t\geq T$. We then plot $\text{L-MSFE}_{LC Boost Factor}(t)/\text{L-MSFE}_{Boost Factor}(t)$, for $t=\text{1977:3},\ldots,\text{2012:10}$}}
    \label{localMSFE1}
\end{figure}

\begin{figure}
    \centering
    \includegraphics[scale=.75]{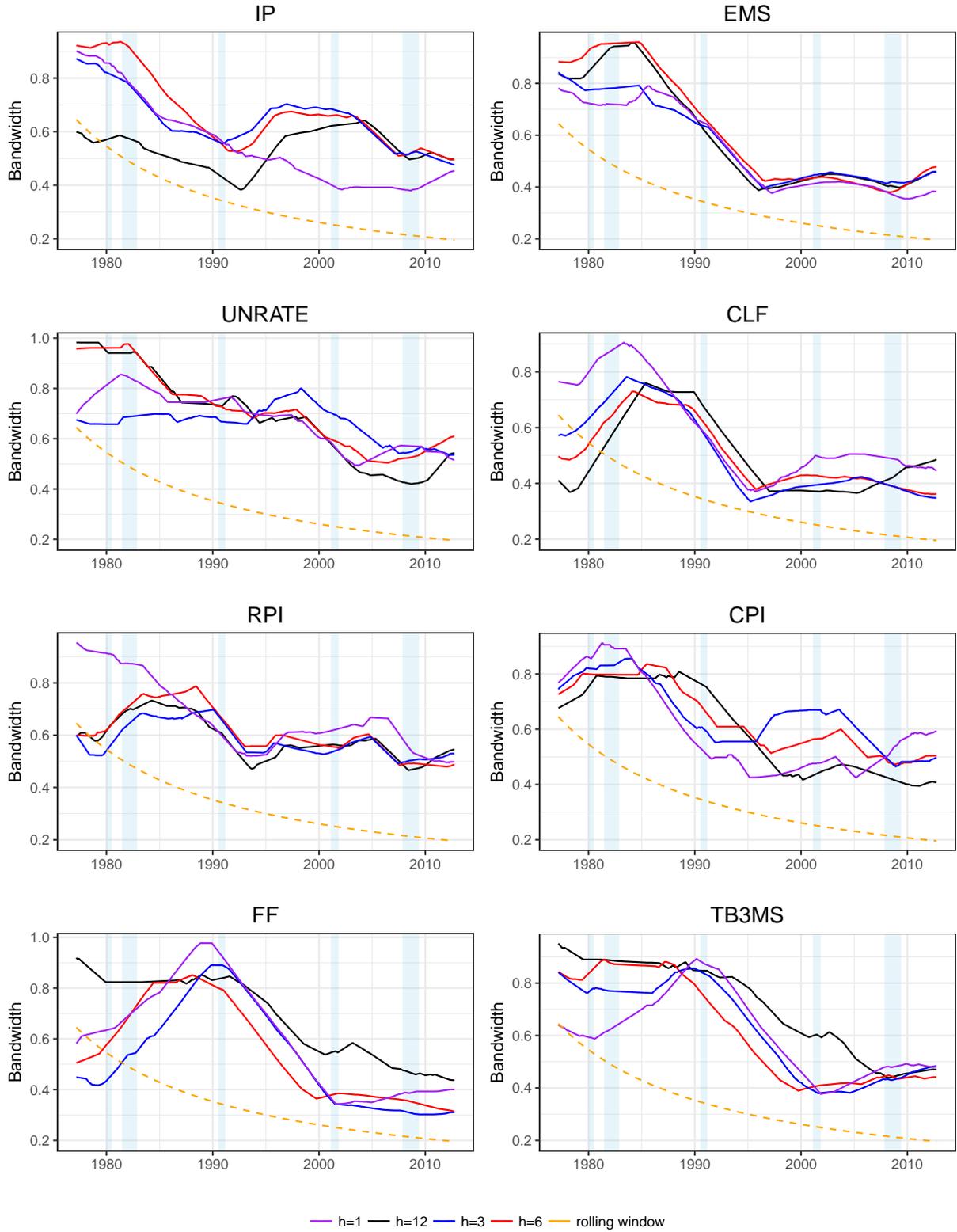}
    \caption{\footnotesize{\textbf{Local Bandwidth of LC-Boost Factor}: To obtain the local bandwidth of LC-Boost Factor we use a rolling mean with a window size of 70 observations, which gives us: $
    \text{L-BW}_{i}(t_0)=\sum_{t=t_0-70}^{t_0+70}\hat{b}_{t,(i)}/140$
for $i=$LC-Boost Factor. We also include the local bandwidth implied by a rolling window estimator which uses the last 120 observations}} 
\label{localBW}
\end{figure}

\clearpage

\begin{spacing}{0.4}
\small
\bibliographystyle{apa}
\bibliography{varselection_latest}

\begin{thebibliography}{}

\bibitem[\protect\astroncite{Adrian et~al.}{2019}]{adrian2019vulnerable}
Adrian, T., Boyarchenko, N., and Giannone, D. (2019).
\newblock Vulnerable growth.
\newblock {\em American Economic Review}, 109(4):1263--89.

\bibitem[\protect\astroncite{Bai and Ng}{2002}]{bai2002determining}
Bai, J. and Ng, S. (2002).
\newblock Determining the number of factors in approximate factor models.
\newblock {\em Econometrica}, 70(1):191--221.

\bibitem[\protect\astroncite{Bai and Ng}{2009}]{bai2009boosting}
Bai, J. and Ng, S. (2009).
\newblock Boosting diffusion indices.
\newblock {\em Journal of Applied Econometrics}, 24(4):607--629.

\bibitem[\protect\astroncite{Basu and Michailidis}{2015}]{basu2015}
Basu, S. and Michailidis, G. (2015).
\newblock Regularized estimation in sparse high-dimensional time series models.
\newblock {\em Ann. Statist.}, 43:1535--1567.

\bibitem[\protect\astroncite{Bates et~al.}{2013}]{bates2013consistent}
Bates, B.~J., Plagborg-M{\o}ller, M., Stock, J.~H., and Watson, M.~W. (2013).
\newblock Consistent factor estimation in dynamic factor models with structural
  instability.
\newblock {\em Journal of Econometrics}, 177(2):289--304.

\bibitem[\protect\astroncite{Breitung and
  Eickmeier}{2011}]{breitung2011testing}
Breitung, J. and Eickmeier, S. (2011).
\newblock Testing for structural breaks in dynamic factor models.
\newblock {\em Journal of Econometrics}, 163(1):71--84.

\bibitem[\protect\astroncite{Buhlmann}{2006}]{Buhlmann2006}
Buhlmann, P. (2006).
\newblock Boosting for high-dimensional linear models.
\newblock {\em Ann. Statist.}, 2(34):559--583.

\bibitem[\protect\astroncite{Buhlmann and Hothorn}{2007}]{Buhlmann2007}
Buhlmann, P. and Hothorn, T. (2007).
\newblock Boosting algorithms: Regularization, prediction and model fitting.
\newblock {\em Statistical Science}, 22(4):477--505.

\bibitem[\protect\astroncite{Cai}{2007}]{cai2007}
Cai, Z. (2007).
\newblock Trending time-varying coefficient time series models with serially
  correlated errors.
\newblock {\em Journal of Econometrics}, 136(1):163--188.

\bibitem[\protect\astroncite{Carrasco and Rossi}{2016}]{carrasco2016sample}
Carrasco, M. and Rossi, B. (2016).
\newblock In-sample inference and forecasting in misspecified factor models.
\newblock {\em Journal of Business \& Economic Statistics}, 34(3):313--338.

\bibitem[\protect\astroncite{Carriero et~al.}{2019}]{carriero2019large}
Carriero, A., Clark, T.~E., and Marcellino, M. (2019).
\newblock Large bayesian vector autoregressions with stochastic volatility and
  non-conjugate priors.
\newblock {\em Journal of Econometrics}.

\bibitem[\protect\astroncite{Casini and Perron}{2018}]{casini2018structural}
Casini, A. and Perron, P. (2018).
\newblock Structural breaks in time series.
\newblock {\em arXiv preprint arXiv:1805.03807}.

\bibitem[\protect\astroncite{Chan et~al.}{2012}]{chan2012time}
Chan, J.~C., Koop, G., Leon-Gonzalez, R., and Strachan, R.~W. (2012).
\newblock Time varying dimension models.
\newblock {\em Journal of Business \& Economic Statistics}, 30(3):358--367.

\bibitem[\protect\astroncite{Chen}{2015}]{chen2015modeling}
Chen, B. (2015).
\newblock Modeling and testing smooth structural changes with endogenous
  regressors.
\newblock {\em Journal of Econometrics}, 185(1):196--215.

\bibitem[\protect\astroncite{Chen and Hong}{2012}]{Bin12}
Chen, B. and Hong, Y. (2012).
\newblock Testing for smooth structural changes in time series models via
  nonparametric regression.
\newblock {\em Econometrica}, 80(3):1157--1183.

\bibitem[\protect\astroncite{Chen et~al.}{2013}]{chen2013}
Chen, X., Xu, M., and Wu, W.~B. (2013).
\newblock Covariance and precision matrix estimation for high-dimensional time
  series.
\newblock {\em Ann. Statist.}, 41:2994--3021.

\bibitem[\protect\astroncite{Clements and Hendry}{1996}]{Clements96}
Clements, M.~P. and Hendry, D.~F. (1996).
\newblock Intercept corrections and structural change.
\newblock {\em Journal of Applied Econometrics}, 11(5):475--494.

\bibitem[\protect\astroncite{Clements and
  Hendry}{2000}]{clements2000forecasting}
Clements, M.~P. and Hendry, D.~F. (2000).
\newblock {\em Forecasting non-stationary economic time series}.
\newblock MIT Press.

\bibitem[\protect\astroncite{Cogley and Sargent}{2001}]{cogley2001evolving}
Cogley, T. and Sargent, T.~J. (2001).
\newblock Evolving post-world war ii {US} inflation dynamics.
\newblock {\em NBER macroeconomics annual}, 16:331--373.

\bibitem[\protect\astroncite{Dahlhaus}{1996}]{dahlhaus1996kullback}
Dahlhaus, R. (1996).
\newblock On the kullback-leibler information divergence of locallystationary
  processes.
\newblock {\em Stochastic processes and their applications}, 62(1):139--168.

\bibitem[\protect\astroncite{Dahlhaus}{2012}]{dahlhaus2012locally}
Dahlhaus, R. (2012).
\newblock Locally stationary processes.
\newblock In {\em Handbook of statistics}, volume~30, pages 351--413. Elsevier.

\bibitem[\protect\astroncite{Dahlhaus et~al.}{1997}]{dahlhaus1997fitting}
Dahlhaus, R. et~al. (1997).
\newblock Fitting time series models to nonstationary processes.
\newblock {\em The annals of Statistics}, 25(1):1--37.

\bibitem[\protect\astroncite{Dahlhaus et~al.}{2018}]{dahlhaus2017towards}
Dahlhaus, R., Richter, S., and Wu, W.~B. (2018).
\newblock Towards a general theory for non-linear locally stationary processes.
\newblock {\em Bernoulli}.
\newblock In press.

\bibitem[\protect\astroncite{Dangl and Halling}{2012}]{dangl2012predictive}
Dangl, T. and Halling, M. (2012).
\newblock Predictive regressions with time-varying coefficients.
\newblock {\em Journal of Financial Economics}, 106(1):157--181.

\bibitem[\protect\astroncite{Ding et~al.}{2017}]{ding2017}
Ding, X., Qiu, Z., Chen, X., et~al. (2017).
\newblock Sparse transition matrix estimation for high-dimensional and locally
  stationary vector autoregressive models.
\newblock {\em Electronic Journal of Statistics}, 11(2):3871--3902.

\bibitem[\protect\astroncite{Doukhan}{1994}]{Doukhan94}
Doukhan, P. (1994).
\newblock {\em Mixing: Properties and Examples}, volume~85 of {\em Lecture
  Notes in Statistics}.
\newblock Springer-Verlag New York.

\bibitem[\protect\astroncite{Efron et~al.}{2004}]{efron2004least}
Efron, B., Hastie, T., Johnstone, I., Tibshirani, R., et~al. (2004).
\newblock Least angle regression.
\newblock {\em The Annals of statistics}, 32(2):407--499.

\bibitem[\protect\astroncite{Eickmeier et~al.}{2015}]{eickmeier2015classical}
Eickmeier, S., Lemke, W., and Marcellino, M. (2015).
\newblock Classical time varying factor-augmented vector auto-regressive
  models—estimation, forecasting and structural analysis.
\newblock {\em Journal of the Royal Statistical Society: Series A (Statistics
  in Society)}, 178(3):493--533.

\bibitem[\protect\astroncite{Fan et~al.}{2011}]{fan2011sparse}
Fan, J., Lv, J., and Qi, L. (2011).
\newblock Sparse high-dimensional models in economics.
\newblock {\em Annu. Rev. Econ.}, 3(1):291--317.

\bibitem[\protect\astroncite{Fenske et~al.}{2011}]{fenske2011identifying}
Fenske, N., Kneib, T., and Hothorn, T. (2011).
\newblock Identifying risk factors for severe childhood malnutrition by
  boosting additive quantile regression.
\newblock {\em Journal of the American Statistical Association},
  106(494):494--510.

\bibitem[\protect\astroncite{Freund and Schapire}{1997}]{freund1997decision}
Freund, Y. and Schapire, R.~E. (1997).
\newblock A decision-theoretic generalization of on-line learning and an
  application to boosting.
\newblock {\em Journal of computer and system sciences}, 55(1):119--139.

\bibitem[\protect\astroncite{Friedman et~al.}{2004}]{friedman2004discussion}
Friedman, J., Hastie, T., Rosset, S., Tibshirani, R., and Zhu, J. (2004).
\newblock Discussion of boosting papers.
\newblock {\em Ann. Statist}, 32:102--107.

\bibitem[\protect\astroncite{Friedman et~al.}{2001}]{ESL}
Friedman, J., Hastie, T., and Tibshirani, R. (2001).
\newblock {\em The elements of statistical learning}, volume~1.
\newblock Springer series in statistics New York, NY, USA:.

\bibitem[\protect\astroncite{Friedman et~al.}{2000}]{friedman2000additive}
Friedman, J., Hastie, T., Tibshirani, R., et~al. (2000).
\newblock Additive logistic regression: a statistical view of boosting (with
  discussion and a rejoinder by the authors).
\newblock {\em The annals of statistics}, 28(2):337--407.

\bibitem[\protect\astroncite{Friedman}{2001}]{friedman2001greedy}
Friedman, J.~H. (2001).
\newblock Greedy function approximation: a gradient boosting machine.
\newblock {\em Annals of statistics}, pages 1189--1232.

\bibitem[\protect\astroncite{Fryzlewicz et~al.}{2011}]{fryzlewicz2011mixing}
Fryzlewicz, P., Rao, S.~S., et~al. (2011).
\newblock Mixing properties of arch and time-varying arch processes.
\newblock {\em Bernoulli}, 17(1):320--346.

\bibitem[\protect\astroncite{Giraitis et~al.}{2013}]{giraitis2013adaptive}
Giraitis, L., Kapetanios, G., and Price, S. (2013).
\newblock Adaptive forecasting in the presence of recent and ongoing structural
  change.
\newblock {\em Journal of Econometrics}, 177(2):153--170.

\bibitem[\protect\astroncite{Goyal and Welch}{2003}]{Welch2003}
Goyal, A. and Welch, I. (2003).
\newblock Predicting the equity premium with dividend ratios.
\newblock {\em Management Science}, 49(5):639--654.

\bibitem[\protect\astroncite{Groen et~al.}{2013}]{groen2013real}
Groen, J.~J., Paap, R., and Ravazzolo, F. (2013).
\newblock Real-time inflation forecasting in a changing world.
\newblock {\em Journal of Business \& Economic Statistics}, 31(1):29--44.

\bibitem[\protect\astroncite{Hamilton}{1989}]{hamilton1989new}
Hamilton, J.~D. (1989).
\newblock A new approach to the economic analysis of nonstationary time series
  and the business cycle.
\newblock {\em Econometrica: Journal of the Econometric Society}, pages
  357--384.

\bibitem[\protect\astroncite{Han and Tsay}{2017}]{han2017high}
Han, Y. and Tsay, R.~S. (2017).
\newblock High-dimensional linear regression for dependent observations with
  application to nowcasting.
\newblock {\em arXiv preprint arXiv:1706.07899}.

\bibitem[\protect\astroncite{Hastie et~al.}{2007}]{hastie2007forward}
Hastie, T., Taylor, J., Tibshirani, R., and Walther, G. (2007).
\newblock Forward stagewise regression and the monotone lasso.
\newblock {\em Electronic Journal of Statistics}, 1:1--29.

\bibitem[\protect\astroncite{Hepp et~al.}{2016}]{hepp2016approaches}
Hepp, T., Schmid, M., Gefeller, O., Waldmann, E., and Mayr, A. (2016).
\newblock Approaches to regularized regression--a comparison between gradient
  boosting and the lasso.
\newblock {\em Methods of information in medicine}, 55(05):422--430.

\bibitem[\protect\astroncite{Hofner et~al.}{2014}]{hofner2014model}
Hofner, B., Mayr, A., Robinzonov, N., and Schmid, M. (2014).
\newblock Model-based boosting in r: a hands-on tutorial using the r package
  mboost.
\newblock {\em Computational statistics}, 29(1-2):3--35.

\bibitem[\protect\astroncite{H{\"o}rmann et~al.}{2010}]{hormann2010weakly}
H{\"o}rmann, S., Kokoszka, P., et~al. (2010).
\newblock Weakly dependent functional data.
\newblock {\em The Annals of Statistics}, 38(3):1845--1884.

\bibitem[\protect\astroncite{Hu et~al.}{2018}]{hu2018estimation}
Hu, L., Huang, T., and You, J. (2018).
\newblock Estimation and identification of a varying-coefficient additive model
  for locally stationary processes.
\newblock {\em Journal of the American Statistical Association}.
\newblock In Press.

\bibitem[\protect\astroncite{Inoue et~al.}{2017}]{inoue2017}
Inoue, A., Jin, L., and Rossi, B. (2017).
\newblock Rolling window selection for out-of-sample forecasting with
  time-varying parameters.
\newblock {\em Journal of Econometrics}, 196(1):55--67.

\bibitem[\protect\astroncite{Kelly and Pruitt}{2015}]{Kelly2015}
Kelly, B. and Pruitt, S. (2015).
\newblock The three-pass regression filter: A new approach to forecasting using
  many predictors.
\newblock {\em Journal of Econometrics}, 186(2):294 -- 316.
\newblock High Dimensional Problems in Econometrics.

\bibitem[\protect\astroncite{Kim and Swanson}{2014}]{kim2014forecasting}
Kim, H.~H. and Swanson, N.~R. (2014).
\newblock Forecasting financial and macroeconomic variables using data
  reduction methods: New empirical evidence.
\newblock {\em Journal of Econometrics}, 178:352--367.

\bibitem[\protect\astroncite{Kock and Callot}{2015}]{Kock2015}
Kock, A. and Callot, A. (2015).
\newblock Oracle inequalities for high dimensional vector autoregressions.
\newblock {\em Journal of Econometrics}, 186:325 -- 344.

\bibitem[\protect\astroncite{Koop and Korobilis}{2012}]{koop2012forecasting}
Koop, G. and Korobilis, D. (2012).
\newblock Forecasting inflation using dynamic model averaging.
\newblock {\em International Economic Review}, 53(3):867--886.

\bibitem[\protect\astroncite{Koop and Korobilis}{2013}]{koop2013large}
Koop, G. and Korobilis, D. (2013).
\newblock Large time-varying parameter vars.
\newblock {\em Journal of Econometrics}, 177(2):185--198.

\bibitem[\protect\astroncite{Koop}{2013}]{koop2013forecasting}
Koop, G.~M. (2013).
\newblock Forecasting with medium and large bayesian vars.
\newblock {\em Journal of Applied Econometrics}, 28(2):177--203.

\bibitem[\protect\astroncite{Lee et~al.}{2016}]{lee2016}
Lee, E.~R., Mammen, E., et~al. (2016).
\newblock Local linear smoothing for sparse high dimensional varying
  coefficient models.
\newblock {\em Electronic Journal of Statistics}, 10(1):855--894.

\bibitem[\protect\astroncite{Lutz and B{\"u}hlmann}{2006}]{lutz2006}
Lutz, R.~W. and B{\"u}hlmann, P. (2006).
\newblock Boosting for high-multivariate responses in high-dimensional linear
  regression.
\newblock {\em Statistica Sinica}, 16(2):471--494.

\bibitem[\protect\astroncite{McCracken and Ng}{2016}]{mccracken2016fred}
McCracken, M.~W. and Ng, S. (2016).
\newblock {FRED-MD}: A monthly database for macroeconomic research.
\newblock {\em Journal of Business \& Economic Statistics}, 34(4):574--589.

\bibitem[\protect\astroncite{McDonald et~al.}{2015}]{mcdonald2015}
McDonald, D.~J., Shalizi, C.~R., and Schervish, M. (2015).
\newblock Estimating beta-mixing coefficients via histograms.
\newblock {\em Electronic Journal of Statistics}, 9(2):2855--2883.

\bibitem[\protect\astroncite{Medeiros and Mendes}{2016}]{MM2016}
Medeiros, M. and Mendes, E. (2016).
\newblock $\ell$1-regularization of high-dimensional time-series models with
  non-gaussian and heteroskedastic errors.
\newblock {\em Journal of Econometrics}, 191:255--271.

\bibitem[\protect\astroncite{Ng}{2013}]{ng-handbook}
Ng, S. (2013).
\newblock Variable selection in predictive regressions.
\newblock In Elliott, G. and Timmermann, A., editors, {\em Handbook of
  Forecasting}, number 754-786 in 2B. North Holland.

\bibitem[\protect\astroncite{Ng}{2014}]{NgBoosting2014}
Ng, S. (2014).
\newblock Viewpoint: Boosting recessions.
\newblock {\em Canadian Journal of Economics/Revue canadienne d'économique},
  47(1):1--34.

\bibitem[\protect\astroncite{Orbe et~al.}{2005}]{orbe2005}
Orbe, S., Ferreira, E., and Rodriguez-Poo, J. (2005).
\newblock Nonparametric estimation of time varying parameters under shape
  restrictions.
\newblock {\em Journal of Econometrics}, 126(1):53--77.

\bibitem[\protect\astroncite{Orbe et~al.}{2006}]{orbe2006}
Orbe, S., Ferreira, E., and Rodriguez-Poo, J. (2006).
\newblock On the estimation and testing of time varying constraints in
  econometric models.
\newblock {\em Statistica Sinica}, pages 1313--1333.

\bibitem[\protect\astroncite{Paye and Timmermann}{2006}]{paye2006instability}
Paye, B.~S. and Timmermann, A. (2006).
\newblock Instability of return prediction models.
\newblock {\em Journal of Empirical Finance}, 13(3):274--315.

\bibitem[\protect\astroncite{Perron et~al.}{2006}]{perron2006dealing}
Perron, P. et~al. (2006).
\newblock Dealing with structural breaks.
\newblock {\em Palgrave handbook of econometrics}, 1(2):278--352.

\bibitem[\protect\astroncite{Pesaran and
  Timmermann}{2007}]{pesaran2007selection}
Pesaran, M.~H. and Timmermann, A. (2007).
\newblock Selection of estimation window in the presence of breaks.
\newblock {\em Journal of Econometrics}, 137(1):134--161.

\bibitem[\protect\astroncite{Petrova}{2019}]{petrova2019quasi}
Petrova, K. (2019).
\newblock A quasi-bayesian local likelihood approach to time varying parameter
  var models.
\newblock {\em Journal of Econometrics}.

\bibitem[\protect\astroncite{Pettenuzzo and
  Timmermann}{2017}]{pettenuzzo2017forecasting}
Pettenuzzo, D. and Timmermann, A. (2017).
\newblock Forecasting macroeconomic variables under model instability.
\newblock {\em Journal of Business \& Economic Statistics}, 35(2):183--201.

\bibitem[\protect\astroncite{Phillips et~al.}{2017}]{phillips2017estimating}
Phillips, P.~C., Li, D., and Gao, J. (2017).
\newblock Estimating smooth structural change in cointegration models.
\newblock {\em Journal of Econometrics}, 196(1):180--195.

\bibitem[\protect\astroncite{Primiceri}{2005}]{primiceri2005time}
Primiceri, G.~E. (2005).
\newblock Time varying structural vector autoregressions and monetary policy.
\newblock {\em The Review of Economic Studies}, 72(3):821--852.

\bibitem[\protect\astroncite{Rapach et~al.}{2010}]{rapach2010out}
Rapach, D.~E., Strauss, J.~K., and Zhou, G. (2010).
\newblock Out-of-sample equity premium prediction: Combination forecasts and
  links to the real economy.
\newblock {\em The Review of Financial Studies}, 23(2):821--862.

\bibitem[\protect\astroncite{Richter and Dahlhaus}{2018}]{richter2017cross}
Richter, S. and Dahlhaus, R. (2018).
\newblock Cross validation for locally stationary processes.
\newblock {\em Ann. Statist.}
\newblock In Press.

\bibitem[\protect\astroncite{Robinson}{1989}]{robinson1989}
Robinson, P.~M. (1989).
\newblock Nonparametric estimation of time-varying parameters.
\newblock In {\em Statistical Analysis and Forecasting of Economic Structural
  Change}, pages 253--264. Springer.

\bibitem[\protect\astroncite{Robinson}{1991}]{robinson1991}
Robinson, P.~M. (1991).
\newblock Time-varying nonlinear regression.
\newblock In {\em Economic Structural Change}, pages 179--190. Springer.

\bibitem[\protect\astroncite{Rosset et~al.}{2004}]{rosset2004boosting}
Rosset, S., Zhu, J., and Hastie, T. (2004).
\newblock Boosting as a regularized path to a maximum margin classifier.
\newblock {\em Journal of Machine Learning Research}, 5(Aug):941--973.

\bibitem[\protect\astroncite{Rossi}{2013}]{rossi2013advances}
Rossi, B. (2013).
\newblock Advances in forecasting under instability.
\newblock In {\em Handbook of economic forecasting}, volume~2, pages
  1203--1324. Elsevier.

\bibitem[\protect\astroncite{Rossi and Sekhposyan}{2010}]{rossi2010have}
Rossi, B. and Sekhposyan, T. (2010).
\newblock Have economic models’ forecasting performance for {US} output
  growth and inflation changed over time, and when?
\newblock {\em International Journal of Forecasting}, 26(4):808--835.

\bibitem[\protect\astroncite{Schinasi and Swamy}{1989}]{schinasi1989}
Schinasi, G.~J. and Swamy, P. A. V.~B. (1989).
\newblock The out-of-sample forecasting performance of exchange rate models
  when coefficients are allowed to change.
\newblock {\em Journal of International Money and Finance}, 8(3):375--390.

\bibitem[\protect\astroncite{Shapiro}{2017}]{shapiro_2017}
Shapiro, J.~M. (2017).
\newblock {\em Is Big Data a Big Deal for Applied Microeconomics?}, volume~2 of
  {\em Econometric Society Monographs}, page 35–52.
\newblock Cambridge University Press.

\bibitem[\protect\astroncite{Stock and Watson}{2003}]{stock2003forecasting}
Stock, J.~H. and Watson, M. (2003).
\newblock Forecasting output and inflation: The role of asset prices.
\newblock {\em Journal of Economic Literature}, 41(3):788--829.

\bibitem[\protect\astroncite{Stock and Watson}{2009}]{stock2009forecasting}
Stock, J.~H. and Watson, M. (2009).
\newblock Forecasting in dynamic factor models subject to structural
  instability.
\newblock {\em The Methodology and Practice of Econometrics. A Festschrift in
  Honour of David F. Hendry}, 173:205.

\bibitem[\protect\astroncite{Stock and Watson}{1996}]{Stock96}
Stock, J.~H. and Watson, M.~W. (1996).
\newblock Evidence on structural instability in macroeconomic time series
  relations.
\newblock {\em Journal of Business \& Economic Statistics}, 14(1):11--30.

\bibitem[\protect\astroncite{Stock and Watson}{2002}]{SW2003}
Stock, J.~H. and Watson, M.~W. (2002).
\newblock Has the business cycle changed and why?
\newblock {\em NBER Macroeconomics Annual}, 17:159--218.

\bibitem[\protect\astroncite{Temlyakov}{2000}]{temlyakov2000weak}
Temlyakov, V.~N. (2000).
\newblock Weak greedy algorithms.
\newblock {\em Advances in Computational Mathematics}, 12(2-3):213--227.

\bibitem[\protect\astroncite{Terasvirta}{1994}]{terasvirta:94}
Terasvirta, T. (1994).
\newblock Specification, estimation, and evaluation of smooth transition
  autoregressive models.
\newblock {\em Journal of the American Statistical Assocation},
  89(425):208--218.

\bibitem[\protect\astroncite{Vogt et~al.}{2012}]{vogt2012nonparametric}
Vogt, M. et~al. (2012).
\newblock Nonparametric regression for locally stationary time series.
\newblock {\em The Annals of Statistics}, 40(5):2601--2633.

\bibitem[\protect\astroncite{Wu and Wu}{2016}]{WuandWu2016}
Wu, W. and Wu, Y. (2016).
\newblock Performance bounds for parameter estimates of high-dimensional linear
  models with correlated errors.
\newblock {\em Electronic Journal of Statistics}, 10:352--379.

\bibitem[\protect\astroncite{Wu}{2005}]{Wu2005}
Wu, W.~B. (2005).
\newblock {Nonlinear system theory: Another look at dependence}.
\newblock {\em Proceedings of the National Academy of Sciences},
  102(40):14150--14154.

\bibitem[\protect\astroncite{Wu}{2011}]{Wu2011}
Wu, W.~B. (2011).
\newblock { Asymptotic theory for stationary processes }.
\newblock {\em Statistics and its Interface}, 4(2):207--226.

\bibitem[\protect\astroncite{Yousuf}{2018}]{Yousuf2017}
Yousuf, K. (2018).
\newblock Variable screening for high dimensional time series.
\newblock {\em Electronic Journal of Statistics}, 12(1):667--702.

\bibitem[\protect\astroncite{Zhang et~al.}{2015}]{zhang2015time}
Zhang, T., Wu, W.~B., et~al. (2015).
\newblock Time-varying nonlinear regression models: Nonparametric estimation
  and model selection.
\newblock {\em The Annals of Statistics}, 43(2):741--768.

\bibitem[\protect\astroncite{Zhou}{2010}]{zhou2010nonparametric}
Zhou, Z. (2010).
\newblock Nonparametric inference of quantile curves for nonstationary time
  series.
\newblock {\em The Annals of Statistics}, 38(4):2187--2217.

\bibitem[\protect\astroncite{Zhou and Wu}{2009}]{zhou2009local}
Zhou, Z. and Wu, W.~B. (2009).
\newblock Local linear quantile estimation for nonstationary time series.
\newblock {\em The Annals of Statistics}, 37(5B):2696--2729.

\end{thebibliography}
\end{spacing}

\clearpage

\section*{Appendix: Proofs}

\label{Proofs}

\begin{proof}[Proof of Theorem 1(a)]
  $ $\newline

The proof follows the framework of \cite{Buhlmann2006}, which handled the case of boosting for iid data using linear least squares base learners.\footnote{Since we refer to \cite{Buhlmann2006} during our proof, we try, when its possible, to keep the notation consistent with their work.} The proof depends on an application of Temlyakov's result \citep{temlyakov2000weak} for the population version of $L_2$ boosting known as ``weak greedy algorithm". We start by considering a step size of $\nu=1$, and smaller step sizes can be handled as in section 6.3 of \cite{Buhlmann2006}.

Throughout this section let $C,c$ refer to generic constants which can change from line to line. We introduce the following notation: Let $\bm{\tilde{x}}_{t-h}(u)$ and $\tilde{Y}_{t}(u)$ be the stationary approximation to $\bm{x}_{t-h,T}$ and $Y_{t,T}$ respectively with approximation error $O_p(T^{-1}) \rightarrow 0$. Let the inner product $\inp[\big]{\tilde{X}_{j,t}(u)
}{\tilde{X}_{k,t}(u)}=E(\tilde{X}_j(u)\tilde{X}_k(u))$ with  $||\tilde{X}_j(u)||^2=E(\tilde{X}^2_j(u))$. For ease of presentation let $||\tilde{X}_j(u)||^2=1$, $\forall j \leq p_T$, we will expand on this in more detail later, but we note that we are not using loss of generality by this assumption. We can define a new scaled process for each $u \in [0,1]$ as $\tilde{X}^{*}_{j,t}(u)=\tilde{X}_{j,t}(u)/||\tilde{X}_j(u)||$. However, we note that we are still placing all our assumptions on the unscaled process i.e weak dependence, smoothness, locally stationary approximation to $X_{j,uT,T}$ etc.

Let $f(u,\bm{\tilde{x}}_{t-h}(u))=\tilde{\bm{x}}_{t-h}(u)\bm{\beta}(u)$ be the stationary approximation to $f(u,\bm{x}_{t-h,T})$ with approximation error $O_p(T^{-1})$. For readability, we define, for any rescaled time point $u \in [0,1]$:
$$ f(\tilde{\bm{x}}_{uT-h}(u)) \equiv f(u,\tilde{\bm{x}}_{uT-h}(u)), \text{and } f(\bm{x}_{uT-h,T})\equiv f(u,\bm{x}_{uT-h,T}).$$
We now define a sequence of remainder functions for the population version of $L_2$ Boosting:
\begin{align*}
    &R^{0}f(\tilde{\bm{x}}_{uT-h}(u))=f(\tilde{\bm{x}}_{uT-h}(u)),\\
    &R^{m}f(\tilde{\bm{x}}_{uT-h}(u))=R^{m-1}f(\tilde{\bm{x}}_{uT-h}(u))-\inp[\big]{R^{m-1}f(\tilde{\bm{x}}_{uT-h}(u))}{\tilde{X}_{\mathcal{S}_m,uT-h}(u)}\tilde{X}_{\mathcal{S}_m,uT-h}(u), m=1,2,\ldots
\end{align*}
Where $\mathcal{S}_m=\argmax_{j}|\inp[\big]{R^{m-1}f(\tilde{\bm{x}}_{uT-h}(u))}{\tilde{X}_{j,uT-h}(u)}|$. Given that this criterion is sometimes infeasible to realize in practice, a weaker criterion is: Choose any $\mathcal{S}_m$ which satisfies:
\begin{align}
    |\inp[\big]{R^{m-1}f(\tilde{\bm{x}}_{uT-h}(u))}{\tilde{X}_{\mathcal{S}_m}(u)}| \geq b*\sup_{j} |\inp[\big]{R^{m-1}f(\tilde{\bm{x}}_{uT-h}(u))}{\tilde{X}_j(u)}|, \text{ for some } b \in (0,1]
    \label{relaxed}
\end{align}
We then obtain: \begin{align*}&f(\tilde{\bm{x}}_{uT-h}(u))=\sum_{j=0}^{m-1}\inp[\big]{R^{j}f(\tilde{\bm{x}}_{uT-h}(u))}{\tilde{X}_{\mathcal{S}_{j+1}}(u)}+R^{m}f(\tilde{\bm{x}}_{uT-h}(u)),\\
&||R^mf(\tilde{\bm{x}}_{uT-h}(u))||^2=||R^{m-1}f(\tilde{\bm{x}}_{uT-h}(u))||^2-|\inp[\big]{R^{m-1}f(\tilde{\bm{x}}_{uT-h}(u))}{\tilde{X}_{\mathcal{S}_m}(u)}|^2
\end{align*}
If (\ref{relaxed}) is met, then we have the following bound, for the population version of $L_2$ Boosting, provided by  \cite{temlyakov2000weak}:
\begin{align}
    ||R^mf(\tilde{\bm{x}}_{uT-h}(u))||^2 \leq B(1+mb^2)^{\frac{-b}{4+2b}}
    \label{temlyakov}
\end{align}
with as defined in (\ref{relaxed}), and $ \sup_{u \in [0,1]}|\bm{\beta}(u)| \leq B < \infty$.

To analyze the sample version of our LC-Boost algorithm, we introduce the following notation:
\begin{align*}&\inp[\big]{X_{j,\cdot,T}}{X_{k,\cdot,T}}_{(T)}=\sum_{t=1}^{T}K_b(t/T-u)X_{j,t-h,T}X_{k,t-h,T}\\
&\inp[\big]{f}{X_{k,\cdot,T}}_{(T)}=\sum_{t=1}^{T}K_b(t/T-u)f(\bm{x}_{t-h})X_{k,t-h,T}\\
&||X_{j,\cdot,T}||^2_{(T)}=\sum_{t=1}^{T}K_b(t/T-u)X_{j,t-h,T}^2
\end{align*}
We suppress the dependence on $u$ in the notation above. We assume that $||X_{j,\cdot,T}||_{(T)}=1$, and once again we do not lose any generality from this. We can define a new scaled process for each $u \in [0,1]$ as  ${X^{*}_{j,t,T}(u)}_{t=1,\ldots, T}= {X_{j,t,T}/||X_{j,\cdot,T}||_{(T)}}_{t=1,\ldots, T}$. As before we are not placing any assumptions such as local stationarity or weak dependence on this scaled process, and we show that our results hold only from the assumptions placed on the unscaled process $X_{j,t,T}$. The key will be the uniform concentration bounds derived in lemma 1. As previously, we can define the sequence of sample remainder functions as:
\begin{align*}
    &\hat{R}_T^{0}f(\bm{x}_{uT-h,T})=f(\bm{x}_{uT-h,T}),\\
    &\hat{R}_T^{m}f(\bm{x}_{uT-h,T})=\hat{R}_T^{m-1}f(\bm{x}_{uT-h,T})-\inp[\big]{\hat{R}_T^{m-1}f}{X_{\widehat{\mathcal{S}}_m,\cdot,T}}_{(T)}X_{\widehat{\mathcal{S}}_m,uT-h,T}, m=1,2,\ldots
\end{align*}
where: $\hat{\mathcal{S}}_1=\argmax_{j}|\inp{Y_{\cdot,T}}{X_{j,\cdot,T}}_{(T)}|$ and $\hat{\mathcal{S}}_m=\argmax_{j}|\inp{\hat{R}_{T}^{m}f}{X_{j,\cdot,T}}_{(T)}|$. Therefore, $\hat{R}_{T}^{m}f(\bm{x}_{uT-h,T})=f(\bm{x}_{uT-h,T})-\hat{F}_{lc}^{(M_T)}(u,\bm{x}_{uT-h,T})$, is the difference between $f(\bm{x}_{uT-h,T})$ and its LC-Boost estimate.

Lastly, to proceed with the proof, we define a sequence of semi-population version remainder functions as: 
\begin{align*}
    &\tilde{R}_T^{0}f(\tilde{\bm{x}}_{uT-h}(u))=f(\tilde{\bm{x}}_{uT-h}(u)),\\
    &\tilde{R}_T^{m}f(\tilde{\bm{x}}_{uT-h}(u))=\tilde{R}_T^{m-1}f(\tilde{\bm{x}}_{uT-h}(u))-\inp[\big]{\tilde{R}_T^{m-1}f(\tilde{\bm{x}}_{uT-h}(u))}{\tilde{X}_{\widehat{\mathcal{S}}_m,uT-h}(u)}\tilde{X}_{\widehat{\mathcal{S}}_m,uT-h}(u), m=1,2,\ldots
\end{align*}
The difference between the population and the semi-population remainder functions, is that the semi-population version uses  selectors $\hat{\mathcal{S}}_m$ estimated from the sample. The strategy of the proof is: first we establish that the selectors $\hat{\mathcal{S}}_m$ satisfy a finite sample analogue of (\ref{relaxed}), which allows us to apply Temlyakov's result (\ref{temlyakov}) to the semipopulation version. Lastly, we analyze the difference between the sample and the semipopulation versions: $\hat{R}_T^{m}f(\bm{x}_{uT-h,T})-\tilde{R}_T^{m}f(\tilde{\bm{x}}_{uT-h}(u))$. 
\\
\\
We need the following lemma:
\begin{lem}
Under conditions  \ref{ConditionA}, \ref{ConditionB}, \ref{ConditionD}, \ref{ConditionE}, and for $\kappa$ as defined in Theorem 1, the following hold:
\begin{enumerate}[noitemsep]
\item $\sup_{j,k \leq p_T}\sup_{u \in [0,1]}|\sum_{t=1}^{T}K_b(t/T-u)X_{j,t,T}X_{k,t,T}-E(\tilde{X}_{j,T}(u)\tilde{X}_{k,T}(u))| = \zeta_{T,1}=O_p(S_T^{-\kappa})$
\\
\item $\sup_{j \leq p_T}\sup_{u \in [0,1]}|\sum_{t=1}^{T}K_b(t/T-u)X_{j,t-h,T}\epsilon_{t,T}| = \zeta_{T,2}=O_p(S_T^{-\kappa})$
\\
\item $\sup_{j \leq p_T}\sup_{u \in [0,1]}|\sum_{t=1}^{T}K_b(t/T-u)X_{j,t,T}f(\bm{x}_{t,T})-E(\tilde{X}_{j,T}(u)f(\tilde{\bm{x}}_{T}(u)))| = \zeta_{T,3}=O_p(S_T^{-\kappa})$
\\
\item$\sup_{j \leq p_T}\sup_{u \in [0,1]}|\sum_{t=1}^{T}K_b(t/T-u)X_{j,t-h,T}Y_{t,T})-E(\tilde{X}_{j,t-h}(u)\tilde{Y}_{T}(u)))| = \zeta_{T,4}=O_p(S_T^{-\kappa})$
\end{enumerate}
\label{lemma1}
\end{lem} 

Now we note that using the above lemma we can bound sums of the form:
\begin{align*}
    \sup_{u\in[0,1]}\sup_{j}|\sum_{t=1}^{T}K_b(t/T-u)X_{j,t-h,T}Y_{t,T}/||X_{j,\cdot,T}||_{(T)}-E(\tilde{X}_{j,t-h}(u)\tilde{Y}_{T}(u)))/||\tilde{X}_{j,t-h}(u)|||
\end{align*}
if we assumed that $||\tilde{X}_{j,t}(u)||,||X_{j,t-h,T}||_{(T)} \neq 1$.

Let $\zeta_T=\max(\zeta_{T,1},\zeta_{T,2},\zeta_{T,3},\zeta_{T,4})=O_p(S_T^{-\kappa})$ and denote by $\omega$ a realization of all $T$ sample points involved in estimation. The next lemma bounds the difference between the sample and population learners at step $m$.

\begin{lem}
Suppose conditions \ref{ConditionA}, \ref{ConditionB}, \ref{ConditionD}, and \ref{ConditionE} hold. Then for $\kappa$ as defined in Theorem 1 and on the set $\mathcal{A}_T=\{\omega: \zeta_T(w)< 1/2\}$, we have:
\begin{align*}
   \sup_{u \in [0,1]} \sup_{j \leq p_T} |\inp[\big]{\hat{R}^{m-1}f}{X_{j,\cdot,T}}_{(T)}-\inp[\big]{\tilde{R}^{m-1}f(\tilde{\bm{x}}_{uT-h}(u))}{\tilde{X}_{j,uT-h}(u)}| \leq C(5/2)^m\zeta_T,
\end{align*}
where $C$ does not depend on $m,T$.
\label{lemma2}
\end{lem}

It's clear from lemma \ref{lemma1}, that $P(\mathcal{A}_T) \rightarrow 1$. Which gives us the following lemma:
\begin{lem}
Suppose the conditions needed for lemma 1 hold, then for $m=m_T \rightarrow \infty$ slow enough we have:
\begin{align*}
\sup_{u \in [0,1]}||\tilde{R}_T^{m}f(\tilde{\bm{x}}_{uT-h}(u))||=o_p(1)
\end{align*}
\label{lemma3}
\end{lem}
We now analyze the term: $\hat{R}_{T}^{m}f(\bm{x}_{uT-h,T})=f(\bm{x}_{uT-h,T})-\hat{F}_{lc}^{(M_T)}(u,\bm{x}_{uT-h,T})$. By the triangle inequality we obtain:
\begin{align}
    ||\hat{R}_{T}^{m}f(\bm{x}_{uT-h,T})|| \leq ||\tilde{R}_T^{m}f(\tilde{\bm{x}}_{uT-h}(u))||+ ||\hat{R}_T^{m}f(\bm{x}_{uT-h,T})-\tilde{R}_T^{m}f(\tilde{\bm{x}}_{uT-h}(u))||
    \label{triangle}
\end{align}
the first term can be handled with lemma \ref{lemma3}. For the second term, let $A_T(m)=||\hat{R}_T^{m}f(\bm{x}_{uT-h,T})-\tilde{R}_T^{m}f(\tilde{\bm{x}}_{uT-h}(u))||$. Using the definitions of the remainder functions, we then have a recursive relation:
\begin{align*}
    A_T(m)& \leq A_T(m-1)+|\inp[\big]{\hat{R}^{m-1}f}{X_{\widehat{\mathcal{S}}_m,\cdot,T}}_{(T)}-\inp[\big]{\tilde{R}^{m-1}f}{\tilde{X}_{\widehat{\mathcal{S}}_m,uT-h}}|||X_{\widehat{\mathcal{S}}_m,uT-h,T}||\\
    &+ ||\inp[\big]{\tilde{R}^{m-1}f}{\tilde{X}_{\widehat{\mathcal{S}}_m,uT-h}}||||X_{\widehat{\mathcal{S}}_m,uT-h,T}-\tilde{X}_{\widehat{\mathcal{S}}_m,uT-h}(u)||\\
    &\leq A_T(m-1)+C(5/2)^m\zeta_T+C/T \text{ on the set } \mathcal{A}_T
\end{align*}
where the above inequality holds uniformly over all $u \in [0,1]$. The last inequality follows from local stationarity and lemma \ref{lemma2}. By the above recursive equation we obtain: $\sup_{u \in [0,1]}||\hat{R}_T^{m}f(\bm{x}_{uT-h,T})-\tilde{R}_T^{m}f(\tilde{\bm{x}}_{uT-h}(u))||\leq 3^{m}\zeta_TC \text{ on the set } \mathcal{A}_T$. If we choose $m=m_T \rightarrow \infty$ slow enough (e.g $m_T=o(log(T))$), then along with lemma \ref{lemma3} and (\ref{triangle})
we obtain:
\begin{align*}
    \sup_{u \in [0,1]}||\hat{R}_{T}^{m}f(\bm{x}_{uT-h,T})||=\sup_{u \in [0,1]}||f(\bm{x}_{uT-h,T})-\hat{F}_{lc}^{(M_T)}(u,\bm{x}_{uT-h,T})|| =o_p(1)
\end{align*}
\end{proof}


\newpage

\setcounter{page}{1}
\setcounter{section}{0}
\setcounter{table}{0}
\renewcommand{\theequation}{A.\arabic{equation}}
\setcounter{equation}{0}
\setcounter{lem}{3}

\title{\large{\bf Online Appendix to Boosting High Dimensional Predictive Regressions with Time Varying Parameters}}

\maketitle

\section{Online Appendix A: Additional Proofs}

\begin{proof}[Proof of Lemma 1]
  $ $\newline
We start with (i), and we bound:
\begin{align}
    &P(\sup_{u \in [0,1]}|\sum_{t=1}^{T}K_b(t/T-u)X_{j,t,T}X_{k,t,T}-E(\tilde{X}_{j,T}(u)\tilde{X}_{k,T}(u))|> S_T^{-\kappa})\\
    &\leq P\bigg(\sup_{u \in [0,1]}|\sum_{t=1}^{T}K_b(t/T-u)[X_{j,t,T}X_{k,t,T}-E(X_{j,t,T}X_{k,t,T})]|\label{variance}\\
    &+|\sup_{u \in [0,1]}|\sum_{t=1}^{T}K_b(t/T-u)E(X_{j,t,T}X_{k,t,T})-E(\tilde{X}_{j,T}(u)\tilde{X}_{k,T}(u))|> S_T^{-\kappa}\bigg)
\end{align}
We deal with the second term, which can be thought of as the bias. The product process $X_{j,t,T}X_{k,t,T}$ is locally stationary with the stationary approximation at rescaled time $t/T$ being $\tilde{X}_{j,t}(t/T)\tilde{X}_{k,t}(t/T)$. One can see this by noting:
\begin{align*}
||\tilde{X}_{j,T}(u)\tilde{X}_{k,T}(u))-\tilde{X}_{j,T}(v)\tilde{X}_{k,T}(v))||_{r/2} & \leq  (||\tilde{X}_{j,T}(u)||_{r}||\tilde{X}_{k,T}(u)-\tilde{X}_{k,T}(v)||_{r}  \\
 &+||\tilde{X}_{k,T}(u)||_{r}||\tilde{X}_{j,T}(u)-\tilde{X}_{j,T}(v)||_{r})\nonumber \\
 & \leq C(|u-v|)
\end{align*} 
which holds uniformly in $u,v \in [0,1]$. We can employ the same techniques to that $||X_{j,t,T}X_{k,t,T}-\tilde{X}_{j,t}(t/T)\tilde{X}_{k,t}(t/T)|| \leq C(T^{-1})$.

Therefore by local stationarity we obtain:
\begin{align}
    |\sup_{u \in [1-b_T,1]}\sum_{t=1}^{T}K_b(t/T-u)E(X_{j,t,T}X_{k,t,T})-E(\tilde{X}_{j,T}(u)\tilde{X}_{k,T}(u))| \leq O(S_T/T)
    \label{bias}
\end{align}
Note that if $u$ is an interior point (i.e $u \in (b_T,1-b_T)$ where $b_T$ is the bandwidth), the bound improves to $O((S_T/T)^2)$. Let $\tilde{\sigma}_{jk}(u)=E(\tilde{X}_{j,t}(u)\tilde{X}_{k,t}(u))$, then by condition 5.3, we obtain:
\begin{align*}
    \tilde{\sigma}_{jk}(t/T)=\tilde{\sigma}_{jk}(u)+\tilde{\dot{\sigma}}_{jk}(u)(t/T-u)+O(b_T^2)
\end{align*}
where $\tilde{\dot{\sigma}}_{jk}(u)$ refers to the derivative of the covariance matrix w.r.t the rescaled time index. This gives us: 
\begin{align}
    |\sup_{u \in (b_T,1-b_T)}|\sum_{t=1}^{T}K_b(t/T-u)E(X_{j,t,T}X_{k,t,T})-E(\tilde{X}_{j,T}(u)\tilde{X}_{k,T}(u))| \leq O((S_T/T)^2)
\end{align}

Now we deal with the term (\ref{variance}), note that the functional dependence measure of the stationary approximation Using this we compute the functional dependence measure of $\tilde{X}_{j,T}(u)\tilde{X}_{k,T}(u))$ as: 
\begin{align}
 \sup_{u \in [0,1]}||\tilde{X}_{j,T}(u)\tilde{X}_{k,T}(u))-\tilde{X}_{j,T}^{*}(u)\tilde{X}_{k,T}^{*}(u))||_{r/2} & \leq  \sup_{u \in [0,1]}(||\tilde{X}_{j,T}(u)||_{r}||\tilde{X}_{k,T}(u)-\tilde{X}_{k,T}(u)^{*}||_{r}  \\
 &+\sup_{u \in [0,1]}||\tilde{X}_{k,T}(u)||_{r}||\tilde{X}_{j,T}(u)-\tilde{X}_{j,T}(u)^{*}||_{r})\nonumber \label{funxdepx}
\end{align} 

Therefore, $\tilde{X}_{j,T}(u)\tilde{X}_{k,T}(u)$ has a finite cumulative dependence measure by the weak dependence condition imposed on $\bm{\tilde{x}}_{t}(u)$.
Taking into account (\ref{bias}), and the above we can then apply theorem 2.7 (iii) in \cite{dahlhaus2017towards} to obtain:
\begin{align*}
    P(\sup_{u \in [0,1]}|\sum_{t=1}^{T}K_b(t/T-u)X_{j,t,T}X_{k,t,T}-&E(\tilde{X}_{j,T}(u)\tilde{X}_{k,T}(u))|
    > S_T^{-\kappa})\\
    &\leq O(TS_T^{-r/2+r\kappa/2}) +O(\exp(-S_T^{1-2\kappa}))
\end{align*}
Applying the union bound then completes the proof.
\\
\\
For (ii), we proceed similarly. Given that $E(X_{j,t-h,T}\epsilon_{t,T})=0, \text{ } \forall j$. We have that $E(\tilde{X}_{j,t-h}(u)\tilde{\epsilon}_{T}(u))=O(T^{-1})$. We bound:
\begin{align}
    &P\bigg(\sup_{u \in [0,1]}|\sum_{t=1}^{T}K_b(t/T-u)X_{j,t-h,T}\epsilon_{t,T}-E(\tilde{X}_{j,t-h}(u)\tilde{\epsilon}_{T}(u))|> S_T^{-\kappa}\bigg)\\
    &\leq P\bigg(\sup_{u \in [0,1]}|\sum_{t=1}^{T}K_b(t/T-u)X_{j,t-h,T}\epsilon_{t,T}|
    > O(S_T^{-\kappa})\bigg)
\end{align}
 Now we apply the same procedure as previously. We have that: 
\begin{align}
 \sup_{u \in [0,1]}||\tilde{X}_{j,t}(u)\tilde{\epsilon}_{t}(u))-\tilde{X}_{j,t}^{*}(u)\tilde{\epsilon}_{t}^{*}(u))||_{\tau} & \leq  \sup_{u \in [0,1]}(||\tilde{X}_{j,t}(u)||_{r}||\tilde{\epsilon}_{t}(u)-\tilde{\epsilon}_{t}(u)^{*}||_{q}  \\
 &+\sup_{u \in [0,1]}||\tilde{\epsilon}_{t}(u)||_{q}||\tilde{X}_{j,t}(u)-\tilde{X}_{j,t}(u)^{*}||_{r})\nonumber \label{funxdepe}
\end{align} 
This has a finite cumulative functional dependence measure by the weak dependence conditions imposed. Once again, by applying theorem 2.7 (iii) in \cite{dahlhaus2017towards} we obtain:
\begin{align}
    &P(\sup_{u \in [0,1]}|\sum_{t=1}^{T}K_b(t/T-u)X_{j,t-h,T}\epsilon_{t,T}-E(\tilde{X}_{j,t-h}(u)\tilde{\epsilon}_{T}(u))|> S_T^{-\kappa}) \leq O(TS_T^{-\kappa+\tau\kappa})+O(\exp(S_T^{1-2\kappa}))
\end{align}
Applying the union bound gives the final result.
\\
\\
For (iii), we note that $f(t/T,\bm{x}_{t,T})\equiv f(\bm{x}_{t,T})$ is a locally stationary process with stationary approximation $f(t/T,\bm{\tilde{x}}_{t}(t/T))$. And the stationary approximation has cumulative functional dependence measure $\sup_{u \in [0,1]}|\bm{\beta}(u)|\Phi_{0,r}^{\bm{x}}$. We can then compute the cumulative dependence measure of the product process $f(\bm{\tilde{x}}_{t}(t/T))\tilde{X}_{j,t}(t/T)$ similarly as for part (i). We then obtain
\begin{align*}
    P(\sup_{u \in [0,1]}|\sum_{t=1}^{T}K_b(t/T-u)X_{j,t,T}f(\bm{x}_{t,T})&-E(\tilde{X}_{j,T}(u)f(\tilde{\bm{x}}_{T}(u)))|> S_T^{-\kappa})\\
    &\leq O(TS_T^{-r/2+r\kappa/2}) +O(\exp(-S_T^{1-2\kappa}))
\end{align*}
We can handle the bias term, in the same way we did for part(i), given that the product process $f(\bm{{x}}_{t,T})X_{j,t,T}$ is locally stationary with the stationary approximation being twice differentiable w.r.t to the rescaled time index. Taking the union bound then gives us the result.
\\
\\
The result for (iv) follows immediately from parts (ii) and (iii).
\end{proof}

\begin{proof}[Proof of Lemma 2]
  $ $\newline
Recall that:
\begin{align*}\tilde{R}_T^{m}f(\tilde{\bm{x}}_{T-h}(u))=\tilde{R}^{m-1}f(\tilde{\bm{x}}_{T-h}(u))-\inp[\big]{\tilde{R}^{m-1}f(\tilde{\bm{x}}_{T-h}(u))}{\tilde{X}_{\widehat{\mathcal{S}}_m,T-h}(u)}\tilde{X}_{\widehat{\mathcal{S}}_m,T-h}(u),\\
\hat{R}_T^{m}f(\bm{x}_{T-h,T})=\hat{R}^{m-1}f(\bm{x}_{T-h,T})-\inp[\big]{\hat{R}^{m-1}f}{X_{\widehat{\mathcal{S}}_m,\cdot,T}}_{(T)}X_{\widehat{\mathcal{S}}_m,T-h,T}
\end{align*}
We denote: $A_T(m,j)=\inp[\big]{\hat{R}^{m-1}f}{X_{j,\cdot,T}}_{(T)}-\inp[\big]{\tilde{R}^{m-1}f}{\tilde{X}_{j,uT-h}}$. We proceed with a recursive analysis. Note that for $m=0$, the result follows from lemma 1. By using the above definitions we get the following recursive relation:
\begin{align*}
    A_T(m,j)& \leq A_T(m-1,j)- (\inp[\big]{\tilde{R}^{m-1}f}{\tilde{X}_{\widehat{\mathcal{S}}_m,T-h}})\bigg(\inp{X_{\widehat{\mathcal{S}}_m,\cdot,T}}{X_{j,\cdot,T}}_{(T)}-\inp{\tilde{X}_{\widehat{\mathcal{S}}_m,uT-h}(u)}{\tilde{X}_{j,uT-h}(u)}\bigg)\\
    &-\bigg(\inp[\big]{\hat{R}^{m-1}f}{X_{\widehat{\mathcal{S}}_m,\cdot,T}}_{(T)}-\inp[\big]{\tilde{R}^{m-1}f}{\tilde{X}_{\widehat{\mathcal{S}}_m,uT-h}}\bigg)\inp{X_{\widehat{\mathcal{S}}_m,\cdot,T}}{X_{j,\cdot,T}}_{(T)}\\
    &=A_T(m-1,j)-I_{T,m}(j)-II_{T,m}(j)
\end{align*}

We then have $\sup_{u \in [0,1]}\sup_j A_T(m,j) \leq \sup_{u \in [0,1]}, \sup_j A_T(m-1,j) + \sup_{u \in [0,1]}\sup_{j}|I_{T,m}(j)| + \sup_{u \in [0,1]} \sup_j |II_{T,m}(j)| $

Now we have that $\sup_{u \in [0,1]}\sup_{j}|I_{T,m}(j)| \leq \sup_{u \in [0,1]}||f(\tilde{\bm{x}}_{T-h}(u))||\zeta_T$, by lemma 1, and the norm reducing property of the remainder functions. Similarly, $\sup_{u \in [0,1]}\sup_{j}|II_{T,m}(j)| \leq (1+\zeta_T) \sup_{u \in [0,1]}\sup_{j}A_T(m-1,j)$. The rest of the proof follows from \cite{Buhlmann2006}.

\end{proof}

\begin{proof}[Proof of Lemma 3]
  $ $\newline

The proof closely follows the one laid out in \cite{Buhlmann2006}, and the introduction of the rescaled time index $u$ does not present any additional difficulties. Therefore we omit the details here.
\end{proof}

Before proceeding to the proof of theorem 1(ii), we first prove corollary 2.

\begin{proof}[Proof of Corollary 2]
  $ $\newline
  
We only need to change lemma 1, from the proof of theorem 1(a). The rest of the proof is essentially the same, if we replace the locally stationary variables with stationary ones. We borrow some arguments from the proof of theorem 2 in \cite{Yousuf2017}. The main technical tool we use is theorem 3 in \cite{WuandWu2016}. For that, we first define the \emph{predictive dependence measure} introduced by \cite{Wu2005}. The predictive dependence measure for stationary univariate and multivariate processes is defined respectively as:
\begin{align}\theta_{q}(\epsilon_i)=||\textrm{E}\left(\epsilon_i|\mathcal{F}_{0}\right)-\textrm{E}\left(\epsilon_i|\mathcal{F}_{-1}\right)||_q,\nonumber\\
\theta_{q}(X_{j,i})=||\textrm{E}\left(X_{j,i}|\mathcal{H}_{0}\right)-\textrm{E}\left(X_{j,i}|\mathcal{H}_{-1}\right)||_q .\end{align}
With the cumulative predictive dependence measures defined as:
\begin{align*}\Theta_{0,q}(\bm{x})=\max_{j \leq p_n}\sum_{i=0}^{\infty}\theta_{q}(X_{ij}), \textrm{ and } \Theta_{0,q}(\bm{\epsilon})=\sum_{i=0}^{\infty}\theta_{q}(\epsilon_i).\end{align*} 

By theorem 1 in \cite{Wu2005}, we have $\Theta_{0,q}(\bm{x}) \leq \Phi_{0,q}^{\bm{x}}$, and similarly $\Theta_{0,q}(\bm{\epsilon}) \leq \Delta_{0,q}^{\epsilon}$. Where $\Phi_{0,q}^{\bm{x}},\Delta_{0,q}^{\epsilon}$ represent the cumulative functional dependence measures. From Section 2 in \cite{WuandWu2016}: $||X_{j,i}||_{q} \leq \Phi_{0,q}^{\bm{x}}$, and $||\epsilon_{i}||_{q} \leq \Delta_{0,q}^{\bm{\epsilon}}$. We only discuss parts (i) and (ii) from lemma 1, the others can be done similarly. We now define $\bm{G}_{jk}=(G_{1,jk},\ldots,G_{T,jk})$ where $G_{i,jk}=X_{j,i}X_{k,i}$, and let $\bm{R}_{j}=(R_{1,j},\ldots,R_{T,j})$ where $R_i= X_{j,i}\epsilon_i$. We need to bound the sums: $\sum_{i=1}^T (G_{i,jk}-E(G_{i,jk}))/T$ and $\sum_{i=1}^T R_{i,j}/T$.

As previously, we have (by Holder's inequality)
\begin{equation}
\sum_{t=0}^{\infty} ||X_{j,t}X_{k,t}-X_{j,t}^{*}X_{k,t}^{*}||_{q}\leq \sum_{t=0}^{\infty} (||X_{j,t}||_{2q}||X_{k,t}-X_{k,t}^{*}||_{2q} + ||X_{k,t}||_{2q}||X_{j,t}-X_{j,t}^{*}||_{2q}) \leq 2\Phi_{0,2q}^2(\bm{x})  \end{equation}
Using these, along with Condition 4.5, we obtain:
\begin{equation} \sup_{q \geq 4} q^{-2\tilde{\alpha}_x}\Theta_{q}(\bm{G}_{jk}) \leq \sup_{q \geq 4} 2q^{-2\tilde{\alpha}_x}\Phi_{0,2q}^{2}(\bm{x}) < \infty \end{equation}
Combining the above and using Theorem 3 in \cite{WuandWu2016}, we obtain:
\begin{equation} P\left(\bigg|\sum_{t=1}^T (G_{t,jk}-E(G_{t,jk}))\bigg| > \frac{c T^{1-\kappa}}{2}\right) \leq  C\exp\left(-\frac{T^{1/2-\kappa}}{\upsilon_x^2}\right)^{\tilde{\alpha}} \end{equation}
Similarly, using the same procedure we obtain:
\begin{equation}P\left(\bigg|\sum_{t=1}^T R_{t,j}\bigg| > \frac{c T^{1-\kappa}}{2}\right) \leq C\exp\left(-\frac{T^{1/2-\kappa}}{\upsilon_x\upsilon_{\epsilon}}\right)^{\tilde{\alpha}'} \end{equation}

We can use the same procedure to get the corresponding bounds for the terms in lemma 1 (iii) and (iv). Now using the above bounds and following the steps in the proof of Theorem 1 we obtain the result.

\end{proof}

\begin{proof}[Proof of Theorem 1(b)]
  $ $\newline
The proof of the LL-Boost is more complicated than the LC-Boost case due to the additional linear term. Fortunately, the population version stays the same between both versions. This allows us to use the same framework as previously, where we relied on Temlyakov's result on weak greedy algorithms. In order, to simplify the presentation we focus on showing the result for $u=T/T$, and using the uniform kernel. One can use similar steps as used in the proof of theorem 1(a), in order to obtain the result for $\sup_{u \in [0,1]}$ and general kernels. We do need to make a number of changes from the proof of theorem 1(a), and we start by introducing the following notation: let $\bm{Z}_{j,t,T}=(X_{j,t,T},X_{j,t,T}(t/T-u))$, and let:
\begin{align*}
    &\hat{\bm{h}}(Y_{\cdot,T},X_{j,\cdot,T})=(\hat{h}_1(Y_{\cdot,T},X_{j,\cdot,T}),\hat{h}_2(Y_{\cdot,T},X_{j,\cdot,T}))\\
    &=\argmin_{\bm{h}}S_{T}^{-1}\sum_{t=T-S_T}^{T}(Y_{t,T}-h_1X_{j,t-h,T}-h_2X_{j,t-h,T}(t/T-u))^2\\
    &\hat{\bm{h}}(X_{k,\cdot,T},X_{j,\cdot,T})=(\hat{h}_1(X_{k,\cdot,T},X_{j,\cdot,T}),\hat{h}_2(X_{k,\cdot,T},X_{j,\cdot,T}))\\
    &=\argmin_{\bm{h}}S_{T}^{-1}\sum_{t=T-S_T}^{T}(X_{k,t-h,T}-h_1X_{j,t-h,T}-h_2X_{j,t-h,T}(t/T-u))^2
\end{align*}
represent the estimated local linear regression coefficients. The arguments to the function $h(\cdot,\cdot)$ refer to the dependent and independent variables respectively. These functions are linear functions of the first argument. We also let $\bm{h}(\tilde{Y},\tilde{X}_{j}),\bm{h}(\tilde{X}_{k},\tilde{X}_{j})$ represent the population version of these coefficients when we replace the locally stationary series in the above equations with their stationary approximations.\footnote{We note that $\bm{h}$ is also a function of the rescaled time point $u$, but we ignore this for now.} We then define our selectors as:
\begin{align*}
    \hat{\mathcal{S}}_1=\argmax_{j}||\bm{\hat{h}}(Y_{\cdot,T},X_{j,\cdot,T}){\bm{Z}_{j,\cdot,T}}||_{(T)}, \ldots, \hat{\mathcal{S}}_m=\argmax_{j}||\bm{\hat{h}}(\hat{R}_{T}^{m}f,X_{j,\cdot,T}){\bm{Z}_{j,\cdot,T}}||_{(T)}
\end{align*}
Where the sample remainder functions are defined as:
\begin{align*}
    &\hat{R}_T^{0}f(\bm{x}_{T-h,T})=f(\bm{x}_{T-h,T}),\\
    &\hat{R}_T^{m}f(\bm{x}_{T-h,T})=\hat{R}_T^{m-1}f(\bm{x}_{T-h,T})-\bm{\hat{h}}(\hat{R}_{T}^{m}f,X_{\hat{\mathcal{S}}_m,\cdot,T}){\bm{Z}_{\hat{\mathcal{S}}_m,T-h,T}}, m=1,2,\ldots
\end{align*}
Therefore, $\hat{R}_{T}^{m}f(\bm{x}_{T-h})=f(\bm{x}_{T-h})-\hat{F}_{ll}^{(M_T)}(u,\bm{x}_{T-h,T})$, is the difference between $f(\bm{x}_{T-h})$ and its LL-Boost estimate. Now the semipopulation version has the same form as in theorem 1(a) except it uses the selected base learners $\hat{\mathcal{S}}_m$ as defined above. The strategy of the proof is similar to the local constant case: first we establish that the selectors $\hat{\mathcal{S}}_m$ satisfy a finite sample analogue of equation (\ref{relaxed}) in our main text, which allows us to apply Temlyakov's result (\cite{temlyakov2000weak}) to the semipopulation version. A key step to doing this is to bound $\sup_{j \leq p_T} ||\bm{Z}_{j,\cdot,T}\bm{\hat{h}}(\hat{R}_T^{m}f,X_{j,\cdot,T})-\inp[\big]{\tilde{R}^{m-1}f(\tilde{\bm{x}}_{T-h}(u))}{\tilde{X}_{j,T-h}(u)}\tilde{X}_{j,\cdot}(u)||_{(T)}$, and we establish this is in lemma \ref{lemma4thm2}. Recall that:\\ $||\tilde{X}_{j,\cdot}(u)||^2_{(T)}=\sum_{t=1}^{T}K_b(t/T-u)\tilde{X}_{j,t}(u)^2$, and we use the uniform kernel.

In order to prove lemma \ref{lemma4thm2} and bound $||\hat{R}_T^{m}f(\bm{x}_{T-h,T})-\tilde{R}_T^{m}f(\tilde{\bm{x}}_{T-h}(u))||$ we need the following three lemmas. Recall that $Y_{t,T}=\alpha^{(m)}_j(t/T)X_{j,t-h,T}+\epsilon_{j,t,T}$, where\\ $\alpha_j(t/T)=E(\tilde{X}_{j,t-h}(t/T)\tilde{Y}_{T}(t/T))/E(\tilde{X}^2_{j,t-h}(t/T))$. We also define the following:
\begin{align}
X_{j,t,T}=\alpha_{jk}(t/T)X_{k,t-h,T}+\epsilon_{jk,t,T}, \text{ where } \alpha^{(m)}_j(t/T)=E(\tilde{X}_{j,t-h}(t/T)\tilde{X}_{k,t-h,T}(t/T))/E(\tilde{X}^2_{k,t-h}(t/T))
\end{align}

\begin{lem}
Under conditions 5.1, 5.2, 5.3, 5.4 and for $\kappa$ as defined in Theorem 1, the following hold:
\begin{enumerate}[noitemsep]
\item $\sup_{j,k \leq p_T}|S_T^{-1}\sum_{t=T-S_T}^{T}\bigg[X_{j,t,T}X_{k,t,T}(t/T-u)^{i}-E(X_{j,t,T}X_{k,t,T}(t/T-u)^{i})\bigg]| =O_p(S_T^{-\kappa+i}/T^{i})$ for $i=1,2$
\\
\item $\sup_{j \leq p_T}|S_T^{-1}\sum_{t=T-S_T}^{T}\bigg[X_{j,t,T}^{i_1}(t/T-u)^{i_2}-E(X_{j,t,T}^{i_1}(t/T-u)^{i_2})\bigg]| = \zeta_{T,i_1,i_2}=O_p(S_T^{-\kappa+i_2}/T^{i_2})$ for $i_1=1,2$ and $i_2=1,2,3$.
\\
\item $\sup_{j,k \leq p_T}|S_T^{-1}\sum_{t=T-S_T}^{T}X_{j,t-h,T}\epsilon_{jk,t,T}(t/T-u)^{i}-E(S_T^{-1}\sum_{t=T-S_T}^{T}X_{j,t-h,T}\epsilon_{jk,t,T}(t/T-u)^{i})|=O_p(S_T^{-\kappa+i}/T^{i})$ for $i=0,1$
\\
\item $\sup_{j \leq p_T}|S_T^{-1}\sum_{t=T-S_T}^{T}X_{j,t-h,T}\epsilon_{j,t,T}(t/T-u)^{i}-E(S_T^{-1}\sum_{t=T-S_T}^{T}X_{j,t-h,T}\epsilon_{j,t,T}(t/T-u)^{i})| =O_p(S_T^{-\kappa+i}/T^{i})$ for $i=0,1$
\\
\end{enumerate}
\label{lemma1thm2}
\end{lem}

\begin{lem}
Under conditions  5.1, 5.2, 5.3, 5.4, and for $\kappa$ as defined in Theorem 1, the following hold:
\begin{enumerate}[noitemsep]
\item $\sup_{j,k \leq p_T}|\hat{\bm{h}}(X_{j,\cdot,T},X_{k,\cdot,T})-\bm{h}(\tilde{X}_j,\tilde{X}_k)| = \zeta_{T,1}=O_p(S_T^{-\kappa})$
\\
\item $\sup_{j \leq p_T}|\hat{\bm{h}}(\epsilon_{\cdot,T},X_{j,\cdot,T})-\bm{h}(\tilde{\epsilon},\tilde{X}_j)| = \zeta_{T,2}=O_p(S_T^{-\kappa})$
\\
\item $\sup_{j \leq p_T}|\hat{\bm{h}}(f,X_{j,\cdot,T})-\bm{h}(\tilde{f},\tilde{X}_j)| = \zeta_{T,3}=O_p(S_T^{-\kappa})$
\\
\item $\sup_{j \leq p_T}|\hat{\bm{h}}(Y_{\cdot,T},X_{j,\cdot,T})-\bm{h}(\tilde{Y},\tilde{X}_j)| = \zeta_{T,4}=O_p(S_T^{-\kappa})$
\end{enumerate}
\label{lemma2them2}
\end{lem} 
We introduce the following notation for the next lemma. Let 
\begin{align*}
    &\hat{\bm{h}}(X_{k,\cdot,T}(\cdot/T-u),X_{j,\cdot,T})=(\hat{h}_1(X_{k,\cdot,T}(\cdot/T-u),X_{j,\cdot,T}),\hat{h}_2(X_{k,\cdot,T}(\cdot/T-u),X_{j,\cdot,T}))\\
    &=\argmin_{\bm{h}}S_{T}^{-1}\sum_{t=T-S_T}^{T}(X_{k,t-h,T}(t/T-u)-h_1X_{j,t-h,T}-h_2X_{j,t-h,T}(t/T-u))^2
\end{align*}
Recall that $\hat{\bm{h}}(X_{k,\cdot,T}(\cdot/T-u),X_{j,\cdot,T})=\hat{\bm{A}}*\hat{\bm{B}}$, where:
\begin{align*}
&\hat{\bm{A}}=\begin{bmatrix}%
    S_T^{-1}\sum_{t=T-S_T}^{T}X_{j,t-h,T}^2& S_T^{-1}\sum_{t=T-S_T}^{T}X_{j,t-h,T}^2(t/T-u)\\
    S_T^{-1}\sum_{t=T-S_T}^{T}X_{j,t-h,T}^2(t/T-u) & S_T^{-1}\sum_{t=T-S_T}^{T}X_{j,t-h,T}^2(t/T-u)^2
    \end{bmatrix}^{-1}\\
    \\
&\hat{\bm{B}}=\begin{bmatrix}
    S_T^{-1}\sum_{t=T-S_T}^{T}X_{j,t-h,T}X_{k,t-h,T}(t/T-u)  \\S_T^{-1}\sum_{t=T-S_T}^{T}X_{j,t-h,T}X_{k,t-h,T}(t/T-u)^2
\end{bmatrix}
\end{align*}
We then let $\bm{h}(X_{k,\cdot,T}(\cdot/T-u),X_{j,\cdot,T})=\bm{A}*\bm{B}$:
\begin{align*}
&\bm{A}=\begin{bmatrix}%
    S_T^{-1}\sum_{t=T-S_T}^{T}E(X_{j,t-h,T}^2)& S_T^{-1}\sum_{t=T-S_T}^{T}E(X_{j,t-h,T}^2)(t/T-u)\\
    S_T^{-1}\sum_{t=T-S_T}^{T}E(X_{j,t-h,T}^2)(t/T-u) & S_T^{-1}\sum_{t=T-S_T}^{T}E(X_{j,t-h,T}^2)(t/T-u)^2
    \end{bmatrix}^{-1}\\
    \\
&\bm{B}=\begin{bmatrix}
    S_T^{-1}\sum_{t=T-S_T}^{T}E(X_{j,t-h,T}X_{k,t-h,T})(t/T-u)  \\S_T^{-1}\sum_{t=T-S_T}^{T}E(X_{j,t-h,T}X_{k,t-h,T})(t/T-u)^2
\end{bmatrix}
\end{align*}

\begin{lem}
Under conditions 5.1, 5.2, 5.3, 5.4, and for $\kappa$ as defined in Theorem 1, the following hold:
\begin{enumerate}[noitemsep]
\item $\sup_{j,k \leq p_T}|\hat{h}_1(X_{k,\cdot,T}(\cdot/T-u),X_{j,\cdot,T})-h_1(X_{k,\cdot,T}(\cdot/T-u),X_{j,\cdot,T})| = \zeta_{T,5}=O_p(S_T^{-\kappa+1}/T)$
\item $\sup_{j,k \leq p_T}|\hat{h}_2(X_{k,\cdot,T}(\cdot/T-u),X_{j,\cdot,T})-h_2(X_{k,\cdot,T}(\cdot/T-u),X_{j,\cdot,T})| = \zeta_{T,6}=O_p(S_T^{-\kappa})$
\end{enumerate}
\label{lemma3thm2}
\end{lem}

Let $\zeta_T=\max(\zeta_{T,1},\zeta_{T,2},\zeta_{T,3},\zeta_{T,4},\zeta_{T,6})=O_p(S_T^{-\kappa})$ and denote by $\omega$ a realization of all $S_T$ sample points involved in estimation. The next lemma bounds the difference between the sample and population learners at step $m$.

\begin{lem}
Suppose conditions 5.1, 5.2, 5.3, 5.4 hold. Then for $\kappa$ as defined in Theorem 1 and on the set $\mathcal{A}_T=\{\omega: \zeta_T(w)< 1/2, \zeta_{T,5}(\omega)\leq S_T/T\}$, we have:
\begin{align*}
    \sup_{j \leq p_T} ||\bm{Z}_{j,\cdot,T}\bm{\hat{h}}(\hat{R}_T^{m}f,X_{j,\cdot,T})-\inp[\big]{\tilde{R}^{m}f(\tilde{\bm{x}}_{T-h}(u))}{\tilde{X}_{j,T-h}(u)}\tilde{X}_{j,T-h}(u)||_{(T)} \leq C^m(\zeta_T+S_T/T),
\end{align*}
where $C$ does not depend on $m,T$.
\label{lemma4thm2}
\end{lem}

\begin{lem}
Suppose the conditions needed for Theorem 1 hold, then for $m=m_T \rightarrow \infty$ slow enough we have:
\begin{align*}
||\tilde{R}_T^{m}f(\tilde{\bm{x}}_{T-h}(u))||=o_p(1)
\end{align*}
\label{lemma5thm2}
\end{lem}
With the above lemmas we are now ready to analyze the term: $\hat{R}_{T}^{m}f(\bm{x}_{T-h,T})=f(\bm{x}_{T-h,T})-\hat{F}_{ll}^{(M_T)}(u,\bm{x}_{T-h,T})$. By the triangle inequality we obtain:
\begin{align}
    ||\hat{R}_{T}^{m}f(\bm{x}_{T-h,T})|| \leq ||\tilde{R}_T^{m}f(\tilde{\bm{x}}_{T-h}(u))||+ ||\hat{R}_T^{m}f(\bm{x}_{T-h,T})-\tilde{R}_T^{m}f(\tilde{\bm{x}}_{T-h}(u))||
    \label{trianglethm2}
\end{align}
the first term can be handled with lemma \ref{lemma5thm2}. For the second term, let $A_T(m)=||\hat{R}_T^{m}f(\bm{x}_{T-h,T})-\tilde{R}_T^{m}f(\tilde{\bm{x}}_{T-h}(u))||$. Using the definitions of the remainder functions, we then have a recursive relation:
\begin{align*}
    A_T(m)& \leq A_T(m-1)+|\hat{h}_1(\hat{R}_{T}^{m}f,X_{\hat{\mathcal{S}}_m,\cdot,T})-\inp[\big]{\tilde{R}^{m-1}f}{\tilde{X}_{\widehat{\mathcal{S}}_m,T-h}}|||X_{\widehat{\mathcal{S}}_m,T-h,T}||\\
    &+ ||\inp[\big]{\tilde{R}^{m-1}f}{\tilde{X}_{\widehat{\mathcal{S}}_m,T-h}}||||X_{\widehat{\mathcal{S}}_m,T-h,T}-\tilde{X}_{\widehat{\mathcal{S}}_m,T-h}(u)||\\
    &\leq A_T(m-1)+C^m(\zeta_T+S_T/T)+O(1/T) \text{ on the set } \mathcal{A}_T
\end{align*}
Where the last inequality follows from local stationarity and lemma \ref{lemma4thm2}.\footnote{Although not the exact statement of \ref{lemma4thm2}, the proof handles the specific term we need.} By the above recursive equation we obtain: $||\hat{R}_T^{m}f(\bm{x}_{T-h,T})-\tilde{R}_T^{m}f(\tilde{\bm{x}}_{T-h}(u))||\leq C^{m}(\zeta_T+S_T/T)$. If we choose $m=m_T \rightarrow \infty$ slow enough (e.g $m_T=o(log(T))$), then along with lemma \ref{lemma5thm2} and (\ref{trianglethm2})
we obtain:
\begin{align*}
    ||\hat{R}_{T}^{m}f(\bm{x}_{T-h,T})||=||f(\bm{x}_{T-h,T})-\hat{F}_{ll}^{(M_T)}(u,\bm{x}_{T-h,T})|| =o_p(1)
\end{align*}

\end{proof}

\begin{proof}[Proof of Lemma \ref{lemma1thm2}]
  $ $\newline
We start with (i), and we bound:
\begin{align}
    &P(|S_T^{-1}\bigg[\sum_{t=T-S_T}^{T}X_{j,t,T}X_{k,t,T}(t/T-u)^{i}-E(\sum_{t=T-S_T}^{T}X_{j,t,T}X_{k,t,T}(t/T-u)^{i})\bigg]|> S_T^{-\kappa+i}/T^{i})
    \label{lem4a}
\end{align}
We have a sum similar to that in lemma 1, except we have weights $(t/T-u)^{i} $. Define the function $K_b^{*}(t/T-u)=K_b(t/T-u)(t/T-u)^i$, and given that $K_b(t/T-u)$ is of bounded variation we have that $K_b^{*}(t/T-u)$ is of bounded variation. Additionally, since we had shown in the proof of lemma 1 that the product process $X_{j,t,T}X_{k,t,T}$ is locally stationary and satisfies the weak dependence condition, we can use theorem 2.7 (iii) in \cite{dahlhaus2017towards} directly to obtain:
\begin{align*}
    (\ref{lem4a})\leq O(S_T^{-r/2+r\kappa/2+1}) +O(\exp(-S_T^{1-2\kappa}))
\end{align*}

Part (ii) can be handled similarly. For part (iii), we have that $\epsilon_{jk,t,T}=X_{j,t,T}-\alpha_{jk}(t/T)X_{k,t-h,T}$, therefore is locally stationary and its stationary approximation satisfies the weak dependence condition, and has $r$ finite moments. Therefore, we get the same result as for (iii). For (iv), note that by definition $\epsilon_{j,t,T}$ has $\min(r,q)$ finite moments, and is locally stationary. Now if we let $r_1=\min(r,q)$ we get the same result as for part (i) with $r_1$ instead of $r$
.
\end{proof}

\begin{proof}[Proof of Lemma \ref{lemma2them2}]
  $ $\newline

We mainly discuss the proof for part (i), the rest can be handled similarly. Note that $X_{k,t-h,T}=\bm{h}(\tilde{X}_j,\tilde{X}_k)\bm{Z}_{j,t-h,T}+\ddot{h}_1((\tilde{X}_j,\tilde{X}_k))(c)(t/T-u)^2+\epsilon_{jk,t,T}$. Where $\ddot{h}_1(\tilde{X}_j,\tilde{X}_k)(\cdot)$ refers to the second derivative of $h_1(\tilde{X}_j,\tilde{X}_k)(\cdot)$. Therefore we have:
\begin{align*}
    |\hat{\bm{h}}(X_{j,\cdot,T},X_{k,\cdot,T})-\bm{h}(\tilde{X}_j,\tilde{X}_k)|=\hat{\bm{A}}^{-1}*\hat{\bm{B}}+\hat{\bm{A}}^{-1}*\bm{\hat{C}}
\end{align*}
where
\begin{align*}
    &\hat{\bm{A}}^{-1}=\begin{bmatrix}%
    S_T^{-1}\sum_{t=T-S_T}^{T}X_{k,t-h,T}^2& S_T^{-1}\sum_{t=T-S_T}^{T}X_{k,t-h,T}^2(t/T-u)\\
    S_T^{-1}\sum_{t=T-S_T}^{T}X_{k,t-h,T}^2(t/T-u) & S_T^{-1}\sum_{t=T-S_T}^{T}X_{k,t-h,T}^2(t/T-u)^2
    \end{bmatrix}^{-1}\\
    \\
&\hat{\bm{B}}=\begin{bmatrix}
    S_T^{-1}\sum_{t=T-S_T}^{T}X_{k,t-h,T}(t/T-u)^2  \\S_T^{-1}\sum_{t=T-S_T}^{T}X_{j,t-h,T}(t/T-u)^3
\end{bmatrix}    \\
\\
&\hat{\bm{C}}=\begin{bmatrix}
    S_T^{-1}\sum_{t=T-S_T}^{T}X_{k,t-h,T}\epsilon_{jk,t-h,T}  \\S_T^{-1}\sum_{t=T-S_T}^{T}X_{j,t-h,T}\epsilon_{jk,t-h,T}(t/T-u)
\end{bmatrix}\\
\\    
   \text{And let:  } &\bm{A}^{-1}=\begin{bmatrix}%
    S_T^{-1}\sum_{t=T-S_T}^{T}E(X_{k,t-h,T}^2)& S_T^{-1}\sum_{t=T-S_T}^{T}E(X_{k,t-h,T}^2)(t/T-u)\\
    S_T^{-1}\sum_{t=T-S_T}^{T}E(X_{k,t-h,T}^2)(t/T-u) & S_T^{-1}\sum_{t=T-S_T}^{T}E(X_{k,t-h,T}^2)(t/T-u)^2
    \end{bmatrix}^{-1}\\
    \\
&\bm{B}=\begin{bmatrix}
    S_T^{-1}\sum_{t=T-S_T}^{T}E(X_{k,t-h,T})(t/T-u)^2  \\S_T^{-1}\sum_{t=T-S_T}^{T}E(X_{j,t-h,T})(t/T-u)^3
\end{bmatrix}    \\
\\
&\bm{C}=\begin{bmatrix}
    S_T^{-1}\sum_{t=T-S_T}^{T}E(X_{k,t-h,T}\epsilon_{jk,t-h,T})  \\S_T^{-1}\sum_{t=T-S_T}^{T}E(X_{j,t-h,T}\epsilon_{jk,t-h,T})(t/T-u)
\end{bmatrix}
\end{align*}
Given $\hat{\bm{A}}$ is a 2 by 2 matrix we can calculate its inverse directly:
\begin{align*}
    &\hat{\bm{A}}^{-1}=\hat{a}_0^{-1}\begin{bmatrix}%
    S_T^{-1}\sum_{t=T-S_T}^{T}X_{k,t-h,T}^2(t/T-u)^2& S_T^{-1}\sum_{t=T-S_T}^{T}-X_{k,t-h,T}^2(t/T-u)\\
    S_T^{-1}\sum_{t=T-S_T}^{T}-X_{k,t-h,T}^2(t/T-u) & S_T^{-1}\sum_{t=T-S_T}^{T}X_{k,t-h,T}^2
    \end{bmatrix}\\
    \\
    & \text{where } \hat{a}_0=[\hat{A}_{11}\hat{A}_{22}-\hat{A}_{12}\hat{A}_{21}]
\end{align*}

We first handle $\hat{h}_1(X_{j,\cdot,T},X_{k,\cdot,T})$, and $\hat{h}_2(X_{j,\cdot,T},X_{k,\cdot,T})$ can be handled similarly. From the above equations we obtain:
\begin{align*}
    &P(|\hat{h}_1(X_{j,\cdot,T},X_{k,\cdot,T})-h_1(X_{j,\cdot,T},X_{k,\cdot,T})| > S_T^{-\kappa}) \\
    &\leq P(|\hat{A}^{-1}_{11}\hat{B}_{1}+\hat{A}^{-1}_{12}\hat{B}_{2}-(A^{-1}_{11}B_{1}+A^{-1}_{12}B_{2})|> cS_T^{-\kappa})\\
    &\leq P(|\hat{A}^{-1}_{11}\hat{B}_{1}-A^{-1}_{11}B_{1}|> cS_T^{-\kappa})+ P(|\hat{A}^{-1}_{12}\hat{B}_{2}-A^{-1}_{12}B_{2}|> cS_T^{-\kappa})
\end{align*}
For the first term, we let:
\begin{align*}
\hat{Q}_1=\hat{B}_1* S_T^{-1}\sum_{t=T-S_T}^{T}X_{k,t-h,T}^2(t/T-u)^2, \text{ and } Q_1=B_1* S_T^{-1}\sum_{t=T-S_T}^{T}E(X_{k,t-h,T}^2)(t/T-u)^2
\end{align*}
we then have: $|\hat{A}^{-1}_{11}\hat{B}_{1}-A^{-1}_{11}B_{1}|=|\hat{Q}_{1}\hat{a}_0^{-1}-Q_{1}a_0^{-1}|=|(\hat{a}_0^{-1}-a_0^{-1})(\hat{Q}_1-Q_1)+(\hat{Q}_1-Q_1)a_0^{-1}+(\hat{a}_0^{-1}-a_0^{-1})Q_1|$. We have that $a_0=O(S_T^2/T^2)$, and $Q=O(S_T^4/T^4)$. Therefore:
\begin{align} P(|\hat{A}^{-1}_{11}\hat{B}_{1}-A^{-1}_{11}B_{1}|> cS_T^{-\kappa})&\leq P(|(\hat{a}_0^{-1}-a_0^{-1})(\hat{Q}_1-Q_1)|>c_2n^{-\kappa}/3)\label{triple1}\\
 &+P(|(\hat{Q}_1-Q_1)a_0^{-1}|>cS_T^{-\kappa}/3|)\label{triple2}\\
 &+P(|(\hat{a}_0^{-1}-a_0^{-1})Q_1|>cS_T^{-\kappa}/3).\label{triple3}
\end{align}
For the RHS of (\ref{triple1}), we obtain:
\begin{align*} P(|(\hat{a}_0^{-1}-a_0^{-1})(\hat{Q}_1-Q_1)|>cn^{-\kappa}/3) &\leq P(|\hat{Q}_1-Q_1|>CS_T^{-\kappa/2})\\
&+ P(|\hat{a}_0^{-1}-a_0^{-1}|>CS_T^{-\kappa/2}).\end{align*}
Therefore we can focus on the terms (\ref{triple2}),(\ref{triple3}). We can handle (\ref{triple2}) directly using the fact that $a_0=O(S_T^2/T^2)$ along with lemma \ref{lemma2them2}. For (\ref{triple3}), note that $Q=O(S_T^4/T^4)$, and $|\hat{a}_0^{-1}-a_0^{-1}|=(\hat{a}_0-a_0)/(\hat{a}_0a_0)$.  We then obtain:
\begin{align*}
    (\ref{triple3}) \leq P(|\hat{a}_0-a_0|> CS_T^{-\kappa}D_TS_T^2/T^2)+P(|\hat{a}_0|< D_T)
\end{align*}
We can now choose $D_T\leq \min_{k \leq p_T}a_0*(log(S_T))^{-1}$. And we obtain the bound by applying lemma \ref{lemma1thm2}. For 
$P(|\hat{A}^{-1}_{12}\hat{B}_{2}-A^{-1}_{12}B_{2}|> cS_T^{-\kappa})$ we obtain a bound in similar fashion. Applying the union bound gives us the result.
\end{proof}

The proof for lemma \ref{lemma3thm2} can be obtained similarly to the proof of lemma \ref{lemma2them2}, therefore we omit the details.

\begin{proof}[Proof of Lemma \ref{lemma4thm2}]
  $ $\newline

We have that:
  \begin{align}
  &\sup_j||\bm{Z}_{j,\cdot,T}\bm{\hat{h}}(\hat{R}_T^{m}f,X_{j,\cdot,T})-\inp[\big]{\tilde{R}^{m-1}f(\tilde{\bm{x}}_{T-h}(u))}{\tilde{X}_{j,T-h}(u)}\tilde{X}_{j,T-h}(u)||_{(T)}\\
  &\leq \sup_j ||X_{j,\cdot,T}\hat{h}_1(\hat{R}_T^{m}f,X_{j,\cdot,T})-\inp[\big]{\tilde{R}^{m-1}f(\tilde{\bm{x}}_{T-h}(u))}{\tilde{X}_{j,T-h}(u)}\tilde{X}_{j,T-h}(u)||_{(T)} \label{lemma4I}\\
  &+ \sup_j ||\hat{h}_2(\hat{R}_T^{m}f,X_j)X_{j,\cdot,T}(\cdot/T-u)||_{(T)}
  \label{lemma4II}
  \end{align}

We will deal with term (\ref{lemma4II}) later. We first have:
\begin{align*}
(\ref{lemma4I}) \leq & \textrm{ }\sup_j||X_{j,\cdot,T}||_{(T)}\sup_j|\hat{h}_1(\hat{R}_T^{m}f,X_{j,\cdot,T})-\inp[\big]{\tilde{R}^{m}f(\tilde{\bm{x}}_{T-h}(u))}{\tilde{X}_{j,T-h}(u)}| \\
&+ \sup_j|\inp[\big]{\tilde{R}^{m}f(\tilde{\bm{x}}_{T-h}(u))}{\tilde{X}_{j,T-h}(u)}|||\tilde{X}_{j,\cdot}(u)-X_{j,\cdot,T}||_{(T)}
\end{align*}
By local stationarity and using the same techniques as the previous lemmas, we can show the term $||\tilde{X}_{j,\cdot}(u)-X_{j,t,T}||_{(T)}=O(S_T/T)$ on the set $\mathcal{A}_T$. Therefore we focus on the term $A_T(m,j)=|\hat{h}_1(\hat{R}_T^{m}f,X_{j,\cdot,T})-\inp[\big]{\tilde{R}^{m}f(\tilde{\bm{x}}_{T-h}(u))}{\tilde{X}_{j,T-h}(u)}|$. Note that by definition\\  $\inp[\big]{\tilde{R}^{m}f(\tilde{\bm{x}}_{T-h}(u))}{\tilde{X}_{j,T-h}(u)}=h_1(\tilde{R}^{m}f(\tilde{\bm{x}}_{T-h}(u)),\tilde{X}_{j,T-h}(u))$.

Now using a similar expansion as in the proof of theorem 1(a), we obtain:
\begin{align*}
    A_T(m,j) &= A_T(m-1,j) - (\inp[\big]{\tilde{R}^{m-1}f}{\tilde{X}_{\widehat{\mathcal{S}}_m,T-h}})\bigg(\hat{h}_1(X_{\widehat{\mathcal{S}}_m,\cdot,T},X_{j,\cdot,T})-h_1(\tilde{X}_{\widehat{\mathcal{S}}_m,uT-h}(u),\tilde{X}_{j,uT-h}(u))\bigg)\\
    &-\bigg(\hat{h}_1(\hat{R}^{m-1}f,X_{\widehat{\mathcal{S}}_m,\cdot,T})-\inp[\big]{\tilde{R}^{m-1}f}{\tilde{X}_{\widehat{\mathcal{S}}_m,uT-h}}\bigg)\hat{h}_1(X_{\widehat{\mathcal{S}}_m,\cdot,T},X_{j,\cdot,T})\\
    &- \hat{h}_2(\hat{R}^{m-1}f,X_{\widehat{\mathcal{S}}_m,\cdot,T})\hat{h}_1(X_{\widehat{\mathcal{S}}_m,\cdot,T}(\cdot/T-u),X_{j,\cdot,T})\\
    &=A_T(m-1,j)-I_{T,m}(j)-II_{T,m}(j)-III_{T,m}(j)
\end{align*}

As previously,we have that $\sup_{j}|I_{T,m}(j)| \leq |f(\tilde{\bm{x}}_{T-h}(u))||\zeta_T$, by lemma \ref{lemma2them2}, and the norm reducing property of the remainder functions. Similarly, $\sup_{j}|II_{T,m}(j)| \leq (1+\zeta_T) \sup_{j}A_T(m-1,j)$. We have to deal with term $\sup_j III_{T,m}(j)$. Recall that by definition and on the set $\mathcal{A}_T$:
\begin{align*}
    \sup_j |A_T(0,j)| = \sup_{j} ||\hat{h}_1(Y_{\cdot,T},X_{j,\cdot,T})-\inp[\big]{\tilde{Y}}{\tilde{X}_{j,T-h}(u)}|| \leq \zeta_T
\end{align*}

On the set $\mathcal{A}_T$, we have $\hat{h}_1(X_{\widehat{\mathcal{S}}_m,\cdot,T}(\cdot/T-u),X_{j,\cdot,T})=O(S_T/T)$, so if we can obtain a bound for the term $\sup_j|\hat{h}_2(\hat{R}^{m-1}f,X_{j,\cdot,T})|$, we have a recursive relationship and we can then use the same procedure as in the proof of theorem 1 (a) to obtain the answer. 

To start, recall that:
\begin{align*}
    \hat{R}_T^{1}f(\bm{x}_{t-h,T})=f(\bm{x}_{t-h,T})-\hat{h}_1(f,X_{\hat{\mathcal{S}}_1,\cdot,T}){X_{\hat{\mathcal{S}}_1,t-h,T}}-\hat{h}_2(f,X_{\hat{\mathcal{S}}_1,\cdot,T}){X_{\hat{\mathcal{S}}_1,t-h,T}})(t/T-u)
\end{align*}
therefore:
\begin{align*}
    |\hat{h}_2(\hat{R}_T^1f,X_{\hat{\mathcal{S}}_2,\cdot,T})| &\leq  |\hat{h}_2(f,X_{\hat{\mathcal{S}}_2,\cdot,T})|+|\hat{h}_1(f,X_{\hat{\mathcal{S}}_1,\cdot,T})||\hat{h}_2(X_{\hat{\mathcal{S}}_1,\cdot,T},X_{\hat{\mathcal{S}}_2,\cdot,T})|\\
    &+|\hat{h}_2(f,X_{\hat{\mathcal{S}}_1,\cdot,T})||\hat{h}_2(X_{\hat{\mathcal{S}}_1,\cdot,T}(\cdot/T -u ),X_{\hat{\mathcal{S}}_2,\cdot,T})|
\end{align*}
On the set $\mathcal{A}_T$, we have that $\sup_j|\hat{h}_2(\hat{R}_T^1f,X_{j,\cdot,T})| \leq |C+\zeta_T|+|C+\zeta_T|^2 + |C+\zeta_T|^2 = O(|C+\zeta_T|^2)$ and we can show the same for $\sup_j|\hat{h}_1(\hat{R}_T^1f,X_{j,\cdot,T})|$.

Now for general $m$, we have the following:
\begin{align}
    \sup_j|\hat{h}_2(\hat{R}_T^{m}f,X_{j,\cdot,T})| &\leq  \sup_j|\hat{h}_2(\hat{R}_T^{m-1}f,X_{j,\cdot,T})|+\sup_j|\hat{h}_1(\hat{R}_T^{m-1}f,X_{j,\cdot,T})||\hat{h}_2(X_{\hat{\mathcal{S}}_m,\cdot,T},X_{j,\cdot,T})| \nonumber\\
    &+\sup_j|\hat{h}_2(\hat{R}_T^{m-1}f,X_{j,\cdot,T})||\hat{h}_2(X_{\hat{\mathcal{S}}_m,\cdot,T}(\cdot/T -u ),X_{j,\cdot,T})|\label{h2}
\end{align}

and additionally:
\begin{align}
    \sup_j|\hat{h}_1(\hat{R}_T^{m}f,X_{j,\cdot,T})| &\leq  \sup_j|\hat{h}_1(\hat{R}_T^{m-1}f,X_{j,\cdot,T})|+\sup_j|\hat{h}_1(\hat{R}_T^{m-1}f,X_{j,\cdot,T})||\hat{h}_1(X_{\hat{\mathcal{S}}_m,\cdot,T},X_{j,\cdot,T})|\nonumber\\
    &+\sup_j|\hat{h}_2(\hat{R}_T^{m-1}f,X_{j,\cdot,T})|\sup_j|\hat{h}_1(X_{\hat{\mathcal{S}}_m,\cdot,T}(\cdot/T -u ),X_{j,\cdot,T})|
    \label{h1}
\end{align}

On the set $\mathcal{A}_T$ we have\\
$\sup_j|\hat{h}_1(X_{\hat{\mathcal{S}}_m,\cdot,T},X_{j,\cdot,T})|, \sup_j|\hat{h}_2(X_{\hat{\mathcal{S}}_m,\cdot,T},X_{j,\cdot,T})|,\sup_j|\hat{h}_2(X_{\hat{\mathcal{S}}_m,\cdot,T}(\cdot/T -u ),X_{j,\cdot,T})| \leq C+\zeta_t$ \\
and $\sup_j|\hat{h}_1(X_{\hat{\mathcal{S}}_m,\cdot,T}(\cdot/T -u ),X_{j,\cdot,T})| \leq O(S_T/T)$. Therefore (\ref{h2}),(\ref{h1}) define a simple bivariate recurrence relation which gives us a loose bound on $|\hat{h}_2(\hat{R}_T^mf,X_{j,\cdot,T})| \leq (C+1/2)^m$. Using this loose bound allows us to bound $III_{T,m}(j)$ and $(\ref{lemma4II})$, and the rest of the proof is straightforward.

\end{proof}

\begin{proof}[Proof of Lemma \ref{lemma5thm2}]
  $ $\newline
Note that: 
\begin{align*}
   \bm{Z}_{j,t,T}\bm{\hat{h}}(\hat{R}_T^{m}f,X_j) &= \inp[\big]{\tilde{R}^{m-1}f(\tilde{\bm{x}}_{T-h}(u))}{\tilde{X}_{j,T-h}(u)}\tilde{X}_{j,T-h}(u)\\
   &+ \bigg(\bm{Z}_{j,t,T}\bm{\hat{h}}(\hat{R}_T^{m}f,X_j)-\inp[\big]{\tilde{R}^{m-1}f(\tilde{\bm{x}}_{T-h}(u))}{\tilde{X}_{j,T-h}(u)}\tilde{X}_{j,T-h}(u)\bigg)
\end{align*}
From which we obtain, on the set $\mathcal{A}_T$, (as defined in lemma \ref{lemma4thm2}): 
\begin{align}
    ||\bm{Z}_{\hat{\mathcal{S}}_m,\cdot,T}\bm{\hat{h}}(\hat{R}_T^{m}f,X_{\hat{\mathcal{S}}_m})||_{(T)} &= \sup_j ||\bm{Z}_{j,\cdot,T}\bm{\hat{h}}(\hat{R}_T^{m}f,X_{j})||_{(T)}\nonumber \\
    & \geq \sup_j ||\inp[\big]{\tilde{R}^{m-1}f(\tilde{\bm{x}}_{T-h}(u))}{\tilde{X}_{j,T-h}(u)}\tilde{X}_{j,\cdot}(u)||_{(T)} - C^m(\zeta_T + S_T/T) \label{triangle_lemma8}
\end{align}

Using the same procedure we obtain, on the set $\mathcal{A}_T$:

\begin{align*}
    ||\inp[\big]{\tilde{R}^{m-1}f(\tilde{\bm{x}}_{T-h}(u))}{\tilde{X}_{\hat{\mathcal{S}}_m,T-h}(u)}&\tilde{X}_{\hat{\mathcal{S}}_m,T-h}(u)||_{(T)} \geq 
    ||\bm{Z}_{\hat{\mathcal{S}}_m,\cdot,T}\bm{\hat{h}}(\hat{R}_T^{m}f,X_{\hat{\mathcal{S}}_m})||_{(T)}  - C^m(\zeta_T + S_T/T) \\
    & \geq \sup_j ||\inp[\big]{\tilde{R}^{m-1}f(\tilde{\bm{x}}_{T-h}(u))}{\tilde{X}_{j,T-h}(u)}\tilde{X}_{j,\cdot}(u)||_{(T)} - 2C^m(\zeta_T + S_T/T)
\end{align*}
where the last inequality follows from (\ref{triangle_lemma8}). Without loss of generality we assume that $||\tilde{X}_{j,\cdot}(u)||_{(T)}=1$, given that we are assuming $||\tilde{X}_{j,t}(u)||=1$, and in light of lemma 1.

Define $\mathcal{B}_T=\{\omega: \sup_j |\inp[\big]{\tilde{R}^{m-1}f(\tilde{\bm{x}}_{T-h}(u))}{\tilde{X}_{j,T-h}(u)}| > 4C^m(\zeta_T + S_T/T)\}$. Then on the set $\mathcal{A}_T \cup \mathcal{B}_T$ we have:
\begin{align*}
    |\inp[\big]{\tilde{R}^{m-1}f(\tilde{\bm{x}}_{T-h}(u))}{\tilde{X}_{\hat{\mathcal{S}}_m,T-h}(u)}&| \geq .5 \sup_j |\inp[\big]{\tilde{R}^{m-1}f(\tilde{\bm{x}}_{T-h}(u))}{\tilde{X}_{j,T-h}(u)}|
\end{align*}

The rest of the proof is very similar to the proof in \cite{Buhlmann2006}, therefore we omit the details here.

\end{proof}

\section{Online Appendix B: Additional Empirical Results}

In this section, we include additional empirical results referred to in the main text. We start by presenting the relative MSFE for the full out of sample period and the three subperiods we studied for horizons $h=6,3,1$.  

We next plot 
\begin{align}
    MSFE^{h}_{(i)}(T_1,T_2)=\frac{\sum_{t=T_1}^{T_2}\hat{\epsilon}^{2}_{t,(i)}}{\sum_{t=T_1}^{T_2}\hat{\epsilon}^{2}_{t,(AR)}}, \quad \quad \text{with} \quad T_2=2018:8
\end{align}
for horizon $h=3$. 

Recall that: 
\begin{align*}
    \text{L-MSFE}_{i}(t_0)=\frac{\sum_{t=t_0-\Delta}^{t_0+\Delta}\hat{\epsilon}_{t,(i)}^{2}}{\sum_{t=t_0-70}^{t_0+\Delta}\hat{\epsilon}_{t,(AR)}^{2}}, \quad\quad \text{RL-MSFE}_{i_1,i_2}(t_0)=\frac{\text{L-MSFE}_{i_1}(t_0)}{\text{L-MSFE}_{i_2}(t_0)}
\end{align*}

We plot the $\text{RL-MSFE}_{i_1,i_2}$ for the following models:
\begin{enumerate}
    \item $i_1$=LC-Boost vs $i_2$=Boost
    \item $i_1=$LC-Boost vs $i_2=$LC-Boost Factor
    \item $i_1=$Boost Factor estimated using a 10 year rolling window vs $i_2=$LC-Boost Factor
    \item $i_1=$LL-Boost Factor vs $i_2$=LC-Boost Factor
\end{enumerate}

Lastly we include the simulation results using $t_5$ innovations.

\begin{table}[]
\small
\centering
\caption{Relative MSFE $h=6$}
\label{forecast6}
\begin{tabular}{|c|c|c|c|c|c|c|c|c|}
\hline
\multicolumn{9}{|c|}{Full Out of Sample Period 1971:9-2018:8}                                             \\ \hline
                & IP  & PAYEMS & UNRATE & CLF & RPI & CPI & FF  & TB3MS \\ \hline
TV-AR           & 1.1 & 1.01   & 1.08   & .77 & 1   & .83 & 1   & .99   \\ \hline
DI            & .81 & .81    & .73    & .98 & .86 & .96 & .91 & .92   \\ \hline
Lasso           & .82 & .78    & .72    & .85 & .91 & .90 & .80 & .93   \\ \hline
Boost           & .79 & .77    & .73    & .79 & .90 & .87 & .86 & .93   \\ \hline
Boost Factor    & .79 & .83    & .69    & .92 & .81 & .89 & .79 & .87   \\ \hline
LC-Boost        & .76 & .76    & .69    & .79 & .89 & .89 & .94 & .98   \\ \hline
LC-Boost Factor & .73 & .74    & .67    & .73 & .80 & .80 & .78 & .90   \\ \hline
LL-Boost Factor & .88 & .95    & .84    & .79 & .79 & .79 & .97 & 1.37  \\ \hline
\end{tabular}
\begin{tabular}{|c|c|c|c|c|c|c|c|c|}
\hline
\multicolumn{9}{|c|}{``Pre-Great Moderation" 1971:9-1982:12}                                                 \\ \hline
                & IP   & PAYEMS & UNRATE & CLF  & RPI & CPI  & FF   & TB3MS \\ \hline
TV-AR           & 1.08 & 1      & 1.12   & .88  & .97 & 1.07 & 1.06 & 1.05  \\ \hline
DI            & .56  & .64    & .61    & 1.40 & .76 & .88  & .91  & .88   \\ \hline
Lasso           & .55  & .62    & .56    & 1.20 & .80 & 1.03 & .77  & .93   \\ \hline
Boost           & .55  & .64    & .63    & 1.20 & .78 & .95  & .79  & .83   \\ \hline
Boost Factor    & .59  & .71    & .64    & 1.14 & .75 & .88  & .76  & .86   \\ \hline
LC-Boost        & .55  & .64    & .59    & .93  & .98 & 1.12 & .91  & .90   \\ \hline
LC-Boost Factor & .60  & .76    & .66    & .85  & .82 & 1.02 & .78  & .94   \\ \hline
LL-Boost Factor & .73  & .99    & .88    & .88  & .80 & .94  & .95  & 1.01  \\ \hline
\end{tabular}
\begin{tabular}{|c|c|c|c|c|c|c|c|c|}
\hline
\multicolumn{9}{|c|}{``Great Moderation" 1983:1-2006:12}                                                            \\ \hline
                & IP   & PAYEMS & UNRATE & CLF & RPI  & CPI & FF   & TB3MS \\ \hline
TV-AR           & 1.08 & .98    & 1.04   & .81 & 1.02 & .93 & .86  & .88   \\ \hline
DI            & 1.07 & 1.03   & .82    & .89 & .96  & .96 & .86  & .96   \\ \hline
Lasso           & 1.06 & 1.11   & .88    & .90 & 1.03 & .86 & .87  & .93   \\ \hline
Boost           & 1.08 & 1.09   & .85    & .80 & 1    & .86 & 1.01 & 1.08  \\ \hline
Boost Factor    & 1.04 & .99    & .70    & .87 & .93  & .83 & .82  & .85   \\ \hline
LC-Boost        & 1.14 & .99    & .85    & .89 & .95  & .74 & 1.04 & 1.1   \\ \hline
LC-Boost Factor & .98  & .80    & .71    & .80 & .85  & .83 & .80  & .76   \\ \hline
LL-Boost Factor & 1.04 & .98    & .84    & .92 & .97  & .88 & 1.06 & 1.05  \\ \hline
\end{tabular}
\begin{tabular}{|c|c|c|c|c|c|c|c|c|}
\hline
\multicolumn{9}{|c|}{``Post Great Moderation" 2007:1-2018:8}                                \\ \hline
                & IP   & PAYEMS & UNRATE & CLF & RPI & CPI  & FF   & TB3MS \\ \hline
TV-AR           & 1.18 & 1.09   & 1.06   & .64 & .98 & .59  & .72  & .76   \\ \hline
DI            & 1.12 & 1      & .84    & .78 & .83 & 1.01 & 1.21 & 1.16  \\ \hline
Lasso           & 1.21 & .86    & .86    & .52 & .91 & .82  & 1    & 1.02  \\ \hline
Boost           & 1.05 & .77    & .79    & .45 & .77 & .82  & 1.53 & 1.43  \\ \hline
Boost Factor    & 1.01 & .98    & .78    & .79 & .76 & .94  & 1.37 & 1.12  \\ \hline
LC-Boost        & .89  & .80    & .70    & .54 & .77 & .58  & .93  & 1.32  \\ \hline
LC-Boost Factor & .82  & .61    & .67    & .54 & .75 & .64  & .69  & .95   \\ \hline
LL-Boost Factor & 1.07 & .80    & .74    & .53 & .62 & .62  & .75  & .92   \\ \hline
\end{tabular}
\end{table}

\begin{table}[]
\small
\centering
\caption{Relative MSFE $h=3$}
\label{forecast3}
\begin{tabular}{|c|c|c|c|c|c|c|c|c|}
\hline
\multicolumn{9}{|c|}{Full Out of Sample Period 1971:9-2018:8}                   
\\ \hline
                & IP   & PAYEMS & UNRATE & CLF  & RPI  & CPI & FF  & TB3MS \\ \hline
TV-AR           & 1.04 & 1.01   & 1.01   & .86  & 1.01 & .91 & .98 & .98   \\ \hline
DI            & .84  & .81    & .78    & .98  & .85  & .96 & .90 & .94   \\ \hline
Lasso           & .89  & .83    & .82    & 1.03 & .87  & .97 & .84 & .96   \\ \hline
Boost           & .89  & .78    & .75    & .89  & .86  & .87 & .93 & 1.04  \\ \hline
Boost Factor    & .80  & .83    & .76    & .96  & .81  & .90 & .82 & .89   \\ \hline
LC-Boost        & .79  & .77    & .73    & .87  & .85  & .85 & .88 & 1.02  \\ \hline
LC-Boost Factor & .77  & .78    & .77    & .87  & .80  & .87 & .79 & .89   \\ \hline
LL-Boost Factor & .84  & .82    & .82    & .90  & .81  & .87 & .92 & 1.05  \\ \hline
\end{tabular}
\begin{tabular}{|c|c|c|c|c|c|c|c|c|}
\hline
\multicolumn{9}{|c|}{``Pre-Great Moderation" 1971:9-1982:12}                                                 \\ \hline
                & IP   & PAYEMS & UNRATE & CLF  & RPI  & CPI  & FF  & TB3MS \\ \hline
TV-AR           & 1.05 & 1      & 1      & .97  & 1.01 & 1.01 & 1   & 1.01  \\ \hline
DI            & .72  & .72    & .72    & 1.29 & .87  & .87  & .89 & .88   \\ \hline
Lasso           & .73  & .74    & .74    & 1.44 & .85  & 1.09 & .84 & .98   \\ \hline
Boost           & .77  & .89    & .74    & 1.14 & .95  & .91  & .92 & 1     \\ \hline
Boost Factor    & .71  & .78    & .78    & 1.15 & .80  & .83  & .80 & .86   \\ \hline
LC-Boost        & .78  & .69    & .75    & 1.02 & .95  & .88  & .89 & 1.02  \\ \hline
LC-Boost Factor & .74  & .82    & .84    & 1.03 & .85  & .97  & .78 & .89   \\ \hline
LL-Boost Factor & .84  & .93    & .90    & 1.05 & .94  & .87  & .87 & 1     \\ \hline
\end{tabular}
\begin{tabular}{|c|c|c|c|c|c|c|c|c|}
\hline
\multicolumn{9}{|c|}{``Great Moderation" 1983:1-2006:12}                                                           \\ \hline
                & IP   & PAYEMS & UNRATE & CLF & RPI & CPI & FF   & TB3MS \\ \hline
TV-AR           & 1.06 & 1.06   & 1.03   & .90 & 1   & .83 & .85  & .88   \\ \hline
DI            & .92  & .95    & .86    & .91 & .87 & .97 & .88  & 1.16  \\ \hline
Lasso           & .96  & 1.04   & .91    & .98 & .92 & .90 & .78  & .88   \\ \hline
Boost           & .97  & .96    & .84    & .91 & .90 & .90 & .89  & 1.12  \\ \hline
Boost Factor    & .89  & .89    & .76    & .94 & .87 & .89 & .86  & .95   \\ \hline
LC-Boost        & .94  & .96    & .81    & .94 & .88 & .83 & .86  & 1     \\ \hline
LC-Boost Factor & .86  & .78    & .79    & .92 & .84 & .83 & .84  & .89   \\ \hline
LL-Boost Factor & .94  & .85    & .85    & .97 & .90 & .86 & 1.23 & 1.28  \\ \hline
\end{tabular}
\begin{tabular}{|c|c|c|c|c|c|c|c|c|}
\hline
\multicolumn{9}{|c|}{``Post Great Moderation" 2007:1-2018:8}                                               \\ \hline
                & IP   & PAYEMS & UNRATE & CLF & RPI & CPI & FF   & TB3MS \\ \hline
TV-AR           & 1    & .98    & 1.08   & .70 & 1   & .92 & .74  & .86   \\ \hline
DI            & 1.03 & .95    & .78    & .83 & .82 & 1   & 1.37 & 1.28  \\ \hline
Lasso           & 1.19 & .81    & .87    & .75 & .83 & .83 & .79  & .97   \\ \hline
Boost           & 1.07 & .84    & .68    & .66 & .79 & .84 & 1.36 & 1.41  \\ \hline
Boost Factor    & .92  & .90    & .70    & .83 & .76 & .95 & 1.45 & 1.20  \\ \hline
LC-Boost        & .66  & .74    & .60    & .65 & .77 & .83 & .65  & 1.06  \\ \hline
LC-Boost Factor & .73  & .66    & .64    & .68 & .75 & .84 & .78  & .93   \\ \hline
LL-Boost Factor & .76  & .67    & .63    & .68 & .68 & .89 & .87  & 1.05  \\ \hline
\end{tabular}
\end{table}

\begin{table}[]
\small
\centering
\caption{Relative MSFE $h=1$}
\label{forecast1}
\begin{tabular}{|c|c|c|c|c|c|c|c|c|}
\hline
\multicolumn{9}{|c|}{Full Out of Sample Period 1971:9-2018:8}                   
\\ \hline
                & IP   & PAYEMS & UNRATE & CLF  & RPI  & CPI & FF  & TB3MS \\ \hline
TV-AR           & 1.04 & 1.04   & 1.01   & .96  & 1.06 & .96 & .99 & 1.02  \\ \hline
DI            & .92  & .87    & .87    & 1.04 & .92  & .93 & .92 & .95   \\ \hline
Lasso           & 1.06 & .92    & .94    & 1.26 & .94  & 1   & .88 & .98   \\ \hline
Boost           & .91  & .88    & .82    & 1.03 & .98  & .94 & .86 & 1     \\ \hline
Boost Factor    & .86  & .91    & .84    & 1.02 & .90  & .90 & .85 & .88   \\ \hline
LC-Boost        & .94  & .91    & .85    & 1.1  & .96  & .95 & .84 & .99   \\ \hline
LC-Boost Factor & .85  & .87    & .85    & 1.02 & .89  & .91 & .84 & .84   \\ \hline
LL-Boost Factor & .92  & .95    & .88    & 1.04 & .98  & .93 & .91 & .88   \\ \hline
\end{tabular}
\begin{tabular}{|c|c|c|c|c|c|c|c|c|}
\hline
\multicolumn{9}{|c|}{``Pre-Great Moderation" 1971:9-1982:12}                                                 \\ \hline
                & IP   & PAYEMS & UNRATE & CLF  & RPI  & CPI  & FF  & TB3MS \\ \hline
TV-AR           & 1.08 & 1.02   & 1.06   & .99  & 1.03 & .99  & 1   & 1.03  \\ \hline
DI            & .86  & .76    & .75    & 1.18 & 1.01 & .97  & .90 & .94   \\ \hline
Lasso           & 1.21 & .83    & 1.05   & 1.95 & 1.15 & 1.22 & .89 & .97   \\ \hline
Boost           & .90  & .78    & .77    & 1.19 & 1.12 & 1.10 & .84 & .97   \\ \hline
Boost Factor    & .80  & .86    & .80    & 1.15 & .97  & .91  & .82 & .81   \\ \hline
LC-Boost        & .96  & .82    & .77    & 1.30 & 1.08 & 1.14 & .82 & .96   \\ \hline
LC-Boost Factor & .81  & .86    & .79    & 1.19 & .96  & 1.02 & .82 & .78   \\ \hline
LL-Boost Factor & .95  & .96    & .87    & 1.20 & 1.12 & 1.04 & .85 & .75   \\ \hline
\end{tabular}
\begin{tabular}{|c|c|c|c|c|c|c|c|c|}
\hline
\multicolumn{9}{|c|}{``Great Moderation" 1983:1-2006:12}                                                            \\ \hline
                & IP  & PAYEMS & UNRATE & CLF  & RPI  & CPI & FF   & TB3MS \\ \hline
TV-AR           & 1   & 1.1    & .99    & .98  & 1.08 & .92 & .92  & .97   \\ \hline
DI            & .94 & .96    & .96    & .98  & .93  & .89 & 1.01 & .98   \\ \hline
Lasso           & .91 & 1      & .93    & .98  & .92  & .94 & .82  & 1.11  \\ \hline
Boost           & .88 & .98    & .90    & 1    & .94  & .89 & 1.02 & 1.25  \\ \hline
Boost Factor    & .91 & .97    & .91    & .99  & .91  & .93 & 1.03 & 1.31  \\ \hline
LC-Boost        & .91 & 1.04   & .93    & .97  & .90  & .79 & .98  & 1.24  \\ \hline
LC-Boost Factor & .91 & .95    & .94    & .97  & .88  & .84 & 1    & 1.30  \\ \hline
LL-Boost Factor & .90 & .99    & .94    & 1.01 & .96  & .84 & 1.50 & 1.84  \\ \hline
\end{tabular}
\begin{tabular}{|c|c|c|c|c|c|c|c|c|}
\hline
\multicolumn{9}{|c|}{``Post Great Moderation" 2007:1-2018:8}                                                \\ \hline
                & IP   & PAYEMS & UNRATE & CLF & RPI  & CPI & FF   & TB3MS \\ \hline
TV-AR           & 1.04 & .94    & .99    & .87 & 1.05 & .97 & .84  & .96   \\ \hline
DI            & .98  & 1.05   & .90    & .95 & .89  & .95 & 1.67 & 1.33  \\ \hline
Lasso           & 1.05 & 1.02   & .81    & .93 & .90  & .88 & .77  & .96   \\ \hline
Boost           & .94  & .93    & .75    & .92 & 1.01 & .85 & 1.35 & 1.24  \\ \hline
Boost Factor    & .87  & .95    & .79    & .94 & .87  & .86 & 1.46 & 1.35  \\ \hline
LC-Boost        & .96  & .83    & .79    & 1.1 & 1    & .95 & .89  & .94   \\ \hline
LC-Boost Factor & .84  & .71    & .79    & .89 & .87  & .89 & .94  & .96   \\ \hline
LL-Boost Factor & .88  & .75    & .78    & .88 & .98  & .94 & .85  & 1.06  \\ \hline
\end{tabular}
\end{table}

\begin{table}[]
\centering
\small
\caption{$\text{DGP  1-14}:$ Relative MSFE, $t_5$ Innovations}
\label{tdist}
\begin{tabular}{|c|c|c|c|c|c|c|}
\hline
DGP & AR (3) & Rolling AR (3) & Rolling Boost & LC-Boost & LL-Boost & Lasso \\ \hline
1   & 1.67   & 1.84           & 1.80          & 1.06     & 1.21     & 1.03  \\ \hline
2   & 1.24   & 1.26           & 1.26          & 1.12     & 1.12     & 1.18  \\ \hline
3   & 1.13   & 1.22           & 1.10          & .76      & .77      & 1.09  \\ \hline
4   & .96    & 1.04           & .94           & .72      & .68      & 1.02  \\ \hline
5   & .77    & .83            & .76           & .87      & .68      & .92   \\ \hline
6   & 4.23   & 4.78           & 1.85          & 1.00     & 1.03     & 1.05  \\ \hline
7   & 3.58   & 3.85           & 1.33          & .90      & .83      & 1.09  \\ \hline
8   & 3.73   & 4.19           & 1.92          & .72      & .70      & 1.11  \\ \hline
9   & 1.81   & 1.97           & 1.03          & .71      & .49      & 1.18  \\ \hline
10  & 1.67   & 1.82           & 1.24          & .92      & .68      & 1.20  \\ \hline
11  & 2.20   & 2.37           & 1.18          & .85      & .75      & 1.10  \\ \hline
12  & 1.08   & 1.21           & .42           & .76      & .21      & 1.13  \\ \hline
13  & 1.99   & 2.05           & 1.06          & .75      & .54      & 1.21  \\ \hline
14  & .81    & .89            & .85           & .92      & .71      & 1.02  \\ \hline
\end{tabular}
\end{table}

\begin{figure}
    \centering
    \includegraphics[scale=.75]{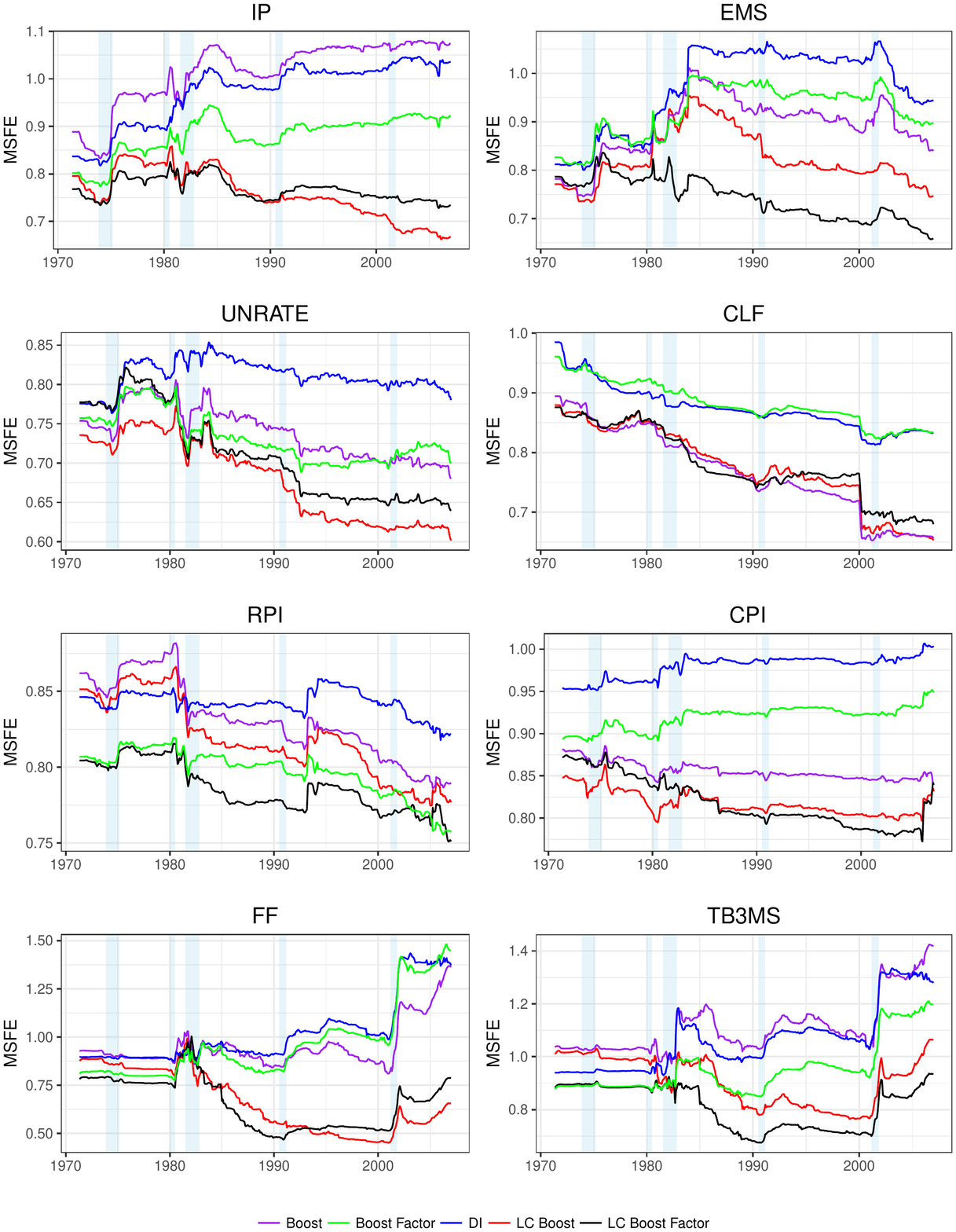}
    \caption{\footnotesize{\textbf{MSFE by start date of out of sample period. Horizon $h=3$}. More specifically we plot: $
    MSFE^{3}_{(i)}(T_1,T_2)=\sum_{t=T_1}^{T_2}\hat{\epsilon}^{2}_{t,(i)}/\sum_{t=T_1}^{T_2}\hat{\epsilon}^{2}_{t,(AR)}$, where we $T_1$ vary from 1971:9 until 2006:12,  with $T_2$=2018:8. Shaded regions represent NBER recession dates.}}
    \label{MSFE3}
\end{figure}

\begin{figure}
    \centering
    \includegraphics[scale=.75]{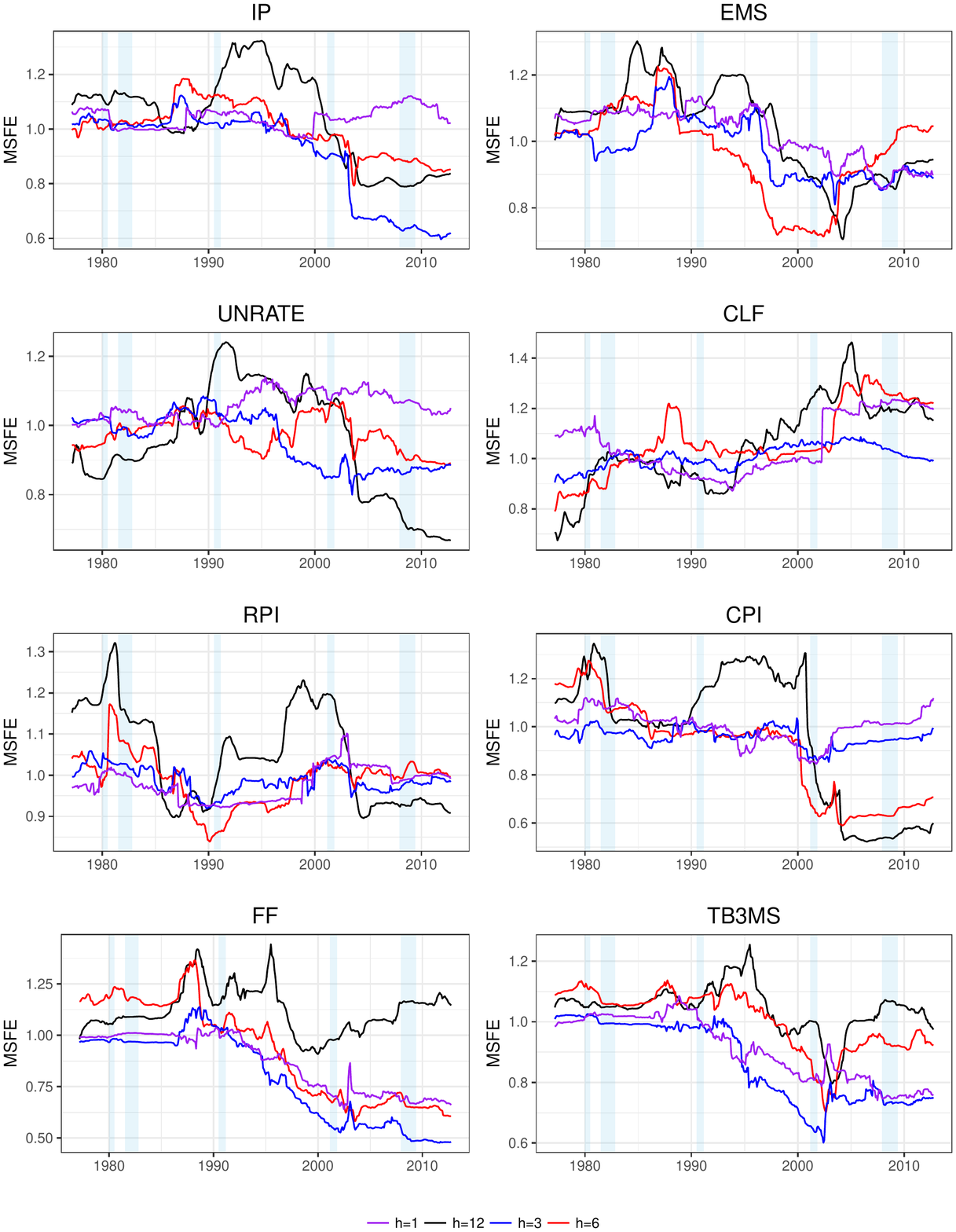}
    \caption{\footnotesize{\textbf{Local MSFE of LC-Boost relative to Local MSFE of Boost}: To obtain the local MSFE of a model we use a rolling mean with a window size of 70 observations, which gives us: $
    \text{L-MSFE}_{i}(t_0)=\sum_{t=t_0-70}^{t_0+70}\hat{\epsilon}_{t,(i)}^{2}/\sum_{t=t_0-70}^{t_0+70}\hat{\epsilon}_{t,(AR)}^{2}$
with the convention that $\hat{\epsilon}_{t,(i)}=0$ for $t\leq 0, t\geq T$. We then plot $\text{L-MSFE}_{LC Boost}(t)/\text{L-MSFE}_{Boost}(t)$, for $t=\text{1977:3},\ldots,\text{2012:10}$}}
    \label{localMSFE3}
\end{figure}

\begin{figure}
    \centering
    \includegraphics[scale=.75]{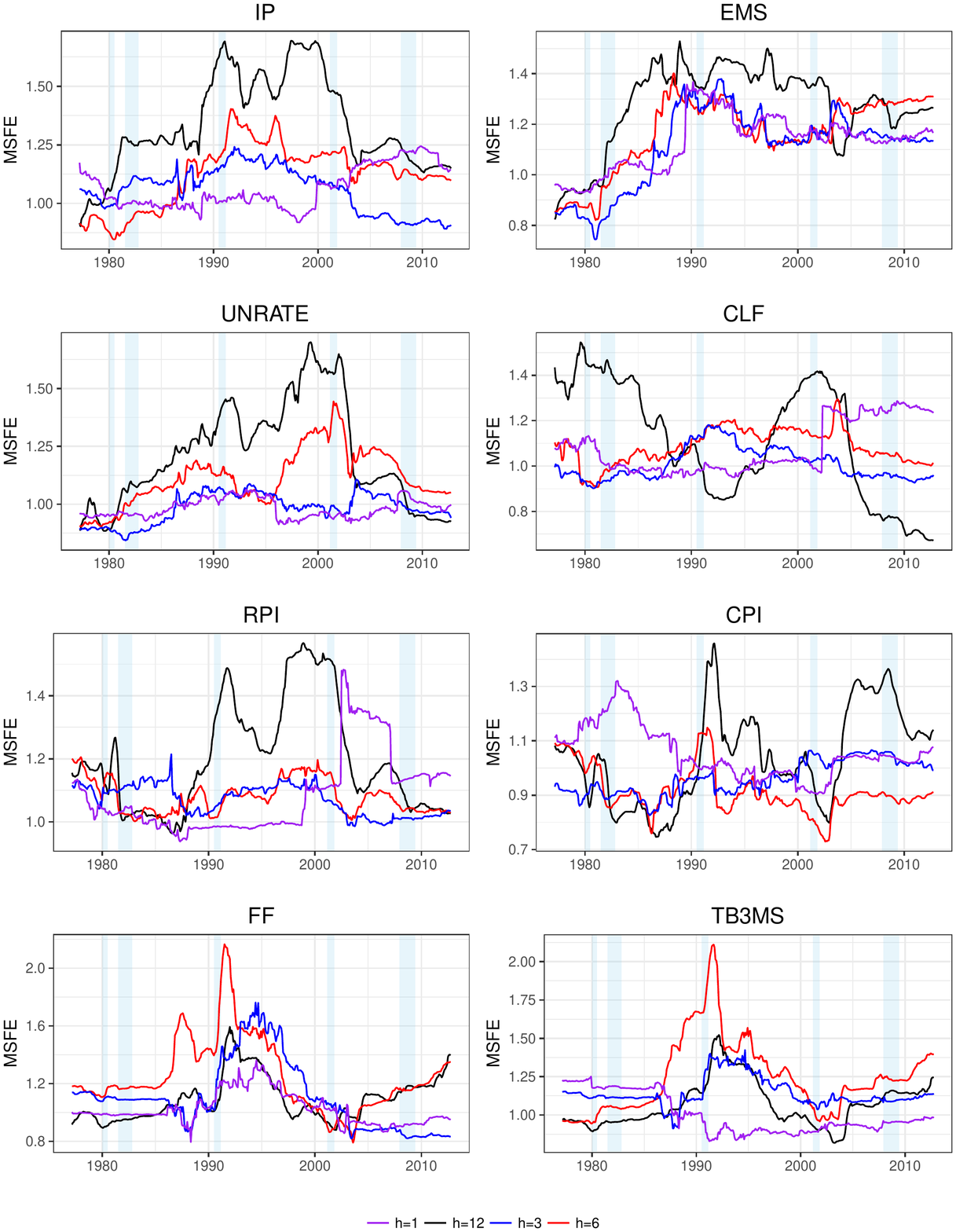}
    \caption{\footnotesize{\textbf{L-MSFE of LC-Boost relative to L-MSFE of LC-Boost Factor}: This figure uses a window size of 70 observations to calculate the L-MSFE, see notes to figure \ref{localMSFE3}}}
\end{figure}

\begin{figure}
    \centering
    \includegraphics[scale=.75]{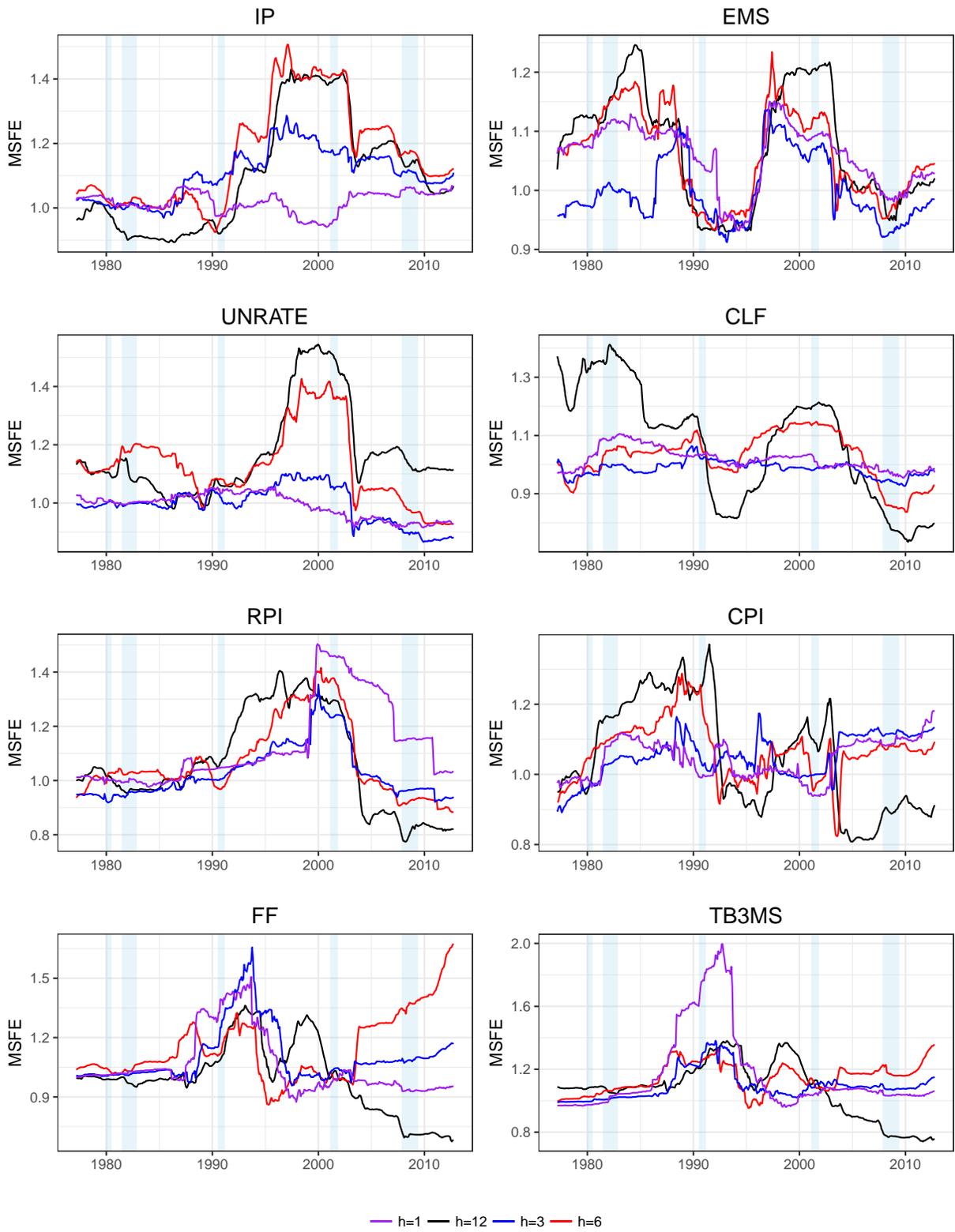}
    \caption{\footnotesize{\textbf{L-MSFE of Boost Factor using 10 year rolling window relative to L-MSFE of LC-Boost Factor}: This figure uses a window size of 70 observations to calculate the L-MSFE, see notes to figure \ref{localMSFE3}}}

\end{figure}

\begin{figure}
    \centering
    \includegraphics[scale=.75]{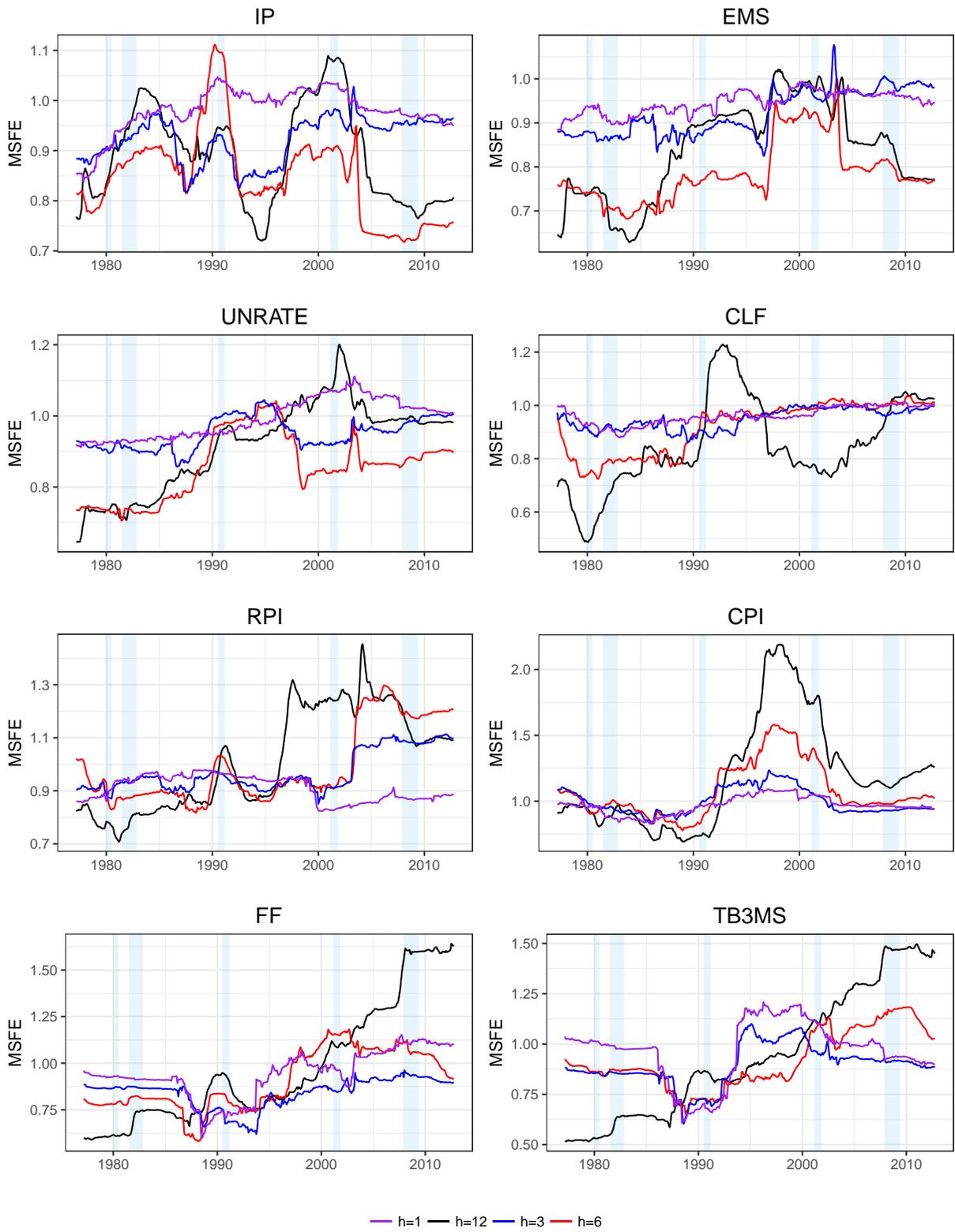}
    \caption{\footnotesize{\textbf{L-MSFE of LC-Boost Factor relative to L-MSFE of LL-Boost Factor}: We use a window size of $70$ observations, see notes to figure \ref{localMSFE3} for details. Colored lines represent the different horizons.}}
    \label{localMSFE2}
\end{figure}

\clearpage

\end{document}